\documentclass[reprint,amsfonts, amssymb, amsmath,  showkeys,prx, superscriptaddress, twocolumn,longbibliography,nofootinbib,dvipsnames]{revtex4-2}

\usepackage{enumerate}
\usepackage{amsmath}
\usepackage{graphicx}
\usepackage{comment}

\usepackage{tikz} 
\usepackage{etoolbox}

\usepackage{dcolumn}
\usepackage{bm,bbm}
\usepackage{physics}
\usepackage{amsmath}
\usepackage{amssymb}
\usepackage{amsthm}
\usepackage{amsfonts}
\usepackage{MnSymbol}
\usepackage[colorlinks=true,citecolor=blue,linkcolor=magenta]{hyperref}
\usepackage{tikz}

\graphicspath{{Images/}}

\usepackage[export]{adjustbox}

\usepackage{listings}
\newcommand{\code}[1]{\texttt{#1}}

\usepackage{algorithm}
\usepackage{algorithmicx}
\usepackage{algpseudocode}

\newtheorem{theorem}{Theorem}
 
\newtheorem{lemma}{Lemma}

\newtheorem{proposition}{Proposition}

\newtheorem{example}{Example}
\newtheorem{result}{Result}

\definecolor{keywordblue}{RGB}{0,0,180}
\lstdefinelanguage{Pseudocode}{
  keywords={for, remove, insert, apply, truncate, overlap
  },
  sensitive=false,
  comment=[l]{\#},
  keywordstyle=\color{keywordblue}\bfseries,
  commentstyle=\color{gray}\ttfamily\itshape
}

\newcommand{\Ebb}{\mathbb{E}}

\newcommand{\EC}{\mathcal{E}}

\newcommand{\FC}{\mathcal{F}}

\newcommand{\PC}{\mathcal{P}}

\newcommand{\RC}{\mathcal{R}}

\newcommand{\Var}{{\rm Var}}
\newcommand{\MSE}{{\rm MSE}}
\newcommand{\STD}{{\rm STD}}
\newcommand{\Covar}{{\rm Covar}}

\renewcommand{\geq}{\geqslant}
\renewcommand{\leq}{\leqslant}

\renewcommand{\vec}[1]{\boldsymbol{#1}}

\newtheorem{protocol}{Protocol}

\newcommand{\thv}{\vec{\theta}}

\graphicspath{{Images/}}

\newcommand{\be}{\begin{equation}}
\newcommand{\ee}{\end{equation}}

\renewcommand{\vec}[1]{\boldsymbol{#1}}

\usepackage{mathtools}

\usepackage{amssymb}

\newcommand{\minrelcoeff}{\ensuremath{\epsilon_{\text{rel}}}}

\newcommand{\maxsize}{\ensuremath{M}}

\newcommand{\ppjl}{\code{\mbox{PauliPropagation.jl}}}
\newcommand{\ppjllink}{\href{https://github.com/MSRudolph/PauliPropagation.jl}{\mbox{\ppjl}}}

\usepackage[normalem]{ulem}

\usepackage{cancel}

\newcommand{\sign}[1]{\text{sign}(#1)}

\makeatletter
\let\sign\undefined
\def\sign{\operatorname{sgn}\@ifnextchar\bgroup{\sign@with}{\relax}}
\def\sign@with#1{\left(#1\right)}
\makeatother

\makeatletter
\renewcommand\onecolumngrid{
\do@columngrid{one}{\@ne}
\def\set@footnotewidth{\onecolumngrid}
\def\footnoterule{\kern-6pt\hrule width 1.5in\kern6pt}
}
\renewcommand\twocolumngrid{
        \def\footnoterule{
        \dimen@\skip\footins\divide\dimen@\thr@@
        \kern-\dimen@\hrule width.5in\kern\dimen@}
        \do@columngrid{mlt}{\tw@}
}
\makeatother

\usepackage[framemethod=TikZ]{mdframed}
\usepackage{tcolorbox}
\usepackage{xcolor} 

\usepackage{color, colortbl}

\definecolor{lightgray}{gray}{0.9}

\definecolor{theorycolour}{RGB}{148, 156, 168} 
\definecolor{softwarecolour}{RGB}{50, 136, 140}
\definecolor{datacolour}{RGB}{101, 143, 226} 
\definecolor{conventioncolour}{RGB}{194, 97, 186}
\definecolor{glossarycolour}{RGB}{255, 205, 105}

\newtcolorbox[auto counter]{pbbox}[2][]{fonttitle=\bfseries,
title=Theory Box~\thetcbcounter: #2,#1,colframe=gray}

\newtcolorbox[use counter from=pbbox]{theorybox}[2][]{
floatplacement=t,float,
colback=theorycolour!5!white,colframe=theorycolour!75!black,title= Theory Box~\thetcbcounter: #2,#1}

\newtcolorbox[auto counter]{pabox}[2][]{fonttitle=\bfseries,
title=Software Box~\thetcbcounter: #2,#1,colframe=gray}

\newtcolorbox[use counter from=pabox]{softwarebox}[2][]{
floatplacement=t,float,
colback=softwarecolour!5!white,colframe=softwarecolour!75!black,title= Software Box~\thetcbcounter: #2,#1}

\newtcolorbox[auto counter]{glossarybox}[2][]{
floatplacement=t,float,
colback=glossarycolour!5!white,colframe=glossarycolour!90!black,title= Glossary Box~\thetcbcounter: #2,#1}

\newtcolorbox[use counter from=pbbox]{theoryboxlarge}[2][]{
  float*=t, % Use starred float option to span full width in two-column layouts
  colback=theorycolour!5!white,
  colframe=theorycolour!75!black,
  title= Theory Box~\thetcbcounter: #2,#1,
  width=\textwidth, % force the content to use the full text width
}

\newtcolorbox[auto counter]{pcbox}[2][]{fonttitle=\bfseries,
title=Software Box~\thetcbcounter: #2,#1,colframe=gray}

\newtcolorbox[use counter from=pcbox]{datastructbox}[2][]{
floatplacement=t,
float, 
colback=datacolour!5!white,colframe=datacolour!75!black,title= Under the Hood Box~\thetcbcounter: #2,#1}

\newtcolorbox[use counter from=pcbox]{datastructboxfull}[2][]{
width=\textwidth, 
colback=datacolour!5!white,colframe=datacolour!75!black,title= Under the Hood Box~\thetcbcounter: #2,#1}

\newtcolorbox[auto counter]{conventionbox}[2][]{
floatplacement=t,float,
colback=conventioncolour!5!white,colframe=conventioncolour!85!black,title= Convention Box~\thetcbcounter: #2,#1}

\usepackage[normalem]{ulem}

\begin{document}

\title{Pushing beyond the memory barrier in propagation algorithms with randomization}

\date{\today}

\author{Manuel S. Rudolph}
% \email{manuel.rudolph@epfl.ch}

\affiliation{Institute of Physics, Ecole Polytechnique F\'{e}d\'{e}rale de Lausanne (EPFL),   Lausanne, Switzerland}
\affiliation{Centre for Quantum Science and Engineering, Ecole Polytechnique F\'{e}d\'{e}rale de Lausanne (EPFL),   Lausanne, Switzerland}

\author{Ricard Puig}
\affiliation{Institute of Physics, Ecole Polytechnique F\'{e}d\'{e}rale de Lausanne (EPFL),   Lausanne, Switzerland}
\affiliation{Centre for Quantum Science and Engineering, Ecole Polytechnique F\'{e}d\'{e}rale de Lausanne (EPFL),   Lausanne, Switzerland}

\author{Slimane Thabet}
\affiliation{Institute of Physics, Ecole Polytechnique F\'{e}d\'{e}rale de Lausanne (EPFL),   Lausanne, Switzerland}
\affiliation{Centre for Quantum Science and Engineering, Ecole Polytechnique F\'{e}d\'{e}rale de Lausanne (EPFL),   Lausanne, Switzerland}

\author{Zo\"{e} Holmes}
\affiliation{Institute of Physics, Ecole Polytechnique F\'{e}d\'{e}rale de Lausanne (EPFL),   Lausanne, Switzerland}
\affiliation{Centre for Quantum Science and Engineering, Ecole Polytechnique F\'{e}d\'{e}rale de Lausanne (EPFL),   Lausanne, Switzerland}

\begin{abstract}
Performant classical tools for simulating quantum circuits are
important for benchmarking and supporting upcoming quantum devices. Pauli
propagation offers promising capabilities, but deep circuits usually cause a rapid
growth in the number of Pauli strings and can quickly exhaust the available memory. 
In this work, we introduce a randomized protocol that combines a specialized resampling schedule, merging, and a controlled amount of bias to keep the propagating Pauli sum within any prescribed memory budget. Additional computational resources can then be spent on embarrassingly parallel repetitions that systematically improve the estimate.
Building
on the open-source library \ppjl{}, we show that 1) increasing memory of each simulation reduces the error systematically more than outside averaging, 2) resampling just enough to apply the next gate delays variance accumulation and gives better error-per-time
trade-offs than more aggressive resampling, and 3) a small, controlled bias is crucial in cutting down variance and pushing far beyond unbiased randomization. The resulting framework can reproduce deterministic calculations while requiring several orders of magnitude less memory.
\end{abstract}

\maketitle

\paragraph*{\textbf{Introduction.}}

\begin{figure}[t]
    \centering
    \includegraphics[width=\columnwidth]{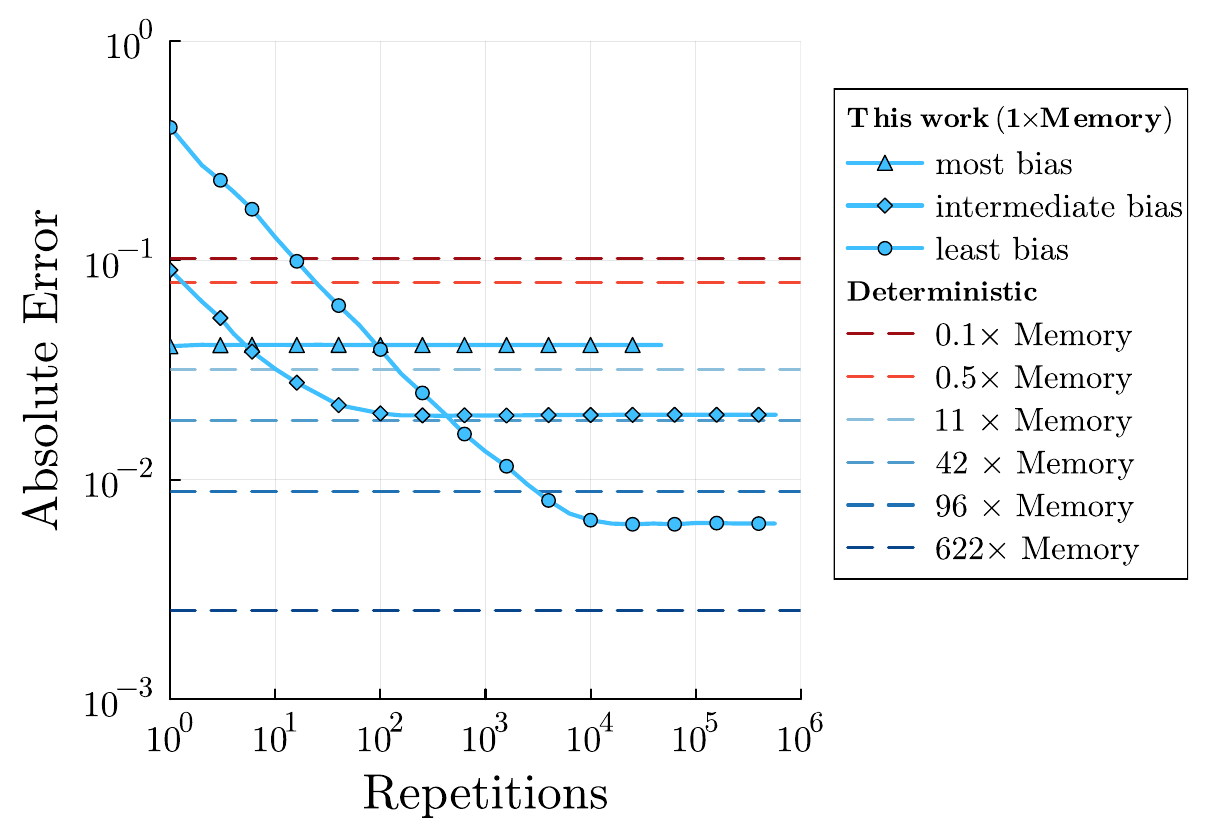}
    \caption{\textbf{Randomized propagation reproduces deterministic calculations that use up to two orders of magnitude more memory.}
    The
    horizontal dashed lines show deterministic calculations and report their
    final memory usage relative to the memory used in our randomized simulation. A carefully chosen bias allows the
    randomized calculation to reach the accuracy of deterministic calculations
    requiring one or two orders of magnitude more memory.}
    \vspace{-4mm}
    \label{fig:bias_assisted_propagation}
\end{figure}

Classical simulation is both a competitor to and an enabler of quantum computing. It provides a baseline for benchmarking devices and determining when quantum hardware offers genuine advantages~\cite{boixo2018characterizing, arute2019quantum, zhong2020quantum, liu2021closing, wu2021strong, kim2023evidence,tindall2023efficient,beguvsic2023fast,anand2023classical,rudolph2023classical, morvan2023phase, bermejo2024quantum, bluvstein2024logical, maslov2024fast, king2025beyond, mauron2025challenging, rudolph2025simulating, tindall2025dynamics, alam2025fermionic, google2025observation, alam2025programmable, haghshenas2026digital, chang2026cost, hartnett2026fast, alam2026onset, barron2026observable,andersen2026precision,andersen2025thermalization}. It also allows algorithms to be developed without spending every iteration on scarce hardware, and contributes directly to compilation~\cite{ravi2022cafqa, rudolph2022synergy, mc2023classically,tepaske2023optimal,causer2023scalable,gibbs2024deep, gibbs2026low, chai2025resource, danna2025circuit, miller2025simulation, chakraborty2026scalable, mansuroglu2025preparation, shrikhande2025rapid}, error-mitigation~\cite{fuller2025improved,suchsland2026quantum,majumder2026quantum,czarnik2020error,strikis2020learning} and error-correction~\cite{trout2018simulating, gidney2021stim, bravyi2024high, haenel2026tsim} pipelines. Performant simulators are thus part of the machinery needed to make quantum hardware useful.

Classical resources, however, are not equally negotiable. A calculation that needs another day, another thousand node-hours, or another line on the computing bill may be painful but remains possible. A calculation that needs twice the RAM of the largest available node simply does not fit. Continually increasing the memory of each simulation introduces communication overheads that may eat up a significant amount of the hoped benefits~\cite{broers2025scalable}. This distinction is particularly sharp for modern GPUs, whose prodigious throughput is often paired with substantially less memory than a large CPU node. Time and money are (potentially) negotiable. Memory is not. 

Pauli propagation~\cite{rall2019simulation,rudolph2025pauli,fontana2023classical,rudolph2023classical,beguvsic2023simulating,beguvsic2023fast,shao2023simulating,angrisani2024classically,lerch2024efficient,beguvsic2024real,angrisani2025simulating,fuller2025improved,teng2025leveraging,rudolph2026thermal,broers2024exclusive,broers2025scalable, angrisani2026dissipation} encounters this wall particularly sharply. The method
typically evolves the Pauli basis expansion of observables by applying a circuit gate by gate. Non-Clifford gates often cause individual Pauli strings to branch, while truncation and merging of duplicate Pauli strings attempt to keep their number manageable~\cite{rudolph2025pauli,beguvsic2024real,aharonov2023polynomial,fontana2023classical,nemkov2023fourier,beguvsic2023simulating,schuster2024polynomial,angrisani2024classically,lerch2024efficient,miller2025simulation,teng2025leveraging}. Once the Pauli sum exhausts the available RAM, faster code often means only that the calculation reaches the same limit sooner.

In this work, we present a framework that uses memory as a hyperparameter rather than a barrier. Given a hard cap on the number of stored Pauli strings, our proposal is to perform state-of-the-art deterministic propagation while the evolving Pauli sum fits, resample stochastically whenever the next update would exceed the cap, and continue propagating. This turns the Pauli sum into an estimator that can continually remain within the chosen memory budget. The predicted expectation value can be an average over independent repetitions, which can be distributed across an arbitrary number of CPUs and GPUs. Additional compute thus buys statistical precision without requiring larger per-simulation memory. Furthermore, accelerating Pauli propagation now has a legitimate purpose for simulating challenging circuits because it allows collecting more outcomes in the same time. 

Fig.~\ref{fig:bias_assisted_propagation} illustrates the practical advantage of our approach. In this case, every randomized simulation is limited to
$10^5$ stored Pauli strings, and with a carefully chosen truncation threshold, it is able to approach
deterministic propagation results whose final Pauli sums contain tens to hundreds of times
more Pauli strings. Fully unbiased propagation would respect the same memory limits, but has drastically higher variance and would thus require vastly more repetitions to reach comparable errors.

Realizing this improvement over deterministic propagation is not a matter of ``just sampling''.
Naive sampling strategies merely exchange a memory problem for a (potentially worse) variance problem. For example, randomly sampling individual Pauli paths uses constant memory, but the variance grows catastrophically with circuit depth. Beyond this, there are numerous different sampling strategies one might be tempted to apply, each with several key hyperparameters. We put forward our efforts to disentangle their effects and arrive at an optimized scheme.
Crucially, we numerically observe that if \textit{any} one of our components is dropped, the resulting scheme deteriorates substantially, and we provide phenomenological theoretical explanations for this. Our protocol is implemented in the open-source package \ppjllink{}~\cite{rudolph2025pauli} (version 0.9 or higher), and is capable of gate application and resampling applied to $10^9$ Pauli strings in approximately one second (see Fig.~\ref{fig:timing}).

Beyond the exactly verifiable testbed, we showcase the capabilities and freedom of imagination in the randomized Pauli propagation framework by simulating the recent quantum experiments in Ref.~\cite{alam2026onset}. Our simulations estimating marginal collision entropies in an $8\times8$ disordered Heisenberg model hint at the true solution in the ergodic regime sitting in-between results generated by the 1D and 2D tensor network variants. We thus add further evidence that the error mitigation of the quantum results over-corrected for errors. \\

\paragraph*{\textbf{Pauli propagation.}}
Most Pauli-propagation implementations work in the Heisenberg picture, evolving
an observable represented by a Pauli sum
\begin{equation}
    O=\sum_{i=1}^{N_O}c_iP_i,
    \label{eq:pauli_sum}
\end{equation}
where $N_O$ is the number of nonzero-coefficient Pauli strings in $O$. A Pauli rotation
generated by an $n$-qubit Pauli string $G$ acts on a Pauli string $P$ as
\begin{align}
    R_G(\theta)[P]
    &:= e^{i\theta G/2} P e^{-i\theta G/2} \nonumber\\
    &=
    \begin{cases}
        P, & [P,G]=0,\\
        \cos(\theta)P + \sin(\theta)P', & [P,G]\neq 0,
    \end{cases}
    \label{eq:pauli_rotation_action}
\end{align}
where $P'=\frac{i}{2}[G,P]$. The second case maps one incoming Pauli string to
two outgoing Pauli strings. Repeated branching can therefore require exponential (in the number of Pauli rotations)
classical resources to store the resulting Pauli strings and coefficients.

While we focus on Pauli propagation, much of the framework below requires only
a sparse expansion in any given basis with efficiently computable branching rules. It
therefore carries over directly to Majorana propagation~\cite{miller2025simulation,d2025majorana}, has close analogues in
stabilizer-propagation methods~\cite{aaronson2004improved,bravyi2019simulation}, and should extend to future propagation
algorithms formulated in other bases.

Propagation algorithms can be viewed as tree-search problems and are naturally organized
along two independent axes: depth-first vs breadth-first and deterministic vs. randomized.
\textit{Deterministic depth-first} propagation independently visits every Pauli path in the propagation tree, which uses arbitrarily little memory but visits nodes redundantly and thus exhibits an undiminished exponential scaling. This has been explored in Refs.~\cite{rudolph2023classical,beguvsic2023simulating}.
Early variants of Pauli propagation~\cite{rall2019simulation} and the related
extended stabilizer simulation~\cite{bravyi2016improved} instead occupy
the \textit{randomized depth-first} quadrant, sampling the pseudo-probability distributions of gate transfer matrices. The resulting unbiased estimators use little memory and are embarrassingly parallel (and so can be efficiently accelerated on GPUs). Their variance, however, grows rapidly with circuit
depth and limits their practical reach. 

The current state-of-the-art largely
occupies the \textit{deterministic breadth-first} quadrant, using gate-by-gate evolution
with merging of identical Pauli strings, together with controlled
truncation~\cite{rudolph2025pauli,beguvsic2024real}. This approach has outgrown the capabilities of depth-first traversal  due to the significantly slowed-down exponential scaling in the runtime. However, it moved the bottleneck to memory. Even one CPU
thread can fill a cluster node in minutes or hours, rendering acceleration borderline useless for large-scale simulations due to all but the highest-end GPUs having substantially less memory than large CPU nodes.

This work establishes the remaining quadrant: \textit{randomized breadth-first} propagation. We inherit the speed-ups and memory compaction of merging Pauli paths, while using stochastic resampling to keep the number of
stored Pauli strings within a fixed memory limit. Independent runs can then be
averaged with little communication beyond their final outputs, allowing for fast parallelisation. \\

\paragraph*{\textbf{Monte Carlo Pauli propagation.}}
\begin{figure*}
    \centering
    \includegraphics[width=0.9\linewidth]{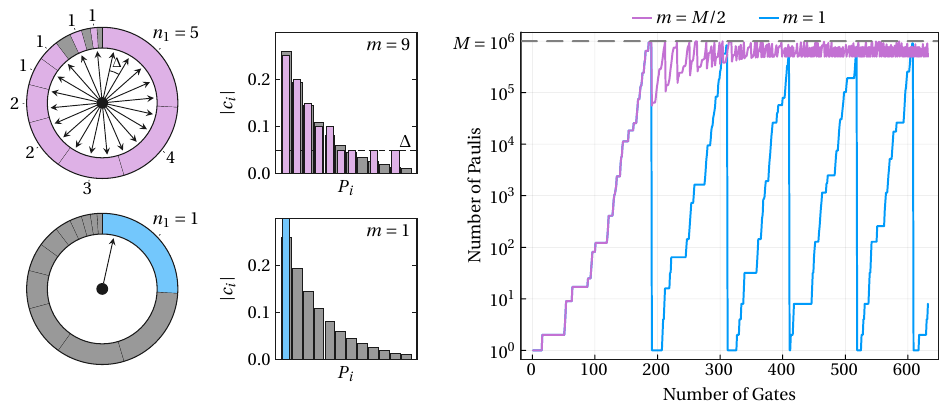}
    \caption{\textbf{Schematic depiction of comb resampling and Monte Carlo Pauli propagation.} Given a Pauli sum $O=\sum_{i=1}^{N_O}c_iP_i$ with arbitrary coefficients $c_i$, their magnitudes $|c_i|$ divided by the coefficient 1-norm $W=\sum_{i=1}^{N_O}|c_i|$ define a distribution that can be sampled to reduce the number of Pauli strings that are stored in memory from $N_O$ to approximately $m$. \textit{Comb resampling} distributes $W/\Delta$ \textit{teeth} with spacing $\Delta$ over the distribution with a random offset in $[0,\Delta)$, and the number of teeth landing within each coefficient's magnitude interval determines the number of times it is drawn with coefficient $\text{sign}(c_i)\Delta$. During propagation, once we surpass $M$ strings, we can resample to small or large sizes $m$, which results in resampling more rarely or more often, respectively. We show that resampling to large $m$ performs best.}
    \label{fig:resampling-schematic}
\end{figure*}

The central challenge for randomized propagation algorithms is to introduce sampling without merely replacing the memory barrier with prohibitive variance.
There are several randomized strategies that one could employ to this end. 
These broadly fall into two categories: \textit{path sampling} strategies which start sampling paths according to the randomized depth-first sampling strategy outlined above and breadth first \textit{resampling} strategies whereby sampling is applied to reduce the number of terms in a stored Pauli sum (see Methods Sec.~\ref{sec:sampling_strategies} for more details).  

Our numerical benchmarking and the theoretical analysis show that successful randomized propagation requires more than just picking a sampling rule. The sampling strategy, the strength of the compression (in the case of resampling strategies), merging, and coefficient truncation are strongly coupled, and apparently natural choices can produce excessive estimator variance and consequently large statistical errors at fixed memory and compute. Disentangling these effects
leads us to the following guiding principle:
\begin{center}
    \textit{Continually make the most of the available memory.}
\end{center}
Concretely, we propose the following protocol. 
\begin{protocol}[Monte Carlo Pauli propagation]
\label{prot:mcpp}
Once close to the memory limit, perform the following steps for
every subsequent gate:
\begin{enumerate}
    \renewcommand{\labelenumi}{(\roman{enumi})}
    \item Apply the gate deterministically.
    \item Merge duplicate Pauli strings.
    \item Truncate according to the chosen coefficient thresholds (introducing a controlled bias).
    \item If the Pauli sum exceeds the memory limit $M$, resample to the largest
    number of unique Pauli strings that leaves sufficient room to apply the
    next gate. There are several resampling schemes that we could use~\cite{murray2016parallel}, but we find calibrated or semi-deterministic variants of comb resampling~\cite{liu1998sequential,kitagawa1996monte} the most effective.
\end{enumerate}  
Independent passes, starting from the very beginning or the first randomized event, can then be averaged until the remaining statistical error
is sufficiently small. 
\end{protocol}

In Fig.~\ref{fig:resampling-schematic} we provide a schematic depiction of resampling to a low ($m=1$) or high ($m=M/2$) target number of unique Pauli strings in step (iv), and how the number of Pauli strings in one pass through a circuit may look accordingly. 
Importantly, the addition of the randomized step (iv) naturally extends steps (i) to (iii) on which propagation algorithms are currently built. It can therefore be combined with any existing (or future) truncation rule and merging scheme, and implemented on top of current Pauli~\cite{rudolph2025pauli,beguvsic2024real,luo2026ppvm,belagali2026propaq} or Majorana~\cite{miller2026monoprop,d2025majorana} propagation packages. 

\begin{figure*}[t]
    \centering
    \includegraphics[width=\textwidth]{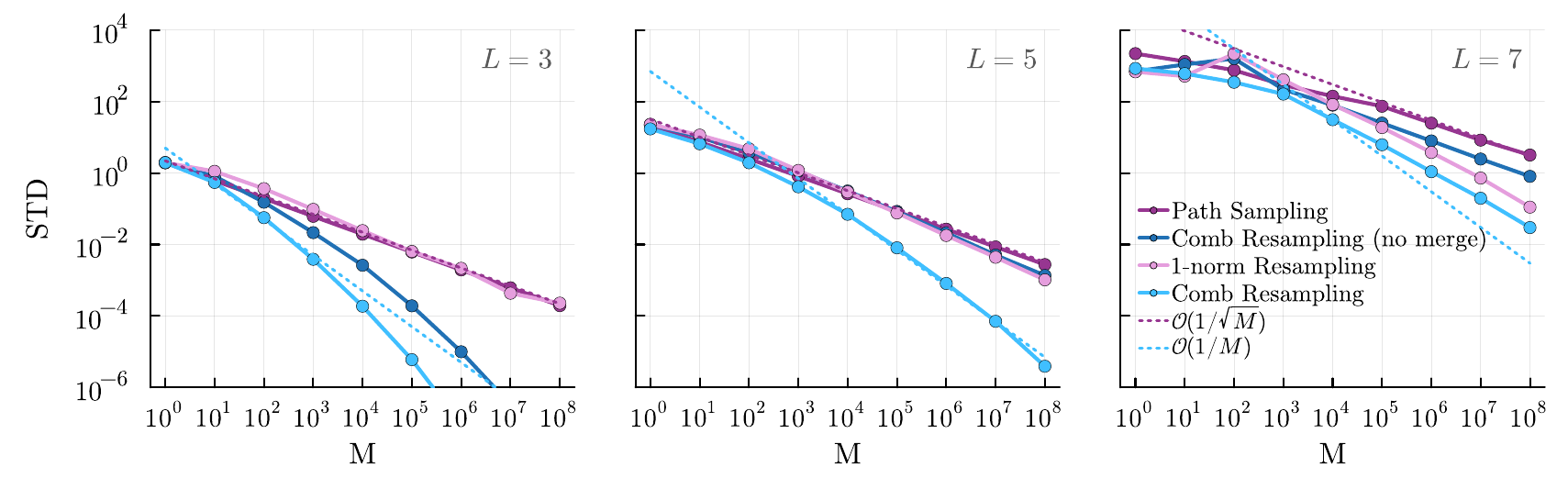}
    \caption{Estimator standard deviation as a function of the memory budget
    $M$ after $L=3,5,$ and $7$ circuit layers. At $L=3$ and $L=5$, comb
    resampling with merging shows an apparent $\mathcal{O}(M^{-1})$ scaling,
    while path sampling, one-norm resampling, and comb resampling without
    merging remain closer to the usual $\mathcal{O}(M^{-1/2})$ behaviour. By
    $L=7$ the engineered method still performs best, but its scaling advantage
    has begun to narrow as the coefficient distribution becomes more
    scrambled.}
    \label{fig:memory_utilization}
\end{figure*}

We next give a compact account of the numerical benchmarks and engineering choices that led to Protocol~\ref{prot:mcpp}. Our testbed is the expectation value computation Tr$[\rho U^\dagger OU]$, where $U$ represents the Trotterization of dynamics generated by the $6\times5$ transverse-field Ising Hamiltonian
\begin{equation}\label{eq:Ising}
    H = J\sum_{\langle i,j\rangle}Z_iZ_j + h\sum_i X_i.
\end{equation}
Here we choose $J=1, h=1.5$ and a Trotter time step $dt=0.1$. We propagate the single-qubit $Z$ observable in the middle of the lattice and use the all-zero state. At this size, we can compare to exact state vector simulation and distinguish errors induced by variance or bias.  
 
In Figure~\ref{fig:memory_utilization}, we compare path sampling with several unbiased resampling strategies, both with and without merging. What we call 1-norm resampling is to draw $m$ independent samples \textit{with replacement} from the pseudo-distribution defined by the Pauli coefficients. This is perhaps the most natural attempt at resampling from $M$ Pauli strings to $m$; however, it does not yield $m$ \textit{unique} Pauli strings. \textit{Comb resampling}, on the other hand, attempts precisely that with dependent yet unbiased sampling. We estimate the statistical uncertainties of the resulting estimators at different numbers of Trotter steps $L$. 

Figure~\ref{fig:memory_utilization} shows that comb resampling and merging can convert additional memory into a much larger reduction in statistical error than the other methods. The dashed $\mathcal{O}(M^{-1/2})$ line gives the usual Monte Carlo scaling expected when a memory budget $M$ supports only $\mathcal{O}(M)$ effectively independent samples. Path sampling, ordinary one-norm resampling, and comb resampling without merging remain close to this behaviour. At $L=3$ and $L=5$, however, comb resampling with merging performs substantially better, with an apparent $\mathcal{O}(M^{-1})$ decrease in estimator standard deviation over the range shown.

This improvement can be understood by distinguishing between
the number of samples drawn, $m_s$, and the number of unique Pauli strings
stored, $m_p$. 
Ordinary one-norm sampling,
in which each Pauli string is drawn with probability
$p_i=|c_i|/\|\vec{c}\|_1$, draws $m_s$ samples. In the absence of merging, these are drawn until the memory is filled and so $m_s \sim M$, giving a statistical error
that decreases as $1/\sqrt{m_s}\sim1/\sqrt{M}$. However, duplicates mean that
the resulting Pauli sum generally contains fewer than $m_s$ unique strings. As a
simple model of comb sampling with merging (and other related resampling schemes~\cite{murray2016parallel}), suppose instead that we continue
sampling until the memory is filled with exactly $m_p \sim  M$ unique strings. The
required number of samples, $m_s$, is then a random variable. For a peaked distribution such as
$p_i\propto i^{-2}$, many samples contribute to strings already in memory, and
we instead find that a typical $m_s$ will give $m_s\propto M^2$ (the derivation is
given in Appendix~\ref{app:error_merging_one_norm}). Since statistical error ordinarily
decreases as $1/\sqrt{m_s}$, a memory budget of $M$ unique Pauli strings can
then support an apparent $1/M$ standard-deviation scaling.
As the circuit becomes deeper and
the coefficient distribution flattens, duplicates become rare. Thus in this deep regime
typically $m_s$ draws will give $M$ unique strings, and the usual $1/\sqrt{M}$ scaling reappears. 

 Figure~\ref{fig:memory_utilization} and the above discussion concern the \textit{unbiased} case. That is, one idea would be to continue with unbiased randomized propagation once conventional propagation hits the memory limit. However, we have found that injecting a controlled bias into the randomized propagation is absolutely central for achieving practical improvements over state-of-the-art deterministic propagation. In particular, if we do not truncate in-between the unbiased resampling events, the algorithm quickly accumulates a prohibitive variance. On the other hand, a carefully chosen threshold just below the target accuracy can delay this onset substantially. Interestingly, truncating small coefficients mirrors the initiator approximation in full configuration interaction quantum Monte Carlo (FCIQMC)~\cite{cleland2010communications}, where only ``walkers'' with large weights may spawn new walkers onto unoccupied basis states.

This truncation removes the small coefficients to suppress variance from the tail and then comb sampling efficiently handles the dominant peaks.
As circuit depth $L$ increases, operator spreading flattens the Pauli spectrum, causing the mean-squared error (MSE) to transition from polynomial to exponential growth in $L$. Increasing the relative truncation threshold $\epsilon_{\rm rel}$ introduces a controlled bias that effectively preserves the semi-peaked distribution longer and delays the exponential blowup of the square root of the MSE (RMSE) to deeper circuits. This is captured by Result~\ref{res:phenomenological_truncated_mse} in Methods where we see that, for a relative truncation threshold $\epsilon_{\rm rel}$, the MSE contains a bias term that scales as $\mathcal{O}(\epsilon_{\rm rel}^2 L^2)$ and an exponentially growing
variance term, with a prefactor that decreases with $\epsilon_{\mathrm{rel}}$, that scales as $\mathcal O(\gamma^{2L}/m_s)$ and where $\gamma \gg 1$ characterizes the growth of the untruncated one-norm.
The delayed polynomial to exponential transition can be seen in Figure~\ref{fig:resampling_increase_bias} in Methods (with the corresponding phenomenological analysis presented in detail in Appendix~\ref{app:one_norm_sampling_mse_model}).

The final calibrated method is shown in Figure~\ref{fig:bias_assisted_propagation}. Every randomized propagation stores at most $M=10^5$ Pauli strings. With
a minimum relative coefficient threshold of $\epsilon_{\mathrm{rel}}=3\times10^{-4}$, the bootstrapped absolute error converges to close to the $1.8\times10^{-2}$ error of the deterministic propagation whose final Pauli sum contains $42M=4.2\times10^6$ Pauli strings. Reducing the threshold to $\epsilon_{\mathrm{rel}}=10^{-4}$ brings the error below $7\times10^{-3}$, outperforming a deterministic calculation storing $96M=9.6\times10^6$ Pauli strings. Thus, with sufficient averaging, a calculation that fits in memory can reproduce results that require almost approximately orders of magnitude more deterministic memory. Thus rather than storing one impossible Pauli sum, we have replaced it with many independent jobs that fit in memory. Note, of course, that higher quality results are attainable via even lower bias at the cost of more averaging. \\

\paragraph*{\textbf{Performance.}}
\begin{figure}
    \centering
    \includegraphics[width=1\linewidth]{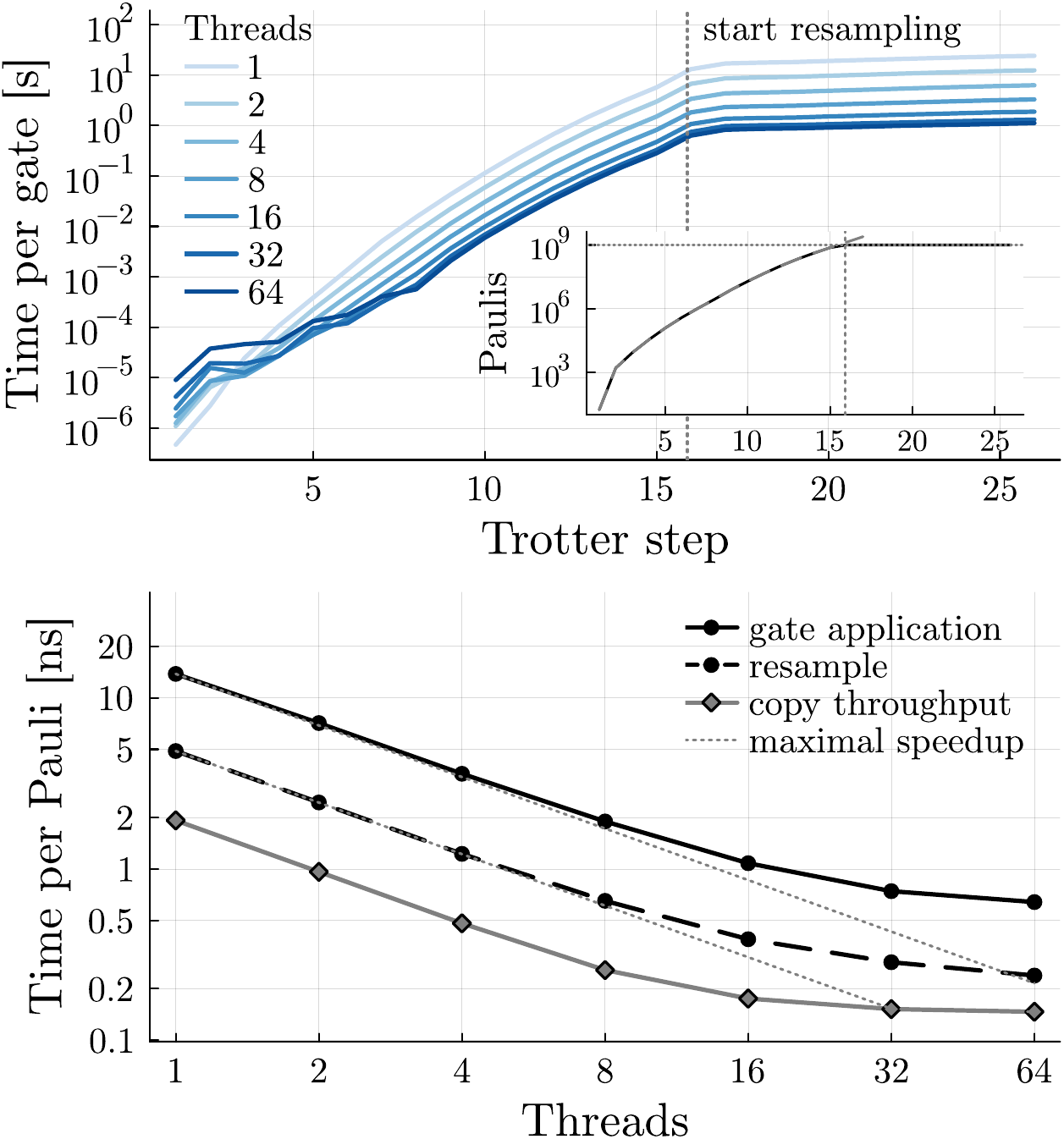}
    \caption{\textbf{Performance of Monte Carlo Pauli propagation on a CPU node.} For a Pauli-$Z$ observable evolved under a Trotterized $6\times6$ TFIM Hamiltonian, we observe that our implementation in \ppjl{} can apply a Pauli rotation and resample $10^9$ Pauli strings in approximately 1 second using 64 threads. We achieve a speedup of approximately $20\times$ from 1 to 64 threads as we approach the limit set by the memory bandwidth of the node.}
    \label{fig:timing}
\end{figure}
There is the danger that a randomized propagation framework relies too much on the in-principle ability to average results until they converge to the desired accuracy. Realistically, one will be restricted to a finite number of repetitions decided by the time one is willing to wait or by financial resources. It is thus important, now more than ever, that propagation algorithms are implemented to be fast. 

Fig.~\ref{fig:timing} depicts the typical phenomenology and performance of our randomized Pauli propagation implementation using the open-source library \ppjllink{}~\cite{rudolph2025pauli} (version 0.9 or higher), run with up to 64 threads on a CPU node with two Xeon Platinum 8360Y processors. We consider the Trotterized evolution of a Pauli-$Z$ observable in the middle of a $6\times6$ lattice under a TFIM Hamiltonian with $J=1$, $h=0.5$, a time step $dt=0.1$, and with an absolute coefficient cutoff of $10^{-7}$. The top panel shows the simulation time per gate and the number of stored Pauli strings as a function of the number of applied Trotter steps. At a cap of $M=10^9$, we start resampling, exchanging the exponential growth in the number of terms for (likely exponentially) accumulating statistical uncertainty. 

The bottom panel of Fig.~\ref{fig:timing} shows the average time per Pauli string it takes to apply a Pauli rotation and resample at the cap of $M=10^9$ as a function of the number of threads. Going from 1 to 64 threads, we achieve a speedup of approximately $20\times$ using a multi-Pauli sum approach inspired by Ref.~\cite{broers2025scalable}, which is similarly implemented in the open-source library \code{monoprop}~\cite{miller2026monoprop}. This allows us to apply a Pauli rotation and resample  $10^9$ Pauli strings in approximately one second, i.e., $1$ns per Pauli string. For comparison, merely copying the Pauli strings and their coefficients from one array to another takes $0.15$ns per Pauli string on the same node, a bound set by its memory bandwidth. 

\medskip

\paragraph*{\textbf{Large-scale implementation.}}

We showcase the flexibility and performance of our method at scale by simulating the Floquet dynamics of an $8\times8$ disordered Heisenberg model previously studied on quantum hardware in Ref.~\cite{alam2026onset}. The initial state $\rho$ is a 2D N\'{e}el state and it is evolved under a \textit{Floquet} circuit generated by a disordered, Heisenberg-like Hamiltonian
\begin{equation}
    H_{ij}=-J\left(X_iX_j+Y_iY_j+Z_iZ_j\right)-h_iZ_i-h_jZ_j,
    \label{eq:disordered_heisenberg}
\end{equation}
where the $h_i$ are random longitudinal fields and $J$ is a fixed interaction strength. We apply four repetitions of
the Floquet layer following the conventions outlined in Ref.~\cite{alam2026onset}. We study the cross-over from the localized to the ergodic regime as $J$ is increased. In the localized regime, the system retains substantial memory of the initial state, while in the ergodic regime reduced subsystems approach Haar-random character.

This crossover can be probed using the \textit{marginal collision entropy} on a patch of qubits $A$. Let $p_A(a)$ be the probability of observing bit string $a$ on the
patch. Its inverse participation ratio (IPR) and the corresponding collision
entropy can be written as
\begin{align}
    \operatorname{IPR}_A
    &:=\sum_a p_A(a)^2
      =\frac{1}{2^{|A|}}
      \sum_{P\in\{I,Z\}^{\otimes |A|}}\langle P\rangle^2,
      \nonumber\\
    S_A&:=-\log_2\operatorname{IPR}_A .
    \label{eq:marginal_collision_entropy}
\end{align}
The IPR measures how concentrated the marginal distribution is. We find that $S_A=0$ for a distribution concentrated on one bit string, while a uniform distribution gives the maximal value $S_A=|A|=9$.

\begin{figure}[t]
    \centering
    \includegraphics[width=0.9\columnwidth]{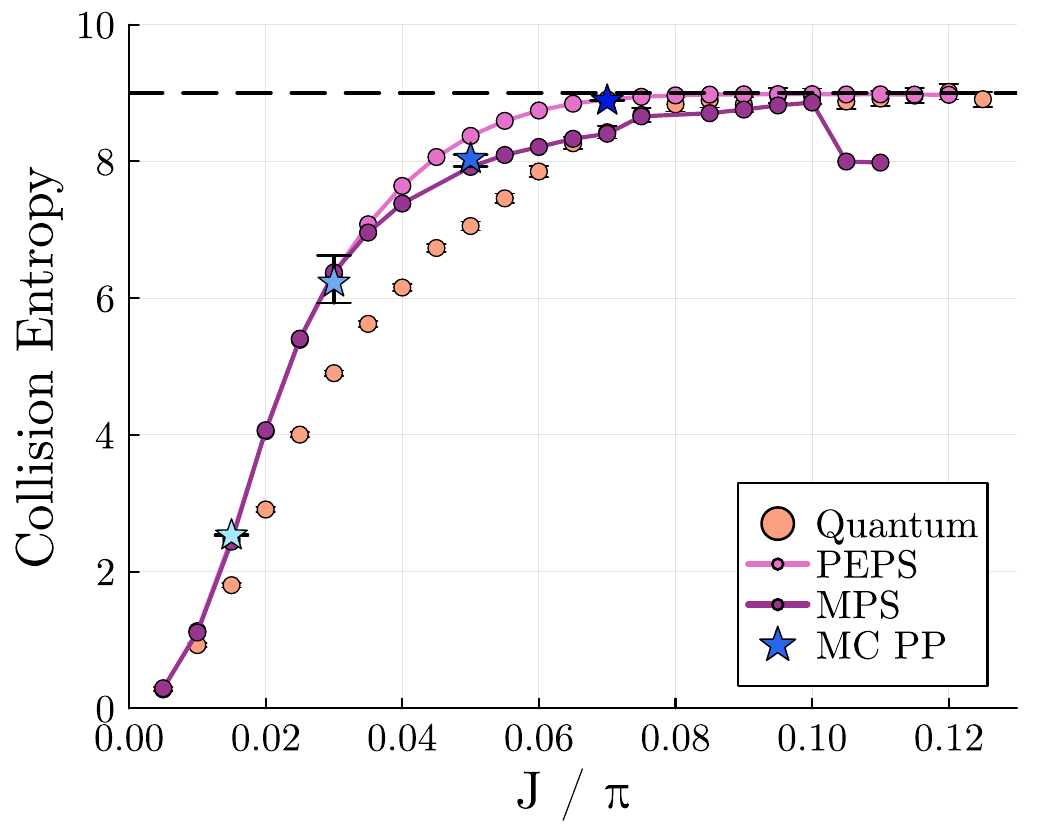}
    \caption{Marginal collision entropy on a $3\times3$ patch after Floquet
    evolution of the $8\times8$ disordered Heisenberg model. Randomized Pauli
    propagation (MC PP) is compared with error-mitigated quantum hardware, PEPS,
    and MPS results. The dashed line marks the maximal value $S_A=9$. The MC PP
    error bars show bootstrapped 16th to 84th percentiles and do not include the
    residual truncation and zero-clipping biases.}
    \label{fig:heisenberg_collision_entropy}
\end{figure}

We consider a  $3\times3$ patch of qubits in the middle of the lattice. Thus, the IPR requires the squared expectation values of 511 non-identity
Pauli strings. It might be tempting to perform randomized Pauli propagation to estimate the expectation values $\langle P \rangle$ and then square them. However, the square will be biased by the magnitude of the variance (see more details in App.~\ref{app:nonlinear_observables}). Instead, we can use the so-called \textit{replica trick}~\cite{overy2014unbiased}. In this approach, we perform two randomized propagations and take the expectation value of their product, allowing us to estimate second-order quantities without additional bias, which would be catastrophic given the several hundred Pauli expectation values we want to sum over. 

Furthermore, instead of individually estimating the expectation values of all 511 non-identity Pauli strings $P_i$, we allocate compute time proportional to $10^{-\gamma |P_i|}$ with $\gamma=0.3$, which approximately balances against the binomial distribution of weights until weight $k=5$, meaning that we have an allocation that nearly uniformly selects the first 5 Pauli weights with an emphasis on weights $2, 3$, and $4$. Low-weight Pauli strings, i.e., the ones that tend to have the largest expectation values, will have the lowest statistical error of the mean and the lowest relative error.

We round out our approach by clipping replica squared expectation values to zero when, after averaging, they remain outside the physical range of $ 0 \leq\langle \hat{P_i} \rangle^2 \leq 1$. Unphysical expectation values (including negative squares) are a signature of high variance which we expect for the highest Pauli weights and thus the lowest expectation values. For further details and a discussion on the incurred biases, see Methods~\ref{app:ergodicity-details}. Here it suffices to say that biases can drastically reduce statistical uncertainty in randomized experiments, but they need to be enjoyed with caution and results must be converged with respect to them.

Figure~\ref{fig:heisenberg_collision_entropy} shows our estimates at
$J/\pi=0.015,0.03,0.05,$ and $0.07$ using only $\maxsize = 10^{7}$, semi-deterministic comb resampling, and relative coefficient threshold $\epsilon_{\text{rel}}$ during propagation, which we continually lower to observe convergence. We compare with results obtained via belief propagation PEPS~\cite{tindall2023gauging,tindall2023efficient,rudolph2025simulating}, MPS~\cite{Vidal2003Efficient,schollwock2011density}, and the error-mitigated quantum results in Ref.~\cite{alam2026onset}. Our estimates are respectively $2.53\pm0.01$, $6.23^{+0.40}_{-0.30}$, $8.04^{+0.13}_{-0.05}$, and $8.89\pm0.00$. They broadly match the tensor network results in the localized regime, where both the PEPS and the MPS can be trusted. In the ergodic regime, where the MPS is not converged in bond dimension (but moving upward) and the PEPS experiences bias from belief propagation and cumulant expansion of high-weight Pauli strings, our Pauli propagation results fall in between. As such, it appears sensible to conclude that the error mitigation efforts in Ref.~\cite{alam2026onset} \textit{over-corrected} hardware outcomes, resulting in collision entropies that are mostly too low. That being said, at the intermediate point at $J/\pi=0.03$ our approach using only $M=10^7$ Pauli strings exhibits its largest statistical uncertainty and the biases cannot be argued to be negligible (see Fig.~\ref{fig:collision_entropy_analysis}). This is (not coincidentally) the regime where conventional Pauli propagation is the most challenging, and our results should not be seen as fully converged. \\

\paragraph*{\textbf{Discussion.}}

In this work, we have developed a randomized propagation framework that naturally recovers and extends the state-of-the-art for propagation algorithms. The overarching goal has been to continually restrict the maximum number of Pauli strings to some threshold $M$, beyond which we turned the estimation of expectation values into a trivially parallelizable averaging problem. We do so by use of sophisticated resampling strategies that allow full control over the current number of Pauli strings while introducing (relatively) low variance. At the same time, we highlight that increasing the memory of a single simulation is still very important. Increasing it can reduce the error asymptotically faster than averaging over more repetitions. Similarly, truncating the propagation remains crucial, as it can remove orders of magnitude of variance at the cost of potentially acceptable bias.

Our proposal gives GPU acceleration of propagation algorithms a legitimate purpose for simulating challenging circuits. Until now, their lower memory capabilities have implied that simulations simply ran out of memory more quickly while yielding worse results. Individual or distributed GPU codes could now perform more independent repetitions in the same time and thus accelerate the estimation of quantities. That being said, given that increasing memory is disproportionately more important than reducing runtime, and Fig.~\ref{fig:timing} demonstrates throughput of approximately $10^9$ Pauli strings per second, it appears that the bar for effective GPU acceleration is non-trivially high.

Beyond our concrete proposal, we describe a randomized propagation toolbox, including importance sampling and estimators for higher-order quantities (see Appendix~\ref{app:additional_tools}) that may inspire practitioners to develop efficient propagation variants that target specific problems or exploit particular structure in the circuit. For example, a second-order variant of our MC propagation scheme was developed in coordination with and used in Ref.~\cite{barron2026observable}, where it was used to estimate expectation values in \textit{out-of-time-ordered correlators} (OTOCs) in scrambling systems. We are confident that imaginative use of our and other randomized tools extends the domain of what is simulable via propagation algorithms. 

\section*{Methods}
\label{sec:methods}

In this section, we describe the methodology used to develop Monte Carlo Pauli propagation. We begin by setting out the design criteria for the sampling scheme, before detailing the different families of sampling schemes that could be employed. We then identify the coupled engineering choices and hyperparameters that must be tuned, and disentangle their effects through careful numerical case studies and analytical arguments. Finally, we provide additional details on the replica trick used to compute nonlinear observables in the ergodicity study.

\subsection{Design criteria}

A randomized propagation algorithm must first and foremost satisfy the hard memory
constraint. Given the various spurious processes taking up memory within a simulation, and importantly our auxiliary data structures for efficient propagation, we define a proxy limit of $M$ stored Pauli strings. If that is ever crossed, we must immediately drop below it again.

Subject to this constraint, the relevant figure of merit is the lowest possible error within a fixed
compute time. For example, we might be able to design an intricate scheme that has very low error, but if it runs very slowly it likely will not be worth it. Let $f$ denote the exact quantity of interest and let $\hat f_r$
be the estimator produced by the $r$th independent randomized pass. We write
$b=\mathbb{E}[\hat f_r]-f$ for its bias and
$\sigma^2=\operatorname{Var}(\hat f_r)$ for its variance. The single-pass
mean-squared error decomposes as
\begin{equation}
    \mathbb{E}\!\left[(\hat f_r-f)^2\right]
    =
    \underbrace{\sigma^2}_{\text{variance}}
    +
    \underbrace{b^2}_{\text{bias}^2}.
    \label{eq:mse_bias_variance}
\end{equation}
If $t_{\mathrm{run}}$ is the runtime of one pass and $T$ the total compute budget, we can fit
\begin{equation}
    R_T
    =\left\lfloor\frac{T}{t_{\mathrm{run}}}\right\rfloor
\end{equation}
independent repetitions into our budget and the estimator for the function $f$ reads
\begin{align}
    \bar f_{R_T}
    =\frac{1}{R_T}\sum_{r=1}^{R_T}\hat f_r.
\end{align}
The error of the $R_T$-repetition estimator becomes
\begin{equation}
    \mathbb{E}\!\left[(\bar f_{R_T}-f)^2\right]
    =\frac{\sigma^2}{R_T}+b^2,
    \label{eq:fixed_compute_mse}
\end{equation}
implying that independent averaging reduces the statistical uncertainty but not the inherent bias.
This allows us to define the fixed-compute standard error as
\begin{align}
    \operatorname{SE}_T
    &:=
    \sqrt{\operatorname{Var}\!\left[\bar f_{R_T}\right]}
    =\frac{\sigma}{\sqrt{R_T}} \nonumber\\
    &\simeq
    \sigma\sqrt{\frac{t_{\mathrm{run}}}{T}}.
    \label{eq:fixed_compute_standard_error}
\end{align}
For an unbiased estimator, $b=0$, this is also the root-mean-squared error of
the averaged estimate. For a biased estimator, further averaging reduces the
standard error but cannot remove the $b^2$ contribution in
Eq.~\eqref{eq:fixed_compute_mse}.

Fully unbiased sampling removes this bias, but its variance can grow rapidly
with circuit depth. This is a phenomenon commonly associated with the \textit{sign problem}~\cite{loh1990sign,troyer2005computational,pan2022sign,pashayan2015estimating}. Once truncation is combined
with randomized propagation, the simulation generally has both types of error.
Our aim is to suppress both for the deepest circuits possible while respecting
fixed memory and compute budgets.

\subsection{Sampling strategies}\label{sec:sampling_strategies}

Consider the Pauli sum in Eq.~\eqref{eq:pauli_sum}. When $N_O$ approaches the limit $M$, we introduce stochastic elements that either replace deterministic ones or are added on top. In the following we unpack our randomized toolbox in hopes that practitioners compose and improve on them. \\

\subsubsection{Path Sampling}
Path sampling substitutes the deterministic application of branching gates with a stochastic choice of which branch to apply.
Consider an operation $\mathcal E_l$ that maps an incoming Pauli string $P_i$ to a weighted sum of Pauli strings,
\begin{equation}
    \mathcal E_l[P_i]
    =\sum_{j=1}^{N_{\mathcal E_l,P_i}}a_{i\rightarrow j}P_j.
    \label{eq:path_local_map}
\end{equation}
The path sampling algorithm chooses one outgoing
branch with probability
\begin{align}
    p(j\mid i)&=\frac{|a_{i\rightarrow j}|}{A_i},
    &
    A_i=\sum_{j=1}^{N_{\mathcal E_l,P_i}}|a_{i\rightarrow j}|,
    \label{eq:path_sampling}
\end{align}
where $A_i$ is the 1-norm of the transition amplitudes $a_{i\rightarrow j}$.
The outgoing coefficient of the sampled branch is set to
\begin{equation}
    c_j'=c_i\,\sign{a_{i\rightarrow j}}A_i,
\end{equation}
where $c_i$ is the path coefficient before sampling. 
For a common Pauli rotation in
Eq.~\eqref{eq:pauli_rotation_action}, the local one-norm is
\begin{equation}
    A_P=|\cos(\theta)|+|\sin(\theta)|,
\end{equation}
with the two branches sampled in proportion to $|\cos(\theta)|$ and
$|\sin(\theta)|$. More details can be found in appendix \ref{app:path_sampling}.

This gives an unbiased estimator that always uses the same constant memory. The sampled paths are fully independent, which can be parallelized at will. Unfortunately, it loses the main advantage of
breadth-first propagation as paths that arrive at the same Pauli string do not
merge, and so the coefficients strictly carry the product of the local one-norms. For most branching gates, including Pauli rotations, we find that $A_P>1$, so this product can grow exponentially with the number of branching gates. 
This makes path sampling a very fast but particularly high-variance algorithm. An illustration of such high variance can be found in Appendix \ref{app:variance_path_sampling}. \\

\subsubsection{Resampling Strategies}
Pauli sum resampling strategies are breadth-first sampling schemes that can be applied to any given Pauli sum $ O=\sum_{i=1}^{N_O}c_iP_i$
to reduce the stored number of Pauli strings in an unbiased fashion. Here we detail several variants.
Throughout, we will use $m$ as the resampling target, i.e., the targeted number of unique Pauli strings after resampling. $m_p$ is the de-facto number of unique Pauli strings that our resampling strategy produces, and $m_s$ is the number of effective samples from the coefficient distribution this corresponds to.
\\

\paragraph*{One-norm resampling.}
Ordinary one-norm resampling draws $m=m_s$ indices independently and with
replacement, with probabilities $p_i=|c_i|/W$, where
\begin{equation}
    W=\sum_{i=1}^{N_O}|c_i|
\end{equation}
is the one-norm of the Pauli sum to be resampled. Each draw contributes a
coefficient $\sign(c_i)W/m_s$. If index $j$ is drawn $N_j$ times, the
resampled Pauli sum after merging is
\begin{equation}
    \widehat{O}
    =\sum_{j=1}^{N_O}\widehat{c}_jP_j,
    \qquad
    \widehat{c}_j
    =\sign(c_j)\frac{N_jW}{m_s}.
    \label{eq:one_norm_resampling}
\end{equation}
The counts $(N_j)_{j=1}^{N_O}$ follow a multinomial distribution with $m_s$
trials and probabilities $(p_j)_{j=1}^{N_O}$. Consequently,
$\mathbb{E}[\widehat{c}_i]=c_i$, and the resampling is unbiased.
Notice how one-norm resampling conserves the one-norm of the resampled Pauli sum exactly:
\begin{equation}
    \sum_j|\widehat{c}_j|
    =\frac{W}{m_s}\sum_jN_j
    =W.
\end{equation}
When resampling reduces the support, the same one-norm is distributed over
fewer coefficients, increasing their average magnitude. Further details are
given in Appendix~\ref{app:one_norm_resampling}, with the corresponding statistical error
analysis in Appendix~\ref{app:one_norm_sampling_error}. \\

\paragraph*{Comb sampling.}
The one-norm sampling strategy described above is sub-optimal for several reasons. First, it requires explicit merging of the samples, which already dominates the runtime of deterministic propagation. Second, it can yield significantly fewer unique Pauli strings  than the target $m$. And third, the near-random memory access to array entries due to the required sorted search can significantly hinder GPU acceleration~\cite{murray2016parallel}. 

Comb sampling, also
called \textit{systematic resampling}~\cite{kitagawa1996monte}, gets around these issues. 
Instead of drawing $m_s$ independent samples from this distribution, comb sampling places a regular grid (a comb) with spacing $\Delta$ and a single random offset $r\sim\operatorname{Uniform}[0,\Delta)$ over the cumulative weights, and counts the number of comb teeth that land in any $ w_i=\sum_{j=1}^{i}\abs{c_j}$.
That number can be computed for any index $i$ via
\begin{equation}
    N_i=
    \left\lfloor\frac{w_i-r}{\Delta}\right\rfloor
    -
    \left\lfloor\frac{w_{i-1}-r}{\Delta}\right\rfloor .
    \label{eq:comb_counts}
\end{equation}
The resampled coefficient is then
\begin{equation}
    \widehat c_i
    =\sign(c_i)N_i\Delta .
    \label{eq:comb_resampling}
\end{equation}
The coefficients are already merged.
The estimator remains unbiased, while each $N_i$ is always the expected number of draws appropriately rounded up or down to the next integer value. The shared offset therefore produces much smaller allocation fluctuations than independent one-norm draws. Concretely, the following theorem helps to explain the biggest difference:
\begin{result}[Upper-bounds on the 1-norm and comb variances, Informal]\label{result:comb1normvariancebounds}
The estimators $\hat f_{\rm 1norm}$ and $\hat f_{\rm comb}$ have variances:
    \begin{align}
        \Var[\hat f_{\rm 1norm}] &\propto \frac{W^2}{m_s} \, , \\
        \Var[\hat f_{\rm comb}] &\propto \frac{W^2}{m_s^2} \sum_i f_i \, , \label{eq:comb}
    \end{align}
where $|W|$ is the total $1$-norm weight and $f_i$ is the fractional remainder of Pauli $i$ that does not completely fill a comb tooth. In general, if the distribution is sufficiently peaked, we have that $\sum_i f_i/m_s < 1$ and hence $        \Var[\hat f_{\rm comb}]  < \Var[\hat f_{\rm 1norm}] $.
\end{result}
 
If $\sum_i f_i$ is a small constant (for example, in the case of the strongly peaked distribution discussed in the main text) then $ \Var[\hat f_{\rm comb}] \propto \frac{1}{m_s^2}$. A schematic and fuller derivation are given in
Appendix~\ref{app:comb_sampling}, and the corresponding error analysis is detailed in Appendix \ref{app:comb_sampling_error}, as well as the additional technical assumptions that are used.

Comb sampling leaves one further parameter to determine, the tooth spacing $\Delta$. The most direct choice of $\Delta$ is $\Delta=W/m$, which lays exactly $m$ teeth across the
distribution. However, this does not generally produce $m$ unique Pauli strings because a coefficient
larger than $\Delta$ catches several teeth and so repeated outcomes are merged leaving
part of the available memory unused. An alternative approach is to \textit{calibrate} comb resampling to a smaller $\Delta$ so that the expected number of unique Pauli strings $m_p$ is closer to the desired number $m$. We do so by computing the expected number of unique Pauli strings $F(\Delta) = \sum_i\min\left(1,\frac{\abs{c_i}}{\Delta}\right)$, and updating $\Delta$ as 
\begin{align}
    \Delta_{k} \frac{F(\Delta_k)}{m}\to\Delta_{k+1}
    \,.
\end{align}
In Appendix~\ref{app:comb_sampling} we expand on why the method works and how it usually converges within a handful of
passes over the coefficient array. \\ 

\paragraph*{Semi-deterministic resampling.}
One main bottleneck of the sampling schemes described above is that they do not have an intrinsic mechanism for dealing with repeated draws stemming from large coefficients $c_i$ that have $N_i \geq 2$. One pragmatic solution to this problem is to use a semi-deterministic variant, where 
large coefficients are not resampled but kept deterministically.
We therefore divide the Pauli sum into a deterministic head $D$, usually
containing the largest coefficients, and a stochastic tail $T=D^c$.
The coefficients indexed by $D$ are kept exactly and the tail is resampled so
that the total Pauli sum contains at most $m$ Pauli strings. Thus
$m_{\text{stoch}}=m-|D|$ Pauli strings are allocated to the tail, which can be covered by any resampling scheme. In practice, we use comb resampling for the stochastic part.
With
$W_{\text{stoch}}=\sum_{i\in T}|c_i|$ the one-norm of the stochastic tail, this gives
\begin{equation}
    \widehat c_i=
    \begin{cases}
        c_i, & i\in D,\\
        \sign{c_i}\,N_i W_{\text{stoch}}/m_{\text{stoch}}, & i\in T.
    \end{cases}
    \label{eq:deterministic_resampling}
\end{equation}
The largest coefficients are retained exactly and do not participate in the resampling at all. A notable connection is that the same construction underlies the fast randomized iteration framework for stochastic linear algebra and projector quantum Monte Carlo~\cite{lim2017fast}.

Removing the head from the sampling
normalization replaces $W$ by the smaller $W_{\text{stoch}}$, which substantially reduces
the noise placed on the tail.
It also makes $N_i \geq 2$ much less common, so the uncalibrated tail spacing $\Delta=W_{\text{stoch}}/m_{\text{stoch}}$ already comes close to storing $m_{\text{stoch}}$ unique Pauli strings. 

Semi-deterministic resampling is unbiased provided that the tail is resampled in an unbiased fashion.
The calibrated and semi-deterministic variants of comb resampling, together with their local variances
and implementation costs, are described in detail in
Appendix~\ref{app:semi_deterministic_comb_resampling}. \\

\paragraph*{Alternative tools.}
There are more randomized tools that one could apply to Pauli propagation.
One is \textit{stochastic pruning}.  It probabilistically removes
small coefficients below a threshold with a probability connected to their magnitude but scales up surviving coefficients to remain unbiased. In special cases, it is equivalent to semi-deterministic resampling, but, without calibration of the threshold, it does not have a mechanism to limit the number of surviving Pauli strings to $m$.

It is also possible to apply gates semi-deterministically. Pauli strings with large coefficients could be allowed to branch deterministically, whereas gate application to small coefficients could be performed via randomized path sampling. Yet again, it is not straightforward to precisely control the number of stored Pauli strings; the main advantage of resampling.

Finally, there is \textit{importance sampling}~\cite{robert2004monte}, which entails re-weighting probabilities to favor Pauli strings
expected to contribute strongly to the final observable. Importance sampling can be employed within any randomized routine, but requires specific domain knowledge. If one knows which propagating Pauli strings end up being more relevant for an expectation value, re-weighting probabilities (while adequately re-normalizing coefficients) can lower the variance of the estimator. If one was mistaken about which strings are important, the variance does not decrease or could even increase. Sampling by 1-norm can already be viewed as a form of importance sampling, and we have not found other strategies, e.g., favoring low Pauli weight~\cite{angrisani2024classically}, to significantly improve on that.

A more detailed description of the alternative tools, including their estimators, biased
variants, and implementation details is deferred to
Appendix~\ref{app:additional_tools}.

\subsection{Calibrating the randomized protocol}

Having introduced the candidate randomized strategies, we now examine the quantitative choices that govern the practical performance of Protocol~\ref{prot:mcpp}. This protocol uses variants of comb resampling combined with coefficient truncation, which we have found to perform best in the main-text benchmarks. We focus on two parameters, the resampling target $m$ (i.e., the number of unique Pauli strings we aim to sample) and the relative truncation threshold $\epsilon_{\mathrm{rel}}$. The resampling target controls how strongly the Pauli sum is compressed and hence how soon another resampling event is required.
The truncation threshold determines how much sampling variance is suppressed in exchange for introducing bias. Together, these parameters determine the runtime of each pass and the accuracy attainable within fixed memory and compute budgets. We analyze their effects using analytical variance estimates and the $6\times5$ transverse-field Ising testbed described in the main text, Eq.~\eqref{eq:Ising}. \\

\begin{figure}
    \centering
    \includegraphics[width=1\linewidth]{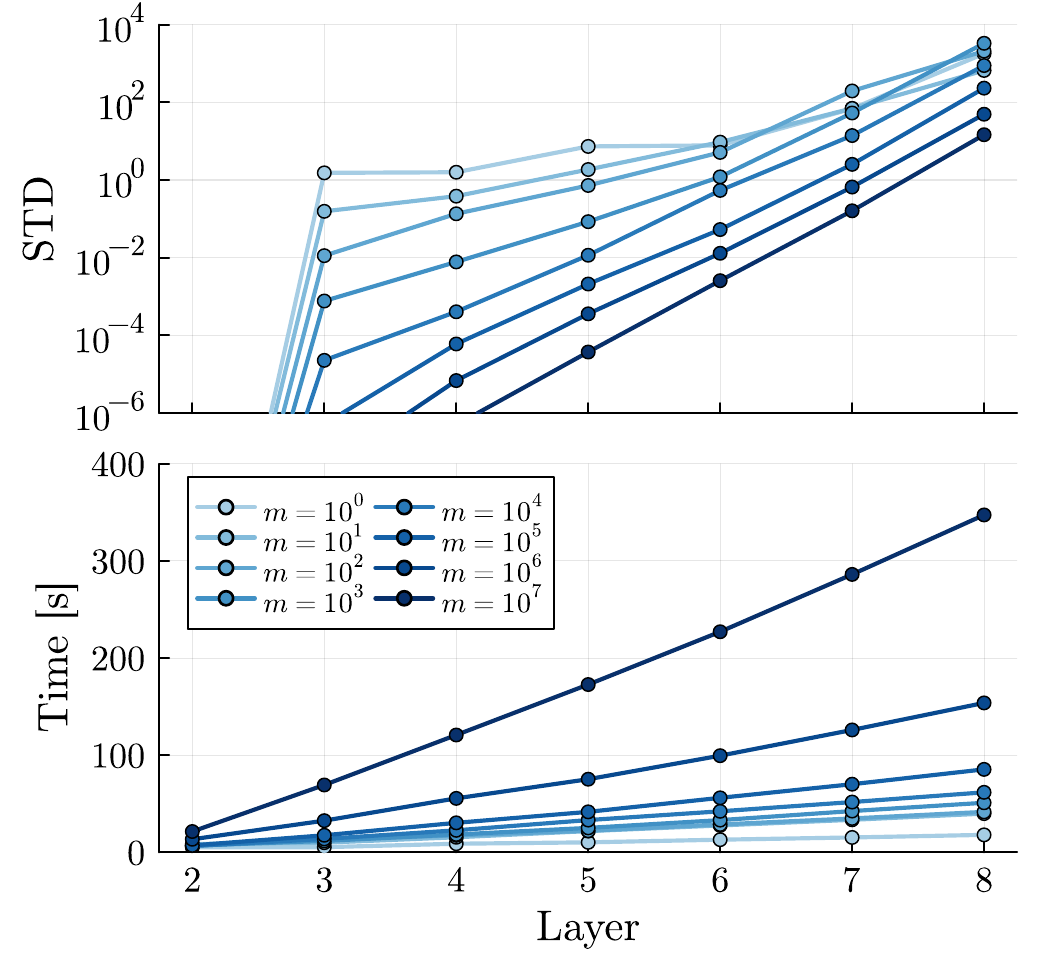}
    \caption{The choice of resampling target $m$ at a fixed limit of
    $M=2\times10^7$ stored Pauli strings, shown for unbiased propagation of a 30-qubit,
    two-dimensional transverse-field Ising circuit. Top, the standard
    deviation of a single pass decreases consistently as $m$ is increased from
    $1$ to $10^7=M/2$. Bottom, the runtime varies by only around a factor of
    twenty across the same range, far too little to offset the
    orders-of-magnitude improvement in statistical error.}
    \label{fig:resampling_size_tradeoff}
\end{figure}

\paragraph*{How many Pauli strings should be retained after resampling?}

A small $m_p$ (i.e., a small number of unique Pauli strings drawn) leaves more room for subsequent deterministic propagation, while a large $m_p$ preserves more of the current Pauli sum but reaches the memory limit again sooner.

The statistical case for targeting large $m_p$ is straightforward. 

Let $\mathcal K_{\mathrm{res}}$ denote the set of gate indices at which
resampling occurs, and let $\STD^2_{m_p}(\hat{f}_k)$ denote the contribution to the final estimator variance from the stochastic error introduced at event $k$. For unbiased resampling, a convenient worst-case bound takes the form
\begin{equation}
    \STD^2_{m_p}(\hat{f})
    \leq
    \sum_{k\in\mathcal K_{\mathrm{res}}(m_p)}
    \STD^2_{m_p}(\hat{f}_k).
    \label{eq:resampling_mse_bound}
\end{equation}
The sampler-specific bounds on $\STD^2_{m_p}(\hat{f}_k)$ are given in Result~\ref{result:comb1normvariancebounds}. The derivation of this result is detailed in Appendices \ref{app:one_norm_sampling_error_unbiased} and \ref{app:comb_sampling_error_unbiased} along with the additional technical assumptions that are used. For both one-norm and comb resampling, the variance introduced by an individual event decreases as the sampling effort $m_s$ is increased, although the two methods exhibit different dependences on $m_s$ and on the coefficient distribution. 

Intuitively, a large $m_p$ reduces the
statistical noise of a single pass, while a small $m_p$ will be quicker and so will permit more independent passes within the same compute budget.
We can weigh up these effects and determine the best targeted $m_p$ by studying
the fixed-compute standard error $\STD_{m_p}(\hat{f})\sqrt{t_{\mathrm{run}}(m_p)}$ introduced in
Eq.~\eqref{eq:fixed_compute_standard_error}. 
Here $\STD_{m_p}(\hat{f})$ includes the STD accumulated over the full schedule $\mathcal K_{\mathrm{res}}(m_p)$. 

The runtime advantage of small $m_p$ is weaker than one might expect. Between
resampling events, the number of stored Pauli strings grows approximately
exponentially from $m_p$ back to $M$. Most of the work is therefore concentrated
in the final few gates, when the Pauli sum is already close to the memory limit.
A much stronger compression can consequently discard far more information while
saving comparatively little time. Whether the reduction in
$t_{\mathrm{run}}(m_p)$ compensates for the increase in $\STD_{m_p}(\hat{f})$ must therefore
be determined empirically.

Figure~\ref{fig:resampling_size_tradeoff} measures both the constituent quantities, $\STD$ and $t_{\mathrm{run}}$, as a function of the number of circuit layers. We fix
$M=2\times10^7$ and vary the resampling target $m$ from $1$ to
$m=10^7=M/2$. The top panel shows that the largest $m$ gives the smallest
single-pass standard deviation at every displayed layer, with an improvement of
several orders of magnitude over $m=1$. The bottom panel shows that the runtime difference between the two extremes
is only around a factor of twenty. One may wonder why the runtime is not more drastically different, which can be explained largely by the exponential growth of the Pauli sum with the number of gates causing the average number of Pauli strings to be only approximately that factor of twenty smaller.

The optimal time vs variance tradeoff is therefore given for the largest target that leaves just enough memory to apply
the next gate. If one Pauli string can branch into at most $b_{\max}$ strings, a
conservative choice is
\begin{equation}
    m=\left\lfloor\frac{M}{b_{\max}}\right\rfloor .
\end{equation}
For example, for Pauli rotations $b_{\max}=2$ and hence $m=M/2$ if $M$ should not be significantly surpassed.\\

\begin{figure}
    \centering
    \includegraphics[width=0.45\textwidth]{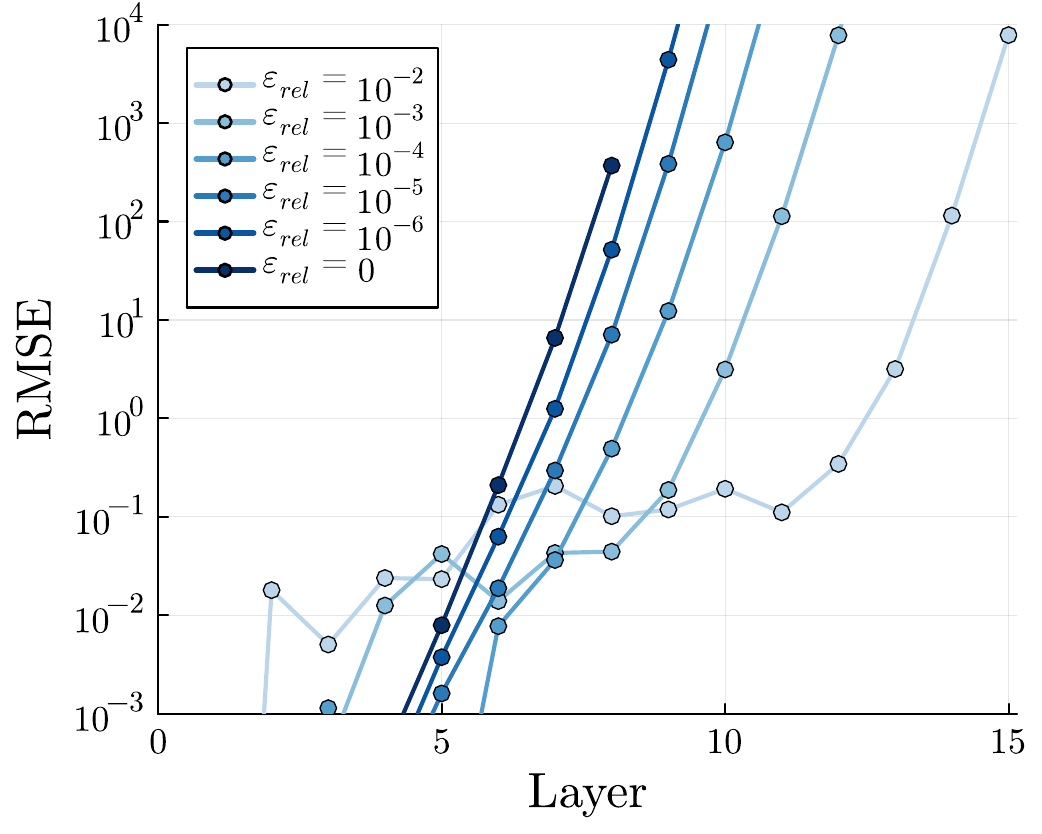}
    \caption{\textbf{Error scaling of resampling Pauli propagation with varying amounts of bias.} In contrast to the previous unbiased results, we now constantly apply a relative coefficient truncation with threshold $\minrelcoeff = 10^{-2}, \dots, 0$, where $\minrelcoeff = 0$ represents the unbiased benchmark. The maximum number of Paulis is $\maxsize = 10^{5}$. We observe that bias can delay the exponential onset of statistical uncertainty for several Trotter layers at the cost of not being able to achieve an error lower than the bias floor. For a bias that is just low enough for a target accuracy, the STD of the biased estimator can be multiple orders of magnitude lower than for the unbiased estimator.}
    \label{fig:resampling_increase_bias}
\end{figure}

\paragraph*{How much controlled bias should be introduced?}

In practice, we introduce controlled bias by discarding the smallest coefficients after deterministic gate application and merging. For the purposes of our phenomenological calculation, assume that between layers $k$ we truncate all coefficients satisfying $|c_i^{(k)}|\leq b_k$, where $b_k>0$ is a chosen threshold. The one-norm of the retained Pauli sum is therefore
\begin{equation}
    \widetilde W_k
    =
    \sum_{i=1}^{N_O}
    |c_i^{(k)}|
    \mathbf{1}\!\left(|c_i^{(k)}|>b_k\right).
\end{equation}
If resampling is required at that layer, it is applied to this truncated sum. For ordinary one-norm resampling, the retained indices are drawn with probabilities
\begin{equation}
    p_i^{(k)}
    =
    \frac{
        |c_i^{(k)}|
        \mathbf{1}\!\left(|c_i^{(k)}|>b_k\right)}
        {\widetilde W_k}.
\end{equation}
If index $j$ is drawn $N_j$ times, the resampled Pauli sum after merging is
\begin{equation}
    \widehat O
    =
    \sum_{j=1}^{N_O}\widehat c_jP_j,
    \qquad
    \widehat c_j
    =
    \sign(c_j^{(k)})\frac{N_j\widetilde W_k}{m_s}.
    \label{eq:one_norm_resampling_after_truncation}
\end{equation}
The same truncation step can be combined with any of the resampling schemes introduced above. We write the ordinary one-norm case explicitly because it makes the resulting bias--variance trade-off particularly transparent.

The threshold may be fixed throughout the propagation or recomputed relative to the largest coefficient at each layer,
\begin{equation}
    b_k
    =
    \epsilon_{\mathrm{rel}}
    \max_i|c_i^{(k)}|,
    \qquad
    \epsilon_{\mathrm{rel}}>0.
\end{equation}
Resampling remains unbiased with respect to the truncated Pauli sum, but the combined truncation-and-resampling procedure is biased with respect to the original sum. Discarding small coefficients reduces the retained one-norm and therefore the variance introduced by subsequent resampling. This competition is captured by the following phenomenological model.

\begin{figure*}
    \centering
    \includegraphics[width=0.7\linewidth]{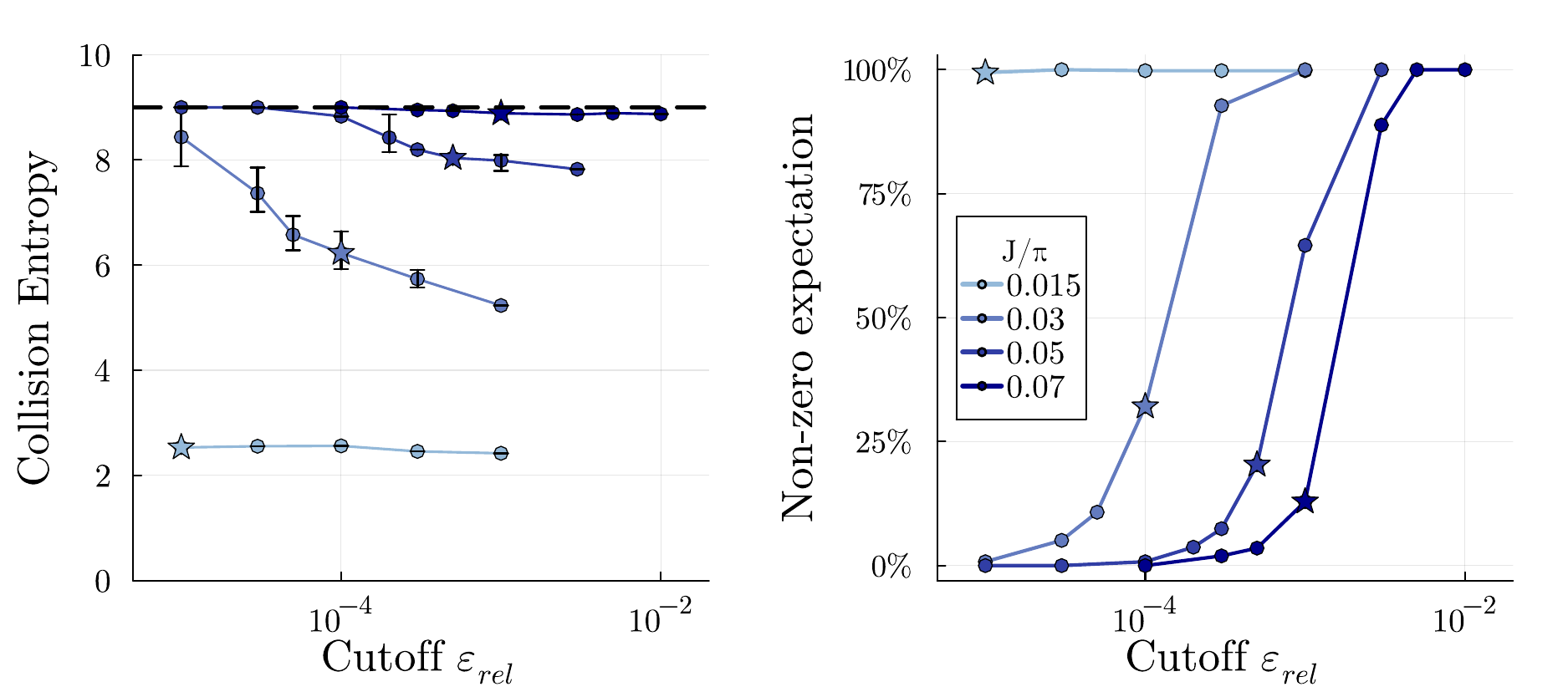}
    \caption{\textbf{Analysis of the randomized Pauli propagation result for estimating ergodicity.} 
    Our approach for estimating marginal collision entropy suffers from two competing biases. Truncating small coefficients below $\epsilon_\text{rel}$ pulls down the estimate, and the clipping of unphysical, high-variance results pulls it up. \textbf{Left:} As we lower $\epsilon_\text{rel}$, we observe the collision entropy estimates grow and partly stabilize, until the clipping bias takes over. \textbf{Right:} Fraction of expectation values for which we have non-zero estimates. Zero estimates are a sign of high variance and clipping. As that fraction changes from $100\%$ to $0\%$, our simulation transitions from mostly deterministic to fully corrupted by statistical noise. We highlight the data points shown in Fig.~\ref{fig:heisenberg_collision_entropy}. For $J/\pi = 0.03$, we cannot confidently point to a statistically converged simulation, which would exhibit a plateau on the left panel.
    }
    \label{fig:collision_entropy_analysis}
\end{figure*}

\begin{result}[Phenomenological model for the MSE of truncated one-norm resampling]
\label{res:phenomenological_truncated_mse}
Assume that relative truncation is applied at each layer with $b_k=\epsilon_{\mathrm{rel}}\max_j|c_j^{(k)}|$, that the accumulated bias per layer is proportional to $\epsilon_{\mathrm{rel}}$, and that the untruncated one-norm grows as $W_k\sim\gamma^k$. Then
\begin{equation}
    \operatorname{MSE}(\widehat f)
    \lesssim
    \left(C_1\epsilon_{\mathrm{rel}}L\right)^2
    +
    C_2
    \frac{
        \left(1-C_3\epsilon_{\mathrm{rel}}\right)^2
        \gamma^{2L}}
        {m_s},
    \label{eq:phenomenological_truncated_mse}
\end{equation}
where $C_1,C_2,C_3>0$ are problem-dependent constants.
\end{result}

The first term describes the truncation-induced bias floor, while the second describes the sampling variance. Although derived for truncated one-norm resampling, the model captures the same qualitative competition encountered when truncation is combined with comb resampling. A fuller derivation is given in Appendices~\ref{app:one_norm_sampling_error_bias} and \ref{app:comb_sampling_error_bias} along with the additional technical assumptions that are used.

Figure~\ref{fig:resampling_increase_bias} tests this trade-off at fixed $M=10^5$. Increasing $\epsilon_{\mathrm{rel}}$ delays the rapid growth of the RMSE, but eventually produces a larger bias floor. The useful choice is therefore the largest threshold whose bias remains below the target accuracy. We use a relative rather than an absolute threshold because resampling changes the characteristic coefficient scale, causing a fixed absolute threshold to become progressively less effective.

In practice, we begin with relatively aggressive truncation and repeat the calculation with successively smaller values of $\epsilon_{\mathrm{rel}}$. This allows the estimates to be converged successively, ending with the least biased estimate that can still be resolved within the available compute budget. \\

\subsection{Numerical details for estimating ergodicity}\label{app:ergodicity-details}

Squaring a one-norm resampled estimator
$\widehat O$ would add its variance as a bias, which over a sum of 511 propagations would dominate the error. Instead, we propagate two
independent \textit{replica} for each Pauli string and form
\begin{equation}
    \widehat q_P
    =\widehat{\langle P\rangle}^{(1)}
     \widehat{\langle P\rangle}^{(2)},
    \,\, \text{such that} \, \, \, \,
    \mathbb{E}[\widehat q_P]=\langle P\rangle^2.
    \label{eq:dual_replica}
\end{equation}
This so-called \textit{replica trick}~\cite{overy2014unbiased} produces an unbiased estimator for the squared expectation value (assuming the one-norm estimator is unbiased), and we can sum up all squared expectation values to attain an estimator for the IPR that is only biased by deterministic truncation applied throughout propagation. 
Appendix~\ref{app:nonlinear_observables} gives the full construction.

To produce the data in Fig.~\ref{fig:collision_entropy_analysis}, each replica is propagated with a limit of $M=10^7$ stored Pauli strings using
semi-deterministic comb resampling and the relative coefficient truncation threshold
$\epsilon_{\mathrm{rel}}$ as indicated. We lower this threshold to observe convergence of the estimator. 

A uniform
allocation of compute time over the 511 Pauli strings would devote most time to weight $3-6$
Pauli strings, of which there are most in a patch of size 9. These strings incur higher variance than low-weight strings but typically contribute very little
in the scrambling regime. We instead allocate repetitions approximately as
$R(P)\propto10^{-0.3|P|}$, concentrating effort on the lower-weight observables
that carry most of the signal.

We further apply a physicality-motivated zero-clipping scheme to the estimates. After the two
independent estimates have been averaged over repetitions, their product
$\widehat{\langle P\rangle}^2$ can remain negative or exceed one, even though the
exact quantity lies in $[0,1]$. Such unphysical values signal that the estimate
is dominated by high variance, most commonly for high-weight Pauli
strings whose true expectation values are expected to be small. We \textit{clip} these
estimates to zero before forming the IPR. This stabilizes the estimator, but introduces a new bias. In fact, this bias competes with the bias introduced by truncating small coefficients. We find that the latter tends to decrease the collision entropy, while clipping unphysical expectations pulls the entropy up. This complicates the convergence, and Fig.~\ref{fig:collision_entropy_analysis} outlines our best-faith attempt to find the least biased estimate for the collision entropy.

\begin{acknowledgements}
    The authors would like to thank Faisal Alam and Ashley Montanaro for providing and discussing the data for Fig.~\ref{fig:heisenberg_collision_entropy}.
    MSR acknowledges funding from the 2024 Google PhD Fellowship and the Swiss National Science Foundation [grant number 200021-219329]. R.P. acknowledges the support of the SNF Quantum Flagship Replacement Scheme (Grant No.~215933). ZH acknowledges support from the Sandoz Family Foundation-Monique de Meuron program for Academic Promotion. We acknowledge support of the NCCR MARVEL, a National Centre of Competence in Research, funded by the Swiss National Science Foundation (grant number 205602).
\end{acknowledgements}

\medskip

\section*{AI Usage Statement}
The authors acknowledge the use of generative AI tools (Claude, ChatGPT, Gemini), in particular thanks to the ChatGPT for Academic Researchers program, for the purposes of literature search, combining drafts, editing text, deriving and formatting equations, and implementing high-performance code. All AI-assisted content was reviewed, verified, and revised by the authors, who take full responsibility for the content of this work.

\clearpage

\bibliography{quantum, refs}

\clearpage
\onecolumngrid
\appendix

\section{Sampling schemes in detail}
\label{app:resampling_comparison}

The Methods section introduces the sampling schemes in the order in which they are
used. Here we give the full definitions and collect their local statistical
properties. As before, we use $m$ to denote the target number of unique Pauli strings retained after each resampling event, $m_p$ for the de-facto number of unique Pauli strings retained by any resampling method, and $m_s$ for the number of effective raw draws. These
numbers need not agree. If several draws select the same Pauli string, merging stores
their aggregate coefficient only once, so $m_p<m_s$ is possible and often
desirable.

Throughout this appendix, the object to be compressed is a Pauli sum
\begin{equation}
    O=\sum_{i=1}^{N_O}c_iP_i,
    \qquad
    W=\sum_{i=1}^{N_O}\abs{c_i}.
    \label{eq:app_pauli_sum_norm}
\end{equation}
For real coefficients, $\sign(c_i)$ denotes the usual sign. For a
complex coefficient it should be read as the phase $c_i/\abs{c_i}$. A
compression is unbiased if every coefficient, and hence the full Pauli sum,
is recovered in expectation.

\subsection{Path sampling}\label{app:path_sampling}

Path sampling is the simplest randomized propagation algorithm and the baseline
that any method using more memory should beat. Consider a circuit channel
$\EC=\EC^{(L)}\circ\cdots\circ\EC^{(1)}$ acting on an
initial Pauli string $P_0$ with coefficient $c_{P_0}=1$. At step $k$, an
incoming Pauli string is mapped to
\begin{equation}
    \EC^{(k)}[P_{k-1}]
    =\sum_{P_k}a_{P_{k-1}\rightarrow P_k}P_k.
    \label{eq:app_local_transfer}
\end{equation}
Define the column one-norm as the sum of the absolute values of the different transition amplitudes
\begin{equation}
    A_k(P_{k-1})
    =\sum_{P_k}\abs{a_{P_{k-1}\rightarrow P_k}}.
    \label{eq:app_local_column_norm}
\end{equation}
Path sampling selects one outgoing Pauli string with probability proportional to the magnitude of the transition amplitude
\begin{equation}
    p(P_k\mid P_{k-1})
    =\frac{\abs{a_{P_{k-1}\rightarrow P_k}}}
    {A_k(P_{k-1})},
    \label{eq:app_path_probability}
\end{equation}
and updates the running coefficient according to
\begin{equation}
    c_{P_k}
    =c_{P_{k-1}}\sign
    \!\left(a_{P_{k-1}\rightarrow P_k}\right)
    A_k(P_{k-1}).
    \label{eq:app_path_update}
\end{equation}

The compensation by the local one-norm makes the step unbiased. Indeed, if we compute the expectation value of the Pauli string and its coefficient from the step $k-1$ to $k$ we find that
\begin{align}
    &\Ebb_{P_k\sim p(\:\cdot\:\mid P_{k-1})}
    \left[c_{P_k}P_k\right] \nonumber\\
    &\quad=\sum_{P_k}
    \frac{\abs{a_{P_{k-1}\rightarrow P_k}}}{A_k(P_{k-1})}
    c_{P_{k-1}}\sign
    \!\left(a_{P_{k-1}\rightarrow P_k}\right)
    A_k(P_{k-1})P_k \nonumber\\
    &\quad=c_{P_{k-1}}\sum_{P_k}
    a_{P_{k-1}\rightarrow P_k}P_k
    =c_{P_{k-1}}\EC^{(k)}[P_{k-1}]\,,
    \label{eq:app_path_unbiased_step}
\end{align}
where in the first equality we substituted the probability defined in Eq.~\eqref{eq:app_path_probability}, and in the second equality we substituted the expression for $\EC^{(k)}[\cdot]$ in Eq.~\eqref{eq:app_local_transfer}. Iterating the construction produces a path $(P_0,P_1,\ldots,P_L)$ and the exact identity
\begin{equation}
    \EC[P_0]
    =\Ebb\left[c_{P_L}P_L\right],
    \qquad
    c_{P_L}=\prod_{k=1}^{L}
    \sign\!\left(a_{P_{k-1}\rightarrow P_k}\right)
    A_k(P_{k-1}).
    \label{eq:app_path_full_estimator}
\end{equation}
By linearity, the same sampled path gives an unbiased estimator of
$\Tr[\rho\EC[P_0]]$.

This method needs only constant memory per path and is embarrassingly parallel (i.e., it can be divided into many independent tasks requiring little or no communication between them). It also prevents paths from meeting before the final average is taken. Duplicate Pauli strings cannot merge, and cancellations that a breadth-first calculation would perform internally must instead be resolved statistically across independent trajectories.

Path sampling can still benefit from filling a CPU or GPU with many independent paths at once. Large batches amortize kernel launches and per-gate overheads. This changes the runtime, but not the $1/\sqrt R$ convergence with the total number $R$ of independent paths.

Resampling the full Pauli sum is more ambitious. Using memory we can alter the statistical error itself by allowing deterministic evolution and interference between stochastic events, as we explore in the next section.

\subsection{Resampling strategies}

\subsubsection{Independent one-norm resampling}
\label{app:one_norm_resampling}

Rather than make a random decision at each gate, one-norm resampling first lets
the full Pauli sum evolve until the memory barrier is hit. When it contains too many Pauli strings, it defines a
probability distribution based on the one norm of each coefficient
\begin{equation}
    p_i=\frac{\abs{c_i}}{W}\,,
    \label{eq:app_one_norm_probability}
\end{equation}
where we recall that $W = \sum_i \abs{c_i}$ is the one norm of the Pauli coefficients over the support in memory.

After sampling from this distribution $m_s$ times, and obtaining $N_i$ times each Pauli $P_i$, the estimator of the Pauli coefficient $\widehat c_i$ can be written as 
\begin{equation}
    \widehat c_i
    =N_i\:\sign(c_i)\frac{W}{m_s}
    \label{eq:app_one_norm_estimator}
\end{equation}
where the vector of the repetitions $N_i$ follows a multinomial distribution $(N_1,N_2,\ldots)\sim\operatorname{Multinomial}(m_s,\{p_i\})$. With this we can easily see that the expectation value of $N_i$ is $\Ebb[N_i]=m_sp_i$, which means that the expectation value of the estimator $\widehat c_i$ of the $P_i$ coefficient is unbiased
\begin{equation}
    \Ebb[\widehat c_i]
    =\sign(c_i)\frac{W}{m_s}m_s
      \frac{\abs{c_i}}{W}
    =c_i.
    \label{eq:app_one_norm_unbiased}
\end{equation}

Similarly, we use the fact that $\Var[N_i] = m_sp_i(1-p_i)$ and we can compute

\begin{equation}
    \Var[\widehat c_i]
    =\sign(c_i)^2\frac{W^2}{m_s^2}m_s
      \frac{\abs{c_i}}{W}\left(1 - \frac{\abs{c_i}}{W}\right)
    =\frac{W}{m_s}|c_i|\left(1 - \frac{\abs{c_i}}{W}\right).
    %\label{eq:app_one_norm_unbiased}
\end{equation}

The method is simple and exactly unbiased, but it does not generally produce
$m_s$ unique Pauli strings. A large coefficient may be drawn many times. Merging all
of those copies saves memory, although it also means that a nominal request for
$m_s$ samples can leave far fewer than $m_s$ stored Pauli strings. This distinction is
minor for a flat coefficient distribution (most samples will be of new strings) but can be substantial for a peaked one (some strings will be sampled many times).

A tempting modification is to draw repeatedly until exactly $m_p$ distinct
Pauli strings have appeared. 
This fixes the number of stored Pauli strings, but adds a huge time overhead when the number of distinct draws to obtain $m_p$ distinct Paulis is significantly larger than $m_p$ itself. Moreover, a naive normalization introduces a small bias. Indeed, if we draw $m_s$ times from this distribution\footnote{Notice that now $m_s$ is a random variable, as we will draw from the distribution until we reach $m_p$ different Paulis.} until we fill up $m_p$ different Paulis and assign the estimator naively as
\begin{equation}
    \widehat c_i
    =N_i\sign(c_i)\frac{W}{m_s}\,,
    \label{eq:app_enhanced_one_norm}
\end{equation}
then $N_i$ and $m_s$ are correlated. We can see that this estimator is biased by a simple second-order expansion of the expectation value of the ratio $\Ebb[N_i/m_s]$
\begin{align}
    \Ebb[\widehat c_i] = & \sign(c_i)\,W\,\Ebb\!\left[\frac{N_i}{m_s}\right]\\
    \approx &  \sign(c_i) W\left(\frac{\Ebb[N_i]}{\Ebb[m_s]} - \frac{\Covar[N_i, m_s]}{\Ebb[m_s]^2} +  \frac{\Ebb[N_i]\Var[m_s]}{\Ebb[m_s]^3}\right)\\
    = &  \sign(c_i) W\left(p_i -\frac{\Covar[N_i,m_s]}{\Ebb[m_s]^2} +p_i\frac{\Var[m_s]}{\Ebb[m_s]^2}\right) 
    \label{eq:app_stopping_expectation}
\end{align}
where in the second line we used that the expectation value of a ratio of two different random variables can be approximated as (see \cite{vanKempen2000meanvariance} which is Theorem 2 of \cite{hurt1976asymptotic} applied to $g(x, y) = x/y$)
\begin{align}
    \Ebb\left[ \frac{x(t)}{y(t)} \right] \approx  \frac{\Ebb[x(t)]}{\Ebb[y(t)]} - \frac{\Covar[x(t), y(t)]}{\Ebb[y(t)]^2} + \frac{\Ebb[x(t)]\Var[y(t)]}{\Ebb[y(t)]^3}\,,
\end{align}
and in Eq.~\eqref{eq:app_stopping_expectation} we used Wald's identity, which states that if $w_1,w_2,\ldots$ are i.i.d.\ with $\Ebb[\abs{w}]<\infty$ and $N$ is a stopping time for the sequence of draws with $\Ebb[N]<\infty$, then $\Ebb[\sum_{k=1}^{N} w_k] = \Ebb[w]\Ebb[N]$. Here $m_s$ is the first draw at which $m_p$ distinct Pauli strings have appeared, which is such a stopping time. We use it to write $\Ebb[N_i] = \Ebb[\sum_j^{m_s} \mathbb{I}_j(i)] = \Ebb[m_s]\Ebb[\mathbb{I}_i]$, where $\mathbb{I}_j(i)$ is an indicator function that is equal to $1$ when the Pauli $P_i$ is sampled at the $j$-th sample and $0$ otherwise, with $\Ebb[\mathbb{I}_j(i)] = p_i$. 

We thus conclude from Eq.~\eqref{eq:app_stopping_expectation} that the average of the estimator is  biased by 
\begin{equation}
    \sign(c_i) W\left(-\frac{\Covar[N_i,m_s]}{\Ebb[m_s]^2} +p_i\frac{\Var[m_s]}{\Ebb[m_s]^2}\right) \, .
\end{equation}

The random stopping time therefore trades exact unbiasedness for an exact number of distinct Pauli strings. We do not use it for the unbiased benchmarks. In the following section we introduce comb resampling and semi-deterministic resampling, two methods that come close to achieving the desired number of stored Paulis when resampling without inducing bias or requiring an exceedingly long time to sample.

\subsubsection{One-norm resampling with bias}
\label{app:one_norm_resampling_bias}

Before turning to the comb schemes, we describe a modification of the resampling step of Appendix~\ref{app:one_norm_resampling} that trades exact unbiasedness for a reduction in variance. If a propagated Pauli sum is peaked a modest number of coefficients stand out, while the bulk of the one-norm is spread over very many small ones. Each of those small coefficients contributes little to the expectation value on its own, and its own sampling variance $\Var[\widehat c_i]=\abs{c_i}(W-\abs{c_i})/m_s$ is correspondingly small, but together they carry most of the one norm and therefore dominate the variance of the protocol. It is therefore often useful to discard the smallest coefficients before resampling. The one-norm resampling with bias method can be decomposed into two steps:
\begin{enumerate}
    \item Discard all coefficients satisfying $\abs{c_i}\leq b$, for a chosen threshold $b >0$.
    \item Perform one-norm resampling on the remaining coefficients.
\end{enumerate}

The first step is deterministic, so it is convenient to absorb it into the coefficients themselves. Let
\begin{align}
    \widetilde c_i
    =c_i\:
      \mathbf 1\!\left(\abs{c_i}>b\right)
      \label{eq:app_truncated_coefficient}
\end{align}
be the truncated coefficients. Then one-norm resampling with bias can be considered as one-norm resampling on the modified coefficients $\widetilde c_i$, that is, on the modified Pauli sum $\widetilde O = \sum_i \widetilde c_iP_i$.

We can similarly define the modified one-norm and the modified probabilities
\begin{align}
    \widetilde W &=\sum_i |\widetilde c_i| = \sum_i |c_i| \mathbf 1\!\left(\abs{c_i}>b\right)
\label{eq:app_truncated_norm}\\
    \widetilde p_i
    &=\frac{|\widetilde c_i|}{\widetilde W} =\frac{ |c_i|}{\widetilde W} \mathbf 1\!\left(\abs{c_i}>b\right)\,,
    \label{eq:app_truncated_probability}
\end{align}
which replace Eq.~\eqref{eq:app_one_norm_probability}. Every statement of Appendix~\ref{app:one_norm_resampling} then carries over under the substitutions $c_i\to\widetilde c_i$, $W\to\widetilde W$ and $p_i\to\widetilde p_i$, the only structural difference being that the discarded Pauli strings have $\widetilde p_i=0$ and can never be drawn. In particular, the estimator can be written similarly as
 \begin{align}
    \widehat c_i
    =N_i\:\sign(c_i)\frac{\widetilde W}{m_s}\,,
    \label{eq:app_truncated_one_norm_estimator}
\end{align}
where $(N_1,N_2,\ldots)\sim\operatorname{Multinomial}(m_s,\{\widetilde p_i\})$. Contrary to the one-norm resampling, this estimator is biased with respect to $O$, because small coefficients will be deterministically discarded, but is unbiased with respect to the truncated sum $\widetilde O$. Concretely, using $\Ebb[N_i]=m_s\widetilde p_i$ we find
\begin{align}
    \Ebb[\widehat c_i] &= \sign(c_i)\frac{\widetilde W}{m_s} m_s \widetilde p_i = \sign(c_i)\widetilde W\frac{|c_i|}{\widetilde W} \mathbf 1\!\left(\abs{c_i}>b\right) = \sign(c_i) |c_i|\mathbf1\!\left(\abs{c_i}>b\right) \nonumber \\
    &= c_i \:\mathbf1\!\left(\abs{c_i}>b\right) = \widetilde c_i\, .
    \label{eq:app_truncated_mean_local}
\end{align}
 Subtracting $c_i$ from Eq.~\eqref{eq:app_truncated_mean_local}, the error made on each coefficient is
\begin{align}
    \Ebb[\widehat c_i] - c_i = -c_i\:\mathbf 1\!\left(\abs{c_i}\leq b\right)\,,
    \label{eq:app_truncated_coefficient_bias}
\end{align}
that is, exactly the coefficient that was thrown away.

Similarly, if $|c_i| > b$, $\widehat c_i$ has the same variance as the unbiased case with $W$ replaced by $\widetilde{W}$, and if $|c_i| \leq b$, $\widehat c_i = 0$ is constant, therefore

\begin{align}
    \Var[\widehat c_i] = \frac{\widetilde W}{m_s} |c_i|\left(1 - \frac{\abs{c_i}}{\widetilde W}\right)\:\mathbf1\!\left(\abs{c_i}>b\right)
\end{align}

\paragraph*{Choosing the threshold.}
An absolute threshold held fixed throughout the propagation does not adapt to the scale of the coefficients, which changes by orders of magnitude as the Pauli sum grows. A threshold relative to the one-norm, $b=\epsilon W$, does adapt, but discards too much once the distribution flattens out and the individual coefficients become small compared to $W$. We therefore take the threshold relative to the largest coefficient,
\begin{align}
    b = \epsilon_{\mathrm{rel}}\max_i|c_i|\,,\qquad \epsilon_{\mathrm{rel}}>0\,,
    \label{eq:app_relative_threshold}
\end{align}
which follows the scale of the distribution at every step, keeping the coefficients that carry its meaningful part and throwing away the rest.

\subsubsection{Systematic or comb resampling}\label{app:comb_sampling}

\begin{figure}
     \centering
     \includegraphics[width=0.95\textwidth]{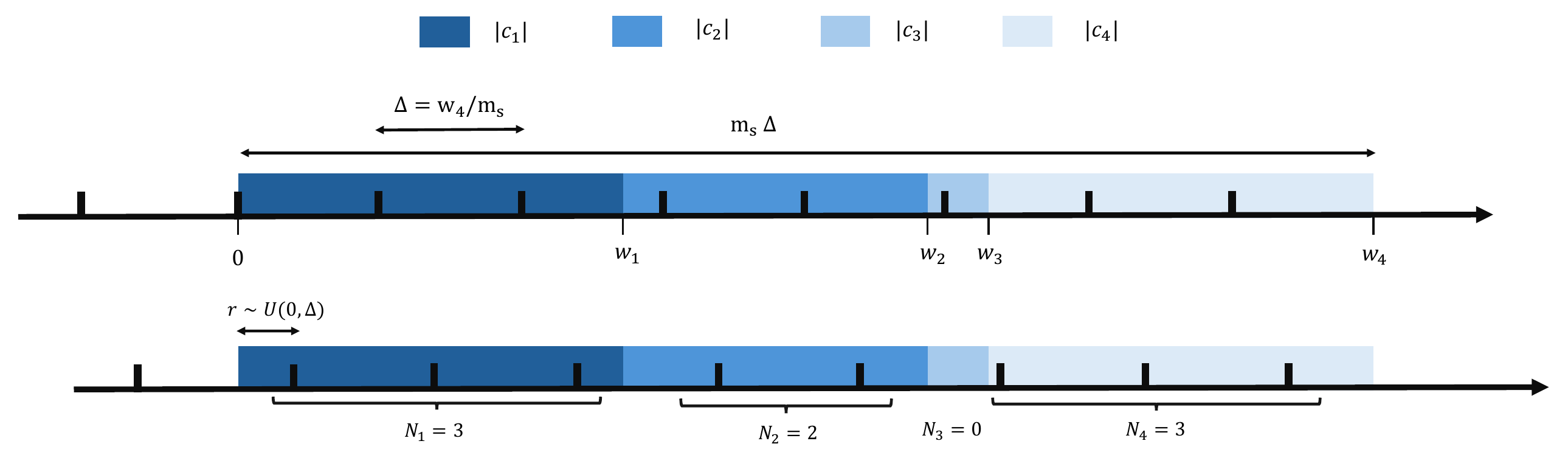}

     \includegraphics[width=0.95\textwidth]{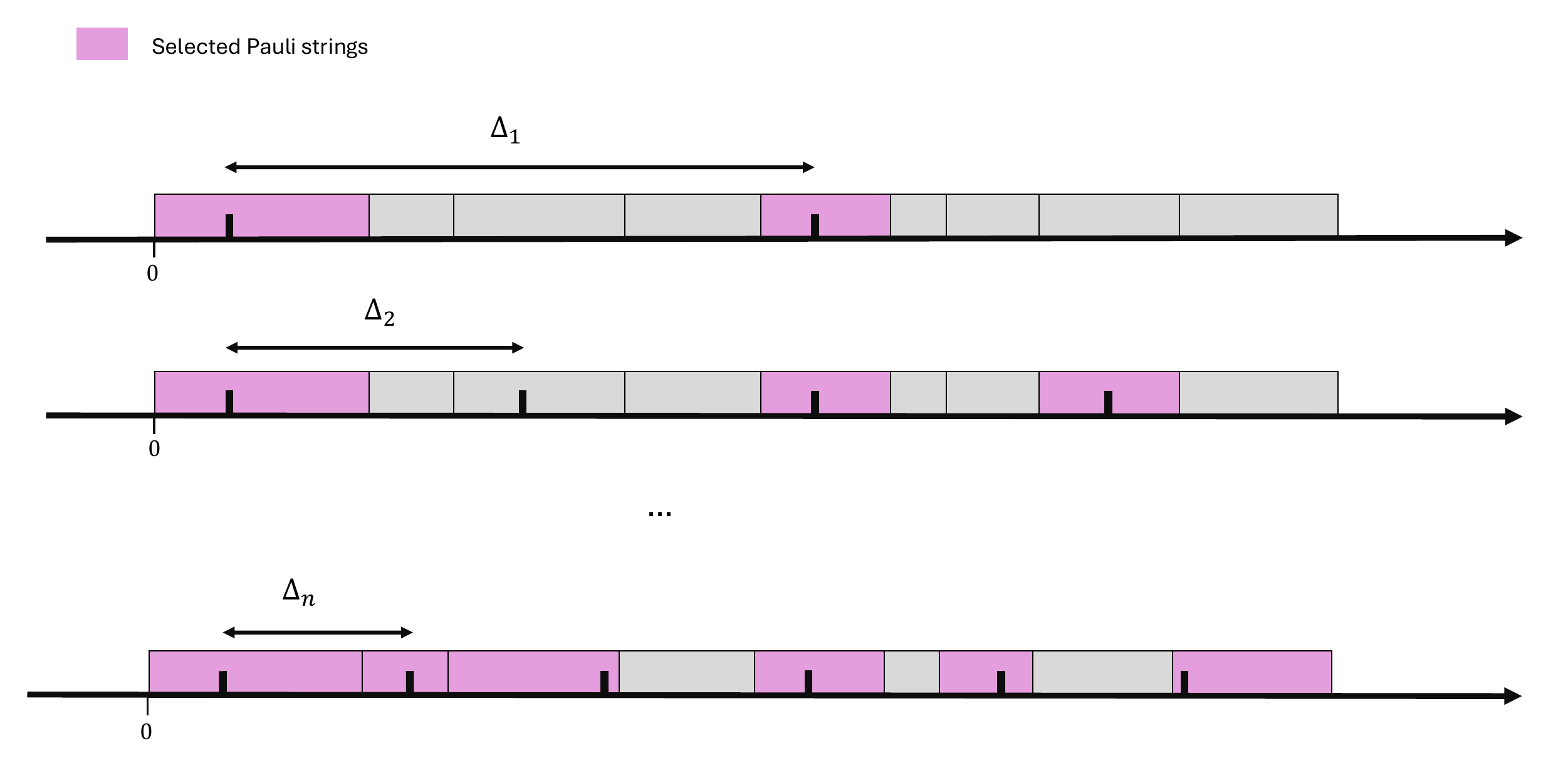}
     \caption{(Top) Illustration of the comb resampling principle. Given absolute value coefficients $|c_i|$s with cumulative one norms $w_i = \sum_{j=1}^{i} |c_j|$ and a number of samples $m_s$, one superposes a grid of width $\Delta = w_{N_O}/ m_s$. The sampling procedure consists first of shifting the grid by a random offset $r$ uniform in $[0, \Delta)$ and then assigning the number $N_i$ associated to the coefficient $c_i$ by counting the number of teeth falling in the interval $[w_{i-1}, w_i)$. (Bottom) Illustration of the calibration process. $\Delta$ is updated and decreases until one gets the desired expected number of Pauli strings.}
     \label{fig:comb_sampling}
\end{figure}

To circumvent the issues created by independent one-norm sampling we use systematic resampling~\cite{kitagawa1996monte}, which we call comb sampling. This method replaces $m_s$ independent random numbers with a more efficient method. Recall the Pauli sum $O=\sum_{i=1}^{N_O}c_iP_i$, with one-norm
$W=\sum_{i=1}^{N_O}\abs{c_i}$. We can define the cumulative weights belonging to the coefficient distributions as
\begin{equation}
    w_i=\sum_{j=1}^{i}\abs{c_j},
    \qquad w_0=0,
    \qquad w_{N_O}=W.
    \label{eq:app_cumulative_weights}
\end{equation}

Instead of drawing $m_s$ independent samples from this distribution, \textit{comb sampling} places a regular grid (a comb) with spacing $\Delta$ and a single random offset $r\sim\operatorname{Uniform}[0,\Delta)$ over the cumulative weights, and counts the number of comb teeth that land in each interval $[w_{i-1},w_i)$. This is illustrated in Fig.~\ref{fig:comb_sampling} (top). That number can be computed for any index $i$ via
\begin{equation}
    N_i=
    \left\lfloor\frac{w_i-r}{\Delta}\right\rfloor
    -
    \left\lfloor\frac{w_{i-1}-r}{\Delta}\right\rfloor\,,
    \label{eq:app_comb_counts}
\end{equation}
where the random variable $r$ determines $N_i$. Note that different $N_i$ are not independent. The resampled coefficient is then
\begin{equation}
    \widehat c_i
    =\sign(c_i)N_i\Delta .
    \label{eq:app_comb_coefficient}
\end{equation}
The coefficients are already merged, only zero-valued coefficients need to be truncated, i.e. Pauli strings with $N_i=0$ are removed from the resampled sum.  

We can see that this estimator is unbiased by understanding how $N_i$ behaves as a random variable. While $N_i$ is computed directly from Eq.~\eqref{eq:app_comb_counts}, its behaviour is easier to read off after massaging that expression. Let us start with the first term of Eq.~\eqref{eq:app_comb_counts},
\begin{align}\label{eq:construct_N}
    \left\lfloor\frac{w_i-r}{\Delta}\right\rfloor = & \left\lfloor\frac{\abs{c_i}}{\Delta} + \frac{w_{i-1}-r}{\Delta}\right\rfloor   \\
    = & \left\lfloor\frac{w_{i-1}-r}{\Delta}\right\rfloor + \left\lfloor\frac{\abs{c_i}}{\Delta} \right\rfloor + \left\lfloor \left\{\frac{\abs{c_i}}{\Delta} \right\} + \left\{ \frac{w_{i-1}-r}{\Delta}\right\} \right\rfloor
\end{align}
where in the first equality we used that $w_i = \abs{c_i} + w_{i-1}$ by Eq.~\eqref{eq:app_cumulative_weights}, and that $x = \lfloor x \rfloor + \{x\}$, i.e.\ $\{x\}$ denotes the non-integer part of $x$. We can rearrange Eq.~\eqref{eq:construct_N} to find
\begin{align}
    \left\lfloor\frac{w_i-r}{\Delta}\right\rfloor - \left\lfloor\frac{w_{i-1}-r}{\Delta}\right\rfloor = \left\lfloor\frac{\abs{c_i}}{\Delta} \right\rfloor + \left\lfloor \left\{\frac{\abs{c_i}}{\Delta} \right\} + \left\{ \frac{w_{i-1}-r}{\Delta}\right\} \right\rfloor
\end{align}
which means that, recalling Eq.~\eqref{eq:app_comb_counts}, we can write
\begin{align}\label{eq:N_i_asbernoully}
    N_i = \left\lfloor\frac{\abs{c_i}}{\Delta} \right\rfloor + I_i\,, \quad I_i = \left\lfloor \left\{\frac{\abs{c_i}}{\Delta} \right\} + \left\{ \frac{w_{i-1}-r}{\Delta}\right\} \right\rfloor\,,
\end{align}
where the $I_i$ are random variables taking the value $0$ or $1$, which are, crucially, not independent of each other, since they are all driven by the single random offset $r$.

Before analysing the unbiasedness of our estimator $\widehat c_i$ we need the expectation value of $I_i$. We can compute it by noting that $u = r/\Delta$ is a uniform random variable, $u \sim \operatorname{Uniform}[0,1)$, and writing $\phi_{i-1} = \left\{ w_{i-1}/\Delta \right\}$ for the non-integer part of the cumulative weight, so that $\left\{ w_{i-1}/\Delta - u\right\} = \left\{ \phi_{i-1} - u\right\}$,
\begin{align}\label{eq:expval_random_I_i_N_i}
    \Ebb[I_i] =& \int_0^1 \left\lfloor \left\{\frac{\abs{c_i}}{\Delta} \right\} + \left\{ \phi_{i-1}- u\right\} \right\rfloor \dd u\\
    =& \int_0^{\phi_{i-1}} \left\lfloor \left\{\frac{\abs{c_i}}{\Delta} \right\} +  \phi_{i-1}- u \right\rfloor \dd u + \int_{\phi_{i-1}}^1 \left\lfloor \left\{\frac{\abs{c_i}}{\Delta} \right\} + 1+ \phi_{i-1}- u \right\rfloor \dd u\\
    =& \int_0^{\phi_{i-1}} \left\lfloor \left\{\frac{\abs{c_i}}{\Delta} \right\} +  x \right\rfloor \dd x + \int_{\phi_{i-1}}^1 \left\lfloor \left\{\frac{\abs{c_i}}{\Delta} \right\} +  x \right\rfloor \dd x \\
    = & \int_0^1 \left\lfloor \left\{\frac{\abs{c_i}}{\Delta} \right\} +  x \right\rfloor \dd x = \int_{1-\left\{{\abs{c_i}}/{\Delta} \right\}}^1  \dd x = \left\{\frac{\abs{c_i}}{\Delta} \right\}\,,
\end{align}
where in the second equality we used that $\left\{ \phi_{i-1} - u\right\} = \phi_{i-1}-u$ for $u<\phi_{i-1}$ and $\left\{ \phi_{i-1} - u\right\} = 1 + \phi_{i-1}-u$ for $u>\phi_{i-1}$, since $\phi_{i-1}\in[0,1)$ and $u\in[0,1)$. In the third equality we used the changes of variables $x = \phi_{i-1}-u$ and $x = 1+\phi_{i-1}-u$ for the first and the second integral respectively, each of which maps the interval of integration onto the one shown. In the two last equalities we first regrouped the two integrals into one, and then used that $\left\lfloor \left\{\abs{c_i}/\Delta \right\} +  x \right\rfloor = 1$ for $x > 1- \left\{\abs{c_i}/\Delta \right\}$ and $0$ otherwise.

Now we are equipped to compute the expectation value of our estimator $\widehat c_i$ in Eq.~\eqref{eq:app_comb_coefficient}
\begin{align}
    \Ebb[\widehat c_i] =&  \sign(c_i)\Delta\Ebb[N_i]\\
     = & \sign(c_i)\Delta \left( \left\lfloor\frac{\abs{c_i}}{\Delta} \right\rfloor + \Ebb[I_i] \right)\\
     = & \sign(c_i)\Delta \left( \left\lfloor\frac{\abs{c_i}}{\Delta} \right\rfloor +  \left\{\frac{\abs{c_i}}{\Delta} \right\}\right) = \sign(c_i)\Delta \frac{\abs{c_i}}{\Delta} = c_i\,,
\end{align}
where in the first equality we used Eq.~\eqref{eq:N_i_asbernoully}, and in the second equality we used Eq.~\eqref{eq:expval_random_I_i_N_i}. The estimator is therefore unbiased.

We can also compute the variance. Since $N_i = \left\lfloor\frac{\abs{c_i}}{\Delta} \right\rfloor + I_i$, we have $\Var[N_i] = \Var[I_i] = \left\{\frac{\abs{c_i}}{\Delta} \right\}\left(1-\left\{\frac{\abs{c_i}}{\Delta} \right\}\right)$.
Therefore

\begin{align}
    \Var[\widehat c_i] = \Delta^2 \Var[I_i] = \Delta^2  \left\{\frac{\abs{c_i}}{\Delta} \right\}\left(1-\left\{\frac{\abs{c_i}}{\Delta} \right\}\right)
\end{align}

Comb resampling has two implementation advantages. The counts can be
generated while walking through the cumulative distribution, so repeated
outcomes are already aggregated and need not be merged in a second pass. The method also admits efficient parallel implementations on CPUs and
GPUs~\cite{murray2016parallel}. \\

\paragraph*{Calibrating the spacing.}
We are now left with the question of how to pick the comb spacing $\Delta$. The natural spacing $W/m$ fixes the number of teeth, not the number of unique
Pauli strings. Any interval, $c_i$, wider than $\Delta$ can catch more than one tooth, and so contribute multiple Pauli strings but only one unique one. 
It follows that a
peaked distribution (i.e., one with many large $c_i$) again produces fewer than $m$ stored Pauli strings. 
The
expected number of unique Pauli strings at spacing $\Delta$ is 
\begin{equation}
    F(\Delta)
    =\sum_{i=1}^{N_O}\min\!\left(1,\frac{\abs{c_i}}{\Delta}\right)\,,
    \label{eq:app_comb_objective}
\end{equation}
where a coefficient larger than the comb spacing is certain to receive at least one tooth, while a smaller one survives only with probability $\abs{c_i}/\Delta$.

Calibrated comb sampling selects $\Delta$ so that 
$F(\Delta)\simeq m$. To do so, we employ a fixed point iteration, updating $\Delta$ until convergence with the following rule
\begin{equation}
    \Delta_{k+1}\leftarrow\Delta_k\frac{F(\Delta_k)}{m}.
    \label{eq:app_comb_update}
\end{equation}
Since $F(\Delta)\leq W/\Delta$, the calibrated spacing satisfies $\Delta\leq W/m$, and the comb lays $m_s=W/\Delta\geq m$ teeth. Calibration therefore increases the effective number of draws $m_s$ while bringing $m_p$ closer to $m$.

With the following proposition, we will show that this procedure converges given that $\Delta$ is initialized at $\max_i |c_i|$ and some technical assumptions. $\Delta$ decreases until the expected number of unique Pauli strings $F(\Delta)$ reaches $m$. This is illustrated in Fig.~\ref{fig:comb_sampling} (bottom).

\begin{proposition}
\label{prop:calibration_convergence}
Consider the fixed point iteration initialized at $\Delta_0 = \max_i |c_i|$ updating $\Delta$ with 
\begin{equation}
    \Delta_{k+1}\leftarrow g(\Delta_k)
    \label{eq:app_comb_update2}
\end{equation}
where $g(\Delta) = \Delta \displaystyle \frac{F(\Delta)}{m}$. Assume that $\max_i |c_i| > \displaystyle \frac{\sum_i |c_i|}{m}$ and that there exists $\Delta > 0$ such that $\left|\{i, |c_i| > \Delta\}\right| > m$.
Then the fixed point iteration converges to $\Delta_{\text{opt}}$ such that $F(\Delta_{\text{opt}}) = m$.
    
\end{proposition}
\begin{proof}
We will first establish that $g(\Delta)$ is continuous and piecewise affine with decreasing slopes. We have by using the expression of $F$ in Eq.~\eqref{eq:app_comb_objective}

\begin{equation}
    g(\Delta) = \frac{\Delta}{m} F(\Delta) = \frac{\Delta}{m}\left(\sum_{i, |c_i| > \Delta} 1 + \frac{1}{\Delta} \sum_{i, |c_i| \leq \Delta} |c_i|\right) = \frac{1}{m}\left(a(\Delta)\:\Delta + b(\Delta)\right)
\end{equation}
where we defined
\begin{align}
    a(\Delta) &= \left|\{i, |c_i| > \Delta\}\right|,\\
    b(\Delta) &= \sum_{i, |c_i| \leq \Delta} |c_i|.
\end{align}

Some useful facts that will be used later are
\begin{align}
    a(\Delta_0) &= a(\max_i |c_i|) = 0\\
    b(\Delta_0) &= \sum_i |c_i|
\end{align}
then by our assumption,
\begin{equation}
    g(\Delta_0) = \sum_i |c_i|/m < \Delta_0
    \label{eq:assumption_delta_0}
\end{equation}

Assume that the $|c_i|$s are ordered such that $|c_1| \leq |c_2| \leq \dots \leq |c_{N_O}|$.
$a(\Delta)$ and $b(\Delta)$ are constant in the intervals $]|c_i|, |c_{i+1}|[$ where $|c_i|\neq  |c_{i+1}|$. Then $g(\Delta)$ is affine in those same intervals and $a(\Delta)$ is a decreasing function so the slopes are decreasing. We now have to prove that $g(\Delta)$ is continuous in all points $|c_{i}|$s. Let $i_0$ be an index, and let us study the limits of $g(\Delta)$ in $|c_{i_0}|^+$ and $|c_{i_0}|^-$. We study $m g(\Delta)$ for lighter notations. We have

\begin{align}
    \text{if }|c_{i_0}|<\Delta <|c_{i_0+1}|, \: m\:g(\Delta) &= \left|\{i, |c_i| > |c_{i_0}|\}\right| \Delta + \sum_{i, |c_i| \leq |c_{i_0}|}|c_i|  \\
    &= \left|\{i, |c_i| > |c_{i_0}|\}\right| \Delta + \sum_{i, |c_i| < |c_{i_0}|}|c_i|+ \sum_{i, |c_i| = |c_{i_0}|}|c_i|\\
    &= \left|\{i, |c_i| > |c_{i_0}|\}\right| \Delta + \sum_{i, |c_i| < |c_{i_0}|}|c_i|+ \left|\{i, |c_i| = |c_{i_0}|\}\right| |c_{i_0}|\\
    \lim_{\Delta \rightarrow |c_{i_0}|^+}m\:g(\Delta) &= \left|\{i, |c_i| > |c_{i_0}|\}\right| |c_{i_0}| + \sum_{i, |c_i| < |c_{i_0}|}|c_i|+ \left|\{i, |c_i| = |c_{i_0}|\}\right| |c_{i_0}|\\
    &= \left|\{i, |c_i| \geq |c_{i_0}|\}\right| |c_{i_0}| + \sum_{i, |c_i| < |c_{i_0}|}|c_i|\\
    \text{if }|c_{i_{0}-1}|<\Delta <|c_{i_0}|, \: m\:g(\Delta) &= \left|\{i, |c_i| \geq |c_{i_0}|\}\right| \Delta + \sum_{i, |c_i| < |c_{i_0}|}|c_i|  \\
    \lim_{\Delta \rightarrow |c_{i_0}|^-}m\:g(\Delta) &= \left|\{i, |c_i| \geq |c_{i_0}|\}\right| |c_{i_0}|+ \sum_{i, |c_i| < |c_{i_0}|}|c_i|
 \end{align}

 We have proved that $\lim_{\Delta \rightarrow |c_{i_0}|^+}g(\Delta) = \lim_{\Delta \rightarrow |c_{i_0}|^-}g(\Delta)$ therefore $g$ is continuous.

 We will now show that under our assumptions, $g(\Delta)$ admits a fixed point $\Delta_{\text{opt}} > 0$ and that the procedure converges to this fixed point. By the contraction mapping principle \cite{hoffman2018numerical}, it is enough to show that $g$ has a Lipschitz constant strictly below 1 on the interval $[\Delta_{\text{opt}}, \:\Delta_0]$.\\

 Let us first establish that $\Delta_{\text{opt}}$ exists. By ~\eqref{eq:assumption_delta_0}, we have that $ \Delta_0 - g(\Delta_0) > 0$
 We will show that there exists $\Delta > 0$ such that $ \Delta-g(\Delta) < 0$. Because $\Delta \mapsto  \Delta-g(\Delta)$ is continuous, by the intermediate values theorem it implies that there exists  $\Delta_{\text{opt}} > 0$ such that $\Delta_{\text{opt}} = g(\Delta_{\text{opt}})$.
 We have
 \begin{align}
     \Delta - g(\Delta) = \Delta - \frac{a(\Delta)}{m}\Delta - \frac{b(\Delta)}{m} = \left(1 -\frac{a(\Delta)}{m} \right)\Delta - \frac{b(\Delta)}{m}
 \end{align}
 By our assumption, we know that there exists $\Delta > 0$ such that $\left|\{i, |c_i| > \Delta\}\right| > m$, then for such a $\Delta$, $1 -\frac{a(\Delta)}{m} < 0$, and the second term in $\Delta - g(\Delta)$ is non positive, which proves that there exists $\Delta > 0$ such that $\Delta - g(\Delta) < 0$.\\

 We will now show that the Lipschitz constant of $g$ is strictly inferior to 1 in the interval $[\Delta_{\text{opt}}, \:\Delta_0]$. We know that $g$ is piecewise affine continuous, then the Lipschitz constant is upper bounded by the maximum of the slopes. Assume that there exists $\Delta_1$ such that $g$ does not cross the identity line in the interval $[\Delta_1,\: \Delta_0]$ (so we have $g(\Delta_1) < \Delta_1$) and that the slope at $\Delta_1$ is superior to 1. Then for all $\Delta < \Delta_1$, we have
 \begin{equation}
      g(\Delta_1) - g(\Delta) \geq \ell(\Delta_1 - \Delta)
 \end{equation}
 where $\ell$ is the lower Lipschitz constant, lower bounded by the lowest slope, so $\ell \geq 1$

 Then we have
 \begin{align}
      g(\Delta_1) - g(\Delta) &\geq \Delta_1 - \Delta\\
      \Delta - g(\Delta) &\geq \Delta_1 - g(\Delta_1) > 0
 \end{align}
 We have then shown that for all $\Delta \leq \Delta_1$, $g(\Delta) < \Delta$ so $g$ can never cross the identity line.

 Because we have shown that $g$ \textsl{does} cross the identity line at some point, it implies that the Lipschitz constant is strictly below 1 in the interval $[\Delta_{\text{opt}}, \Delta_0]$.

\end{proof}

Proposition \ref{prop:calibration_convergence} does not give a bound on the number of necessary iterations to converge, it depends in general on the distribution of coefficients.
The process usually reaches the target tolerance within roughly five iterations. The
calibration can approximately double the time spent resampling because each
iteration requires another pass over the coefficient array. It is worthwhile
when resampling dominates the error and less attractive when gate application
or merging dominates the runtime.

\subsubsection{Semi-deterministic comb resampling}
\label{app:semi_deterministic_comb_resampling}

There is little value in repeatedly sampling a coefficient that is almost
certain to survive. Semi-deterministic comb resampling therefore divides the
Pauli sum into a deterministic head and a stochastic tail. Choose a threshold
$\tau$ and define $D$ as the set of indices whose coefficients are deterministically propagated, and $T$ as its complement
\begin{align}
    D=\{i:\abs{c_i}>\tau\}\,,\quad T=D^c=\{i:\abs{c_i}\leq\tau\}\,,
\end{align}
we can also call $m_{\rm det}$ the number of distinct Paulis supported on $D$, and $m_{\rm stoch}$ the number of distinct Paulis that the stochastic part must compress down to
\begin{align}
    m_{\mathrm{det}}=\abs{D}\,, \quad m_{\mathrm{stoch}}=m-m_{\mathrm{det}}\,.
\end{align}
Finally we can define the one-norm of the stochastic coefficients as well as the comb size as 
\begin{align}
    W_{\mathrm{stoch}}=\sum_{i\in T}\abs{c_i}\,,\quad \Delta_{\mathrm{stoch}}=\frac{W_{\mathrm{stoch}}}{m_{\mathrm{stoch}}}\,.
\end{align}

The coefficients from the resampled Pauli sum are
\begin{equation}
    \widehat c_i=
    \begin{cases}
        c_i, & i\in D,\\
        \sign(c_i)N_i\Delta_{\mathrm{stoch}},
        & i\in T,
    \end{cases}
    \label{eq:app_semidet_estimator}
\end{equation}
where the counts $N_i$ for the smaller $|c_i|$ are generated via comb sampling described in the previous Appendix~\ref{app:comb_sampling}. The larger coefficients are exact, while the coefficients in the tail remain unbiased.

A safe and cheap choice sets $\tau$ from $m$ alone,
\begin{equation}
    \tau=\frac{W}{m}\,,
    \label{eq:app_semidet_threshold}
\end{equation}
which guarantees that fewer than $m$ terms are kept deterministically, since $k$ coefficients above $W/m$ would already carry a one norm of at least $kW/m$. 

Once the large coefficients are removed, most Pauli strings in the tail receive either zero or one tooth. The uncalibrated tail comb therefore lands close to the desired
number of stored Pauli strings, which removes most of the benefit of the
iterative calibration. One may lower $\tau$ to use still more of the available memory, but percentile-based thresholds require either an expensive global selection or a
noisy estimate from a subset of coefficients.

Residual resampling~\cite{liu1998sequential} is a related
construction. It assigns the integer part of every ideal multiplicity
deterministically and samples only the residual probabilities. It likewise
concentrates randomness in the part of the Pauli sum that cannot be allocated
exactly. Which variant is fastest depends on the coefficient distribution and
the target hardware, particularly on whether scans, sorting, and random-number
generation are cheap relative to gate application.

\section{Theoretical analysis of the errors of Monte Carlo Pauli Propagation.}
\label{app:error_accumulation}
In this section we present different results and theorems that help explain the numerical results. 

\subsection{Summary of results and overview of proof strategies}

We start by showing that the error of the simple path sampling method can blow up exponentially. We do so by analyzing an example of a Trotter circuit of a TFIM  (Example \ref{ex:variance_path_sampling_lower_bound}) where we give an exponential lower bound of the variance.
We then derive a general formulation of the mean square error of the estimator obtained after sampling in Lemma \ref{lemma:general_mse}. This lemma will be the starting point of all the future error analysis.
In Theorem \ref{thm:one_norm_sampling_error_unbiased}, we derive an upper bound of the mean square error of the one-norm sampling estimator without any truncation. It grows quadratically with the one norm of the Pauli sum at each layer.
In Theorem \ref{thm:one_norm_sampling_error_bias}, we show how the bound of Theorem \ref{thm:one_norm_sampling_error_unbiased} is modified when one introduces deterministic truncation before sampling. The error can be upper bounded by the same bound as in the case without truncation, but with a modified one norm, to which is added a bias term depending on the truncation. In Result~\ref{result:app_truncated_mse_model}, we propose a phenomenology model to explain the results observed in Fig. \ref{fig:resampling_increase_bias}, by making assumptions on the behavior of the one-norm and the discarded one-norm after truncation. We then explain the errors of comb sampling and comb sampling with truncation observed in Fig. \ref{fig:memory_utilization} by deriving similar error bounds in results \ref{result:comb_sampling_error_unbiased} and \ref{result:comb_sampling_error_bias}. We present these as results and not as theorems because we make technical assumptions that we cannot guarantee will universally hold.
All the bounds derived so far are controlled by the number of raw draws $m_s$, whereas the memory only holds the number $m_p$ of distinct Pauli strings. The two differ because repeated draws are merged, and how much they differ depends on how peaked the coefficient distribution is. We analyze this in Result ~\ref{result:app_effective_draws}.

\subsection{Variance of the path sampling estimator}
\label{app:variance_path_sampling}
In this section, we will illustrate in an example why the simple path sampling estimator produces an error that scales exponentially with depth by computing a lower bound.

\begin{example}[Exponential lower bound in the variance of path-sampling]
\label{ex:variance_path_sampling_lower_bound}
    Consider the task of estimating the expectation value $ f = \Tr[\ketbra{0} U(\thv) Z_1 U(\thv)^\dagger]$, where the unitary matrix corresponds to local $X$ rotations and two-qubit $ZZ$ rotations, the $ZZ$ having a locally connected topology denoted by a set of edges $E$ (like a square lattice) in which qubit $1$ has degree $d\geq1$,
    \begin{align}
        U(\thv) = \left(\prod_{i=1}^{n}R_X^{(i)}(\theta_{1})\prod_{(i,j)\in E}R_{ZZ}^{(i,j)}(\theta_{2})\right)^{L} 
        \label{eq:example_circuit_ising}
    \end{align}
    where we note that all the $X$ rotations share the same parameter $\theta_1$, as well as $\theta_2$ for the $ZZ$, with $0<\theta_1,\theta_2 < \pi/4$, and that $L\geq2$. Then if $\widehat f$ is the path-sampling estimator of $f$ built from $N$ independent paths, as described in Appendix~\ref{app:path_sampling}, and $N$ is held fixed, there exist a constant $\gamma>1$ and a depth $L_0$ such that the variance of $\widehat f$ is lower bounded as
    \begin{align}
        \Var[\widehat f]\geq \gamma^L
        \label{eq:example_variance_claim}
    \end{align}
    for every $L\geq L_0$. One may take any $\gamma$ with $1<\gamma<A_2^{2d}$, where $A_2=\abs{\cos\theta_2}+\abs{\sin\theta_2}$; the rate is therefore set by the degree $d$ of qubit $1$.
\end{example}

\begin{proof}

 We are trying to estimate $\Tr[\ketbra{0} U Z_1 U^\dagger]$ of a system of $n$ qubits with path sampling, as described in Appendix~\ref{app:path_sampling}. 

Recall that path sampling obtains $\widehat f$ the estimator of $f$, by propagating one Pauli, that gets randomly chosen at each new unitary according to the magnitude of its transition amplitudes, as described in Eq.~\eqref{eq:app_path_probability}. Recalling Eq.~\eqref{eq:app_path_full_estimator}, we find that the estimator $\widehat f$ can be written as
\begin{align}
    \widehat f = \frac{1}{N}\sum_{i=1}^N \Tr[\widehat c_{P_K}^{(i)} P_K^{(i)} \ketbra{0}]\,,\quad \widehat c_{P_K} = \prod_{k=1}^K \sign{a_{P_{k-1}\to P_k}}A_k(P_{k-1})
    \label{eq:path_estimator_runs}
\end{align}
where the sum on $i$ denotes the different number of runs.

To compute the variance of such estimator we define the notion of Pauli paths. A path can be used to model the evolution of a Pauli along the circuit, and is represented as a vector $\Phi \in \{-1,0,1\}^K$, where $K$ is the total number of gates. We write $\vartheta_k$ for the rotation angle of gate $k$, so that $\vartheta_k\in\{\theta_1,\theta_2\}$ depending on whether gate $k$ is an $R_X$ or an $R_{ZZ}$ rotation; the separate symbol avoids a clash between the gate index and the two rotation-type labels. The entries of $\Phi$ are 

\begin{itemize}
    \item $\Phi_k = 0$ if the gate $k$ commutes with $P_{k-1}$
    \item $\Phi_k = 1$ if $P_{k-1}$ is chosen with probability $\frac{|\cos\vartheta_k|}{|\cos\vartheta_k| + |\sin\vartheta_k|}$
    \item $\Phi_k = -1$ if $i[P_{k-1}, Q_k]/2$ is chosen with probability $\frac{|\sin\vartheta_k|}{|\cos\vartheta_k| + |\sin\vartheta_k|}$
\end{itemize}
where basically we see that a path ``stores'' the history of the Pauli when using the path sampling strategy. We note the final coefficient and Pauli associated to $\Phi$ as $c(\Phi)$ and $P_\Phi$. Because a path is a probabilistic outcome, it can be associated with a probability $p(\Phi)$, and its contribution to the final estimator is $w(\Phi) = c(\Phi)\Tr[P_\Phi \ketbra{0}]$, which is nonzero only if $P_\Phi \in \{ I,  Z\}^n$. We have for a path ending in $\{ I,  Z\}^n$, the $w(\Phi)$ is:
\begin{align}\label{eq:strength_path}
|w(\Phi)| = \prod_{k=1}^{K}A_k(P_{k-1})= \prod_{k=1}^{K}(|\cos\vartheta_{k}|+|\sin\vartheta_{k}|)^{|\Phi_{k}|}
\end{align}
with an associated probability
\begin{align}\label{eq:prob_path}
p(\Phi) = \prod_{k=1}^K \left( \frac{|\cos\vartheta_k|}{|\cos\vartheta_k| + |\sin\vartheta_k|} \right)^{\mathbb{I}\{\Phi_k=1\}} \left( \frac{|\sin\vartheta_k|}{|\cos\vartheta_k| + |\sin\vartheta_k|} \right)^{\mathbb{I}\{\Phi_k=-1\}}
\end{align}
We used the fact that if for an index $k$, $\Phi_k = 0$, then $P_k = P_{k-1}$ deterministically and $A_k(P_{k-1}) = 1$. If $|\Phi_k|=1$, then $A_k(P_{k-1}) = |\cos\vartheta_k| + |\sin\vartheta_k|$ and the corresponding probability depends on the sign of $\Phi_k$.

Let $Y$ denote the contribution of a single sampled path, that is, the random variable that takes the value $w(\Phi)$ when the path $\Phi$ is sampled, so that the estimator of Eq.~\eqref{eq:path_estimator_runs} reads $\widehat f = \frac{1}{N}\sum_{i=1}^N Y_i$ with $Y_1,\ldots,Y_N$ independent copies of $Y$. We use a separate letter here because $W$ denotes the one norm everywhere else in this appendix. The second moment of $Y$ is
\begin{align}
\Ebb[Y^{2}] &= \sum_{\Phi}w(\Phi)^{2}p(\Phi) \nonumber\\
&= \sum_{\Phi\,:\,P_\Phi\in\{I,Z\}^n}|w(\Phi)|\prod_{k=1}^{K}|\cos\vartheta_{k}|^{\mathbb{I}\{\Phi_{k}=1\}}|\sin\vartheta_{k}|^{\mathbb{I}\{\Phi_{k}=-1\}}\,,
\label{eq:second_moment_path_sum}
\end{align}
where we used Eqs.~\eqref{eq:strength_path} and~\eqref{eq:prob_path} to cancel one factor of $A_k(P_{k-1})$ per gate between $w(\Phi)^2$ and $p(\Phi)$. The sum may be restricted to the paths ending in $\{I,Z\}^n$ because every other path has $w(\Phi)=0$, and Eq.~\eqref{eq:strength_path} holds precisely on that set.

We lower bound the second moment of Eq.~\eqref{eq:second_moment_path_sum} by restricting the sum to a particular family of paths, which is legitimate because every term in that sum is nonnegative. Recall that the propagation starts from $P_0 = Z_1$, and that expanding the product of Eq.~\eqref{eq:example_circuit_ising} in the conjugation $U(\thv)Z_1U(\thv)^\dagger$ makes the $R_{ZZ}$ rotations of a given layer act on the running Pauli string before the $R_X$ rotations of the same layer. Since $L\geq2$, the first and the last $R_X$ rotation acting on qubit $1$ are distinct gates, and we may consider the paths satisfying the following conditions:
\begin{enumerate}
    \item At the first noncommuting $R_X$ rotation acting on qubit $1$, the path takes direction $-1$. This transforms $Z_1$ into $Y_1$, up to an irrelevant sign.
    \item At every noncommuting $R_X$ rotation other than the two singled out in conditions 1 and 4, whether it acts on qubit $1$ or on one of its neighbours and whether it precedes or follows them within its layer, the path takes direction $+1$, so that the current Pauli string is left unchanged. We refer to these as the intermediate rotations.
    \item Among the $R_{ZZ}$ rotations acting on qubit $1$, the path takes direction $-1$ an even number of times. The Pauli string at the end is equal to $Y_1 \otimes Q$ where $Q \in \{I, Z\}^{\otimes n-1}$.
    \item At the last $R_X$ rotation acting on qubit $1$, the path again takes direction $-1$. The Pauli string at the end is in $\{I, Z\}^{\otimes n}$.
\end{enumerate}

In each of the remaining $L-1$ layers, qubit $1$ carries $X$ or $Y$ while all of its neighbours carry $I$ or $Z$, so each of the $d$ incident edges hosts a rotation that anticommutes with the string, and therefore one can choose $2k$ times direction $-1$ among $n_{ZZ} = d(L-1)$ possible gates. All other $R_X$ rotations take the direction $1$, and we note their number $r(\Phi)$.

Denote by $ A_a = |\cos\theta_a| + |\sin\theta_a|,\, a \in \{1,2\}$ so that $A_1$ multiplies $|w(\Phi)|$ at every noncommuting $R_X$ rotation and $A_2$ at every noncommuting $R_{ZZ}$ rotation, in Eq.~\eqref{eq:strength_path}.  Collecting the factors of Eq.~\eqref{eq:strength_path} and Eq.~\eqref{eq:prob_path} gate by gate, its contribution to Eq.~\eqref{eq:second_moment_path_sum} is
\begin{align}
    w(\Phi)^{2}p(\Phi) = \left(A_1|\sin\theta_1|\right)^{2}\left(A_1|\cos\theta_1|\right)^{r(\Phi)}A_2^{n_{ZZ}}|\sin\theta_2|^{2k}|\cos\theta_2|^{n_{ZZ}-2k}\,,
    \label{eq:family_path_weight}
\end{align}
where the first factor collects the two direction-$(-1)$ choices of conditions 1 and 4, the second the $r(\Phi)$ direction-$(+1)$ choices of condition 2, and the last three the $n_{ZZ}$ noncommuting $R_{ZZ}$ rotations of condition 3. The one-norm factors satisfy $A_{1,2}>1$ for every $\theta_{1,2}$ that is not a multiple of $\pi/2$. The restriction $0<\theta_1<\pi/4$ is what makes the intermediate factors harmless as well, since then
\begin{align}
    A_1|\cos\theta_1| = \cos^{2}\theta_1 + |\sin\theta_1\cos\theta_1| = \frac{1 + \cos 2\theta_1 + |\sin 2\theta_1|}{2} > 1\,,
    \label{eq:intermediate_x_factor}
\end{align}

so that each of the $r(\Phi)$ intermediate factors is at least one and may be omitted when deriving a lower bound. Since $A_1>1$ as well, Eq.~\eqref{eq:family_path_weight} gives
\begin{align}
    w(\Phi)^{2}p(\Phi) \geq |\sin\theta_1|^{2}A_2^{n_{ZZ}}|\sin\theta_2|^{2k}|\cos\theta_2|^{n_{ZZ}-2k}\,.
    \label{eq:family_path_weight_bound}
\end{align}

Distinct paths of the family differ in which of the $n_{ZZ}$ noncommuting incident $R_{ZZ}$ rotations take direction $-1$, and there are $\binom{n_{ZZ}}{2k}$ ways of selecting $2k$ of them. Keeping only these paths in Eq.~\eqref{eq:second_moment_path_sum} and inserting Eq.~\eqref{eq:family_path_weight_bound} therefore yields
\begin{align}
    \Ebb[Y^{2}] \geq |\sin\theta_1|^{2}A_2^{n_{ZZ}}\sum_{k=0}^{\lfloor n_{ZZ}/2\rfloor}\binom{n_{ZZ}}{2k}|\sin\theta_2|^{2k}|\cos\theta_2|^{n_{ZZ}-2k}\,.
    \label{eq:second_moment_even_sum}
\end{align}
The even part of a binomial expansion resums to
\begin{align}
    \sum_{k=0}^{\lfloor n_{ZZ}/2\rfloor}\binom{n_{ZZ}}{2k}x^{2k}y^{n_{ZZ}-2k} = \frac{(y+x)^{n_{ZZ}} + (y-x)^{n_{ZZ}}}{2}\,,
    \label{eq:even_binomial_identity}
\end{align}
as follows by adding the expansions of $(y+x)^{n_{ZZ}}$ and $(y-x)^{n_{ZZ}}$, in which the odd powers of $x$ cancel pairwise. Setting $x = |\sin\theta_2|$ and $y = |\cos\theta_2|$, we obtain
\begin{align}
    \Ebb[Y^{2}] &\geq |\sin\theta_1|^{2} (|\cos\theta_2|+ |\sin\theta_2|)^{n_{ZZ}} \frac{ (|\cos\theta_2|+ |\sin\theta_2|)^{n_{ZZ}} +  (|\cos\theta_2|- |\sin\theta_2|)^{n_{ZZ}}}{2}
    \label{eq:second_moment_resummed}\\
    &\geq \frac{|\sin\theta_1|^{2}}{2} (|\cos\theta_2|+ |\sin\theta_2|)^{2n_{ZZ}}
    = \frac{|\sin\theta_1|^{2}}{2}\left( |\cos\theta_2|+ |\sin\theta_2|\right)^{2d(L-1)}\,,
    \label{eq:second_moment_lower_bound}
\end{align}
where in the second line we discarded $ (|\cos\theta_2|-|\sin\theta_2|)^{n_{ZZ}}$, and in the last equality we substituted $n_{ZZ} = d(L-1)$. Since $(|\cos\theta_2|+|\sin\theta_2|) > 1$ for $0 < |\theta_{1,2}| < \pi/4$ , the second moment of a single-path contribution grows exponentially with the number of layers $L$.

Finally, path sampling is unbiased by Eq.~\eqref{eq:app_path_unbiased_step}, so $\Ebb[Y] = f$, and $f$ is the expectation value of a Pauli operator, so $|f| \leq 1$. The variance of a single-path contribution therefore obeys
\begin{align}
    \Var[Y] = \Ebb[Y^{2}] - f^{2} \geq \frac{|\sin\theta_1|^{2}}{2}\left( |\cos\theta_2|+ |\sin\theta_2|\right)^{2d(L-1)} - 1\,.
    \label{eq:variance_single_path}
\end{align}
The $N$ paths entering Eq.~\eqref{eq:path_estimator_runs} are independent and identically distributed, hence
\begin{align}
    \Var[\widehat f] = \frac{\Var[Y]}{N} \geq \frac{1}{N}\left[\frac{|\sin\theta_1|^{2}}{2}\left( |\cos\theta_2|+ |\sin\theta_2|\right)^{2d(L-1)} - 1\right]\,.
    \label{eq:variance_path_estimator}
\end{align}
For a fixed number of paths $N$, the right hand side grows exponentially in the circuit depth $L$. Equivalently, for every constant $\gamma$ with $1 < \gamma < A_2^{2d}$ the ratio $(A_2^{2d}/\gamma)^{L}$ diverges, so there exists an $L_0$ such that
\begin{align}
    \Var[\widehat f] \geq \gamma^{L}
    \label{eq:variance_path_exponential}
\end{align}
for every $L \geq L_0$, which is the statement of Eq.~\eqref{eq:example_variance_claim}.
\end{proof}

\subsection{Layered propagation and the exact error recursion}
In this subsection we present a general lemma that can be used to compute the MSE of a propagated observable.

\begin{lemma}[General MSE formula for layered propagation]\label{lemma:general_mse}
    Let a unitary circuit $U$ of $L$ gates be $U=U_L\cdots U_2U_1$ and define the state after $j$ forward layers by $\rho_j=U_j\cdots U_1\rho U_1^\dagger\cdots U_j^\dagger$, likewise, define the back-propagated observable after $k$ layers as $O_k=U_{L-k+1}^\dagger\cdots U_L^\dagger O U_L\cdots U_{L-k+1}=\sum_i c_i^{(k)}P_i$, so that $O_0=O$ and $O_L=U^\dagger OU$, with $f = \Tr[\rho U^\dagger O U]$, then, if compression is applied at the layers of a subset $\RC\subseteq\{1,\dots,L\}$, the mean squared error after $L$ layers is 
    \begin{equation}
        \MSE[\widehat f]
        =\Ebb\left[
          \left(\sum_{k\in\RC}\sum_i
          \delta c_i^{(k)}\Tr[\rho_{L-k}P_i]\right)^2
          \right]\,,
    \end{equation}
    where $\delta c_j^{(k+1)}  = \widehat c_j^{(k+1)}-c_j^{(k+1)}$ is the error due to compression at step $k+1$.
\end{lemma}
\begin{proof}

Write the circuit as $U=U_L\cdots U_2U_1$ and define the state after $j$ forward
layers by
\begin{equation}
    \rho_j=U_j\cdots U_1\rho U_1^\dagger\cdots U_j^\dagger\,.
    \label{eq:app_forward_state}
\end{equation}

Similarly, we can define the backward-propagated observable after $k$ layers to be
\begin{equation}
    O_k=U_{L-k+1}^\dagger\cdots U_L^\dagger O
         U_L\cdots U_{L-k+1}
       =\sum_i c_i^{(k)}P_i\,,
    \label{eq:app_exact_observable_layer}
\end{equation}
where $c_i^{(k)}$ represents the coefficient of the $P_i$ Pauli, with $k$ representing the layer. Then the expectation value of $O_k$ and a general $\rho$ can be written as
\begin{equation}
    \Tr[\rho O_k]
    =\sum_i c_i^{(k)}\Tr[\rho P_i].
    \label{eq:app_pauli_trace_expansion}
\end{equation}

Our randomized algorithm stores estimates of $c_i^{(k)}$ instead. Indeed, we have $\widehat O_k=\sum_i\widehat c_i^{(k)}P_i$. To be able to find a general expression for the error, we need to find a recursion step that allows us to go from layer $k$ to layer $k+1$, and we can do so as follows
\begin{equation}
    U_{L-k}^\dagger P_iU_{L-k}
    =\sum_j a_{ij}^{(k+1)}P_j\,.
    \label{eq:app_layer_transfer_coefficients}
\end{equation}

Then, if we propagate the current random coefficients, obtained by using the randomized compression algorithm, exactly through layer $k+1$, the resulting coefficients would be

\begin{equation}
    c_j^{(k+1)}
    =\sum_i\widehat c_i^{(k)}a_{ij}^{(k+1)}\,,
    \label{eq:app_precompression_coefficient}
\end{equation}
where we note that we do not use $\hat c_j^{(k+1)}$ as this is an exact propagation from the previous $\hat c_i^{(k)}$. It is the same notation as the true expansion defined in Eq.~\eqref{eq:app_pauli_trace_expansion} that we will keep for simplicity for the rest of the paper unless explicitly stated otherwise.

After compression, we can write the stored coefficient as the exact propagated coefficient $c_j^{(k+1)}$ and a perturbation due to the randomness or truncation $\delta c_j^{(k+1)}$
\begin{equation}
    \widehat c_j^{(k+1)}=c_j^{(k+1)}+\delta c_j^{(k+1)}\,.
    \label{eq:app_local_delta_definition}
\end{equation}

Therefore, when propagating from step $k$ to step $k+1$, the running estimate $\Tr[\rho_{L-k}\widehat O_k]$ of $f = \Tr[\rho U^\dagger O U]$ transforms as
\begin{align}
    \sum_i\widehat c_i^{(k)}
      \Tr[\rho_{L-k}P_i]
    &=\sum_jc_j^{(k+1)}
      \Tr[\rho_{L-k-1}P_j]\\
    &=\sum_j\left(
      \widehat c_j^{(k+1)}-\delta c_j^{(k+1)}
      \right)
      \Tr[\rho_{L-k-1}P_j]
\end{align}
where in the first identity we used the map in Eq.~\eqref{eq:app_layer_transfer_coefficients}, and in the second we used Eq.~\eqref{eq:app_local_delta_definition} to write the propagated coefficient as a function of the estimated coefficient and the error. By rearranging this sum, we can find 
\begin{align}
    \sum_j
      \widehat c_j^{(k+1)}
      \Tr[\rho_{L-k-1}P_j] &= \sum_i\widehat c_i^{(k)}
      \Tr[\rho_{L-k}P_i]
      +\sum_j\delta c_j^{(k+1)}
      \Tr[\rho_{L-k-1}P_j]\label{eq:app_layer_error_recursion}
\end{align}
that gives us a description of the compressed function at step $k+1$ as a function of the coefficients on the previous step, and the error. This is the crucial expression that will allow us to write the error as a recursive formula. Indeed, we can apply Eq.~\eqref{eq:app_layer_error_recursion} recursively and write
\begin{equation}
    \widehat f = \displaystyle\sum_j \widehat c_j^{(L)} \Tr[\rho P_j] = \sum_{k=1}^L \sum_j \delta c_j^{(k)}\Tr[\rho_{L-k}P_j] + \sum_j \widehat c_j^{(0)}\Tr[\rho_{L}P_j]
\end{equation}
to obtain a formula of our estimator $\widehat f$ as a function of the errors and estimators at each step. The only thing that is left now is to compute the error. Hence, with the convention $\widehat c_j^{(0)} = c_j^{(0)}$, we have that the original $f$ can be written as $ f = \displaystyle\sum_j \widehat c_j^{(0)}\Tr[\rho_{L}P_j]$. We can readily compute the error of our estimator, that reads
\begin{equation}
    \widehat f-f
    =\sum_{k=1}^{L}\sum_i
      \delta c_i^{(k)}\Tr[\rho_{L-k}P_i]\,.
    \label{eq:app_full_error_recursion}
\end{equation}
If the compression occurs only at the layers of a subset $\RC\subseteq\{1,\dots,L\}$, the outer sum is restricted to $k\in\RC$, because $\delta c_i^{(k)}=0$ at every layer where no compression takes place. Since $P_i$ is a Pauli string and $\rho_{L-k}$ a density matrix, we also have $|\Tr[\rho_{L-k}P_i]|\leq1$ at every layer, a bound used repeatedly below.

Then, we can write the mean squared error by plugging in Eq.~\eqref{eq:app_full_error_recursion}
\begin{equation}
    \MSE[\widehat f]
    =\Ebb\left[
      \left(\sum_{k\in\RC}\sum_i
      \delta c_i^{(k)}\Tr[\rho_{L-k}P_i]\right)^2
      \right]\,,
    \label{eq:app_mse_expanded}
\end{equation}
which concludes the proof.

\end{proof}

Equipped with this lemma, we can move forward to bounding the MSE of the unbiased one-norm, as well as the comb sampling strategy.

\subsection{One-norm sampling}
\label{app:one_norm_sampling_error}

\subsubsection{Bound on the variance of the one-norm sampling strategy}
\label{app:one_norm_sampling_error_unbiased}
Before turning to the mean squared error of the one-norm sampling with truncation, let us get a feeling for the computation by analysing the variance of the estimator, assuming no truncation is done. The variance tends to zero when the number of samples goes to infinity, and the bound grows quadratically with the one norm of the Pauli sum at each layer.

\begin{theorem}[Unbiased one-norm sampling]\label{thm:one_norm_sampling_error_unbiased}
Consider $f = \Tr[\rho U^\dagger O U]$, such that $U = U_L \dots U_2 U_1$. Suppose that we propagate the observable forward using the one-norm sampling strategy of Appendix~\ref{app:one_norm_resampling} to reduce the memory overhead at every layer $k$ of a subset $\RC\subseteq\{1,\dots,L\}$, using $m_s$ independent one-norm draws at each of them. Then the variance is bounded as
\begin{equation}
    \Var[\widehat f]
    \leq\frac{1}{m_s}
    \sum_{k\in \RC}\Ebb\left[W_k^2\right]\,,
    \label{eq:app_unbiased_general_bound}
\end{equation}
where $W_k=\sum_i\abs{c_i^{(k)}}$ is the one norm of the Pauli sum arriving at layer $k$. For $k\geq1$ this is itself a random variable, because the coefficients arriving at layer $k$ depend on the outcome of the resampling performed at the earlier layers, which is why the bound carries an expectation value. If $W_k\leq\overline W_k$ almost surely for some deterministic $\overline W_k$, then $\Ebb[W_k^2]$ may be replaced by $\overline W_k^2$. In particular, if resampling is applied at every layer, $\RC=\{1,\dots,L\}$, and the one norm grows at most exponentially, $W_k \leq W_0 \gamma^k$ with $\gamma>1$, then we can bound the variance as
\begin{equation}
    \Var[\widehat f]
    \leq
    \frac{W_0^2\gamma^2}{m_s}
    \frac{\gamma^{2L}-1}{\gamma^2-1}\,,
    \label{eq:app_unbiased_exponential_bound}
\end{equation}
We stress that the exponent here is the circuit depth $L$ and not the number $\abs{\RC}$ of resampling events. For a sparse $\RC$ the relevant quantity is $\sum_{k\in\RC}\gamma^{2k}$, which is controlled by the largest index in $\RC$ rather than by how many indices it contains.
\end{theorem}

\begin{proof}
Every one-norm resampling step is unbiased by Eq.~\eqref{eq:app_one_norm_unbiased}, so $\Ebb[\widehat f]=f$ and the variance coincides with the mean squared error, $\Var[\widehat f]=\MSE[\widehat f]$; we bound the latter. Throughout, $\FC_k$ denotes the $\sigma$-algebra generated by every resampling outcome at the layers preceding $k$, so that conditioning on $\FC_k$ fixes the coefficients $c_i^{(k)}$ arriving at layer $k$, and with them $W_k$ and $p_i^{(k)}$, leaving only the layer-$k$ draws random. Every layer-$k$ moment below is conditional in this sense.
Let us start by defining, at a layer $k$, the total one norm $W_k$ and the probability of the $i$-th outcome $p_i^{(k)}$
\begin{equation}
    W_k=\sum_i\abs{c_i^{(k)}}\,,
    \qquad
    p_i^{(k)}=\frac{\abs{c_i^{(k)}}}{W_k}\,.
    \label{eq:app_layer_probabilities}
\end{equation}

Let $m_s$ be the number of draws of the multinomial distribution created by sampling from the probabilities $p_i^{(k)}$. Recall that the estimator for this sampling strategy was defined in Eq.~\eqref{eq:app_one_norm_estimator} as
\begin{equation}
    \widehat c_i^{(k)}
    =\sign(c_i^{(k)})
      \frac{W_k}{m_s}N_i^{(k)}\,,
    \label{eq:app_layer_multinomial_estimator}
\end{equation}
where $N_i^{(k)}$ is the number of draws at layer $k$ that return the Pauli string $P_i$.

We have already shown in Eq.~\eqref{eq:app_one_norm_unbiased} that $\widehat c_i^{(k)}$ is unbiased, and that 
\begin{align}
    \Ebb[N_i^{(k)}]
    &=m_sp_i^{(k)}\,.
\end{align}

By using standard properties of the multinomial distribution~\cite{tsun2020multinomial}, we can also see that the variance of $N_i^{(k)}$
\begin{align}
   \Var[N_i^{(k)}]
    &=m_sp_i^{(k)}(1-p_i^{(k)}) \,,
    \label{eq:app_multinomial_var}
\end{align}
and similarly, the covariance between two different counts $N_i^{(k)}$ and $N_j^{(k)}$ with $i\neq j$ in the same layer is
\begin{align}
    \Covar[N_i^{(k)},N_j^{(k)}] =-m_sp_i^{(k)}p_j^{(k)}\,.
    \label{eq:app_multinomial_covar}
\end{align}

Recall that the goal is to compute the variance of the estimator. To that end, we are going to use Lemma~\ref{lemma:general_mse}. Before we do that, we can use the results we just derived to individually compute the quantities required. We start by recalling the definition of $\delta c_j^{(k)}  = \widehat c_j^{(k)}-c_j^{(k)}$, and using it to compute $\Ebb[\delta c_i^{(k)}]$
\begin{align}
    \Ebb[\delta c_i^{(k)}]&=0\,,
    \label{eq:app_zero_local_bias}
\end{align}
where we just need to use that $\Ebb[\widehat c_j^{(k)}] = c_j^{(k)}$ as already shown in Eq.~\eqref{eq:app_one_norm_unbiased}. Other quantities that we need to compute to use Lemma~\ref{lemma:general_mse} are $\Ebb[(\delta c_i^{(k)})^2], \Ebb[\delta c_i^{(k)}\delta c_j^{(k)}]$. Let us start with the second moment of $\delta c_i^{(k)}$
\begin{align}
    \Ebb[(\delta c_i^{(k)})^2]
    =& \Var[\widehat c_i^{(k)}]\\
    =& \frac{W_k^2}{m_s^2}\:\Var[N_i^{(k)}]\\
    =&\frac{W_k}{m_s}\abs{c_i^{(k)}}
      \left(1-\frac{\abs{c_i^{(k)}}}{W_k}\right)\,,
    \label{eq:app_local_coeff_variance}
\end{align}
in the first equality we just used that $\Ebb[(\delta c_i^{(k)})^2] = \Ebb[(\widehat c_i^{(k)} - \Ebb[\widehat c_i^{(k)}])^2]$, the definition of the variance. In the second equality we used the form of $\widehat c_i^{(k)}$ in Eq.~\eqref{eq:app_layer_multinomial_estimator}, and in the last equality we used Eq.~\eqref{eq:app_multinomial_var} along with $p_i^{(k)}$ in Eq.~\eqref{eq:app_layer_probabilities}. Finally, we can move to $\Ebb[\delta c_i^{(k)}\delta c_j^{(k)}]$
\begin{align}
    \Ebb[\delta c_i^{(k)}\delta c_j^{(k)}] =& \Covar[\widehat c_i^{(k)},\widehat c_j^{(k)}]\\
    =& \frac{W_k^2}{m_s^2}\sign(c_i^{(k)})\sign(c_j^{(k)})\Covar[N_i^{(k)},N_j^{(k)}]\\
    = & -\frac{c_i^{(k)}c_j^{(k)}}{m_s}\,,\qquad i\neq j\,,
    \label{eq:app_local_coeff_covariance}
\end{align}
where in the first equality we used the definition of covariance $\Covar[x,y] = \Ebb[(x-\Ebb[x])(y-\Ebb[y])]$. Similarly to the previous equation, in the second equality we use the definition of $\widehat c_i^{(k)}$ in Eq.~\eqref{eq:app_layer_multinomial_estimator}. In the final equality we used Eq.~\eqref{eq:app_multinomial_covar} along with $p_i^{(k)}$ in Eq.~\eqref{eq:app_layer_probabilities}.

Before applying the lemma we note that the sampling performed at a layer $k$ is independent of the history of the sampling before $k$. Concretely, for $k > \ell$,
\begin{equation}
    \Ebb[\delta c_i^{(k)}\:\delta c_j^{(\ell)}] = \Ebb\left[\Ebb\left[\delta c_i^{(k)}\delta c_j^{(\ell)}\,\middle|\,\FC_k\right]\right]  = \Ebb\left[\delta c_j^{(\ell)}\,\Ebb\left[\delta c_i^{(k)}\,\middle|\,\FC_k\right]\right] = 0
    \label{eq:onenorm_conditioning}
\end{equation}
because $\Ebb[\delta c_i^{(k)}\mid\FC_k] = 0$ by Eq.~\eqref{eq:app_zero_local_bias}, and because $\delta c_j^{(\ell)}$ is $\FC_k$-measurable for every $\ell<k$, hence constant under the inner expectation value.

With this, we are ready to invoke Lemma~\ref{lemma:general_mse} and massage the equation using the results we obtained above
\begin{align}\label{eq:mse_nobias_expansion}
    \MSE[\widehat f]
    =&\Ebb\left[
      \left(\sum_{k\in\RC}\sum_i
      \delta c_i^{(k)}\Tr[\rho_{L-k}P_i]\right)^2
      \right]\\
       = &  \sum_{k\in\RC} \Ebb\left[
      \left(\sum_i
      \delta c_i^{(k)}\Tr[\rho_{L-k}P_i]\right)^2
      \right] + \sum_{k\neq \ell\in\RC}\Ebb\left[
      \left(\sum_i
      \delta c_i^{(k)}\Tr[\rho_{L-k}P_i]\right)\left(\sum_i
      \delta c_i^{(\ell)}\Tr[\rho_{L-\ell}P_i]\right)
      \right]\\
      =& \sum_{k\in\RC} \Ebb\left[
      \left(\sum_i
      \delta c_i^{(k)}\Tr[\rho_{L-k}P_i]\right)^2
      \right] + \sum_{k\neq \ell\in\RC}\sum_{i, j}\Ebb\left[
      \delta c_i^{(k)}
      \delta c_j^{(\ell)}\right] \Tr[\rho_{L-k}P_i] \Tr[\rho_{L-\ell}P_j]\\
      =& \sum_{k\in\RC} \Ebb\left[
      \left(\sum_i
      \delta c_i^{(k)}\Tr[\rho_{L-k}P_i]\right)^2
      \right]\,,
\end{align}
where in the second equality we expand the sum over $k$, in the third equality we used that $\Tr[\rho_{L-k}P_i]$ and $\Tr[\rho_{L-\ell}P_j]$ are constants with respect to $\Ebb[\cdot]$. In the final equality we used the fact that for $k \neq \ell$ and all $i, j$ we have $\Ebb[\delta c_i^{(k)}\:\delta c_j^{(\ell)}] = 0$ as derived in Eq.~\eqref{eq:onenorm_conditioning}, to remove the crossed terms.

Now we can use Eqs.~\eqref{eq:app_local_coeff_variance} and
\eqref{eq:app_local_coeff_covariance}, to analyze the conditional variance at a single
event $k$, i.e.
\begin{align}\label{eq:app_exact_event_variance_init}
   \Ebb\left[\left(\sum_i\delta c_i^{(k)}\Tr[\rho_{L-k}P_i]\right)^2\right]
    =& \sum_i \Ebb[(\delta c_i^{(k)})^2] \Tr[\rho_{L-k}P_i]^2 + \sum_{i\neq j}\Ebb[\delta c_i^{(k)} \delta c_j^{(k)}]\Tr[\rho_{L-k}P_i]\Tr[\rho_{L-k}P_j]\\
    =&\frac{W_k}{m_s}
      \sum_i\abs{c_i^{(k)}}\Tr[\rho_{L-k}P_i]^2 - \frac{1}{m_s}\sum_i (c_i^{(k)}\Tr[\rho_{L-k}P_i])^{2} \\
      &- \frac{1}{m_s} \sum_{i \neq j}c_i^{(k)}c_j^{(k)}\Tr[\rho_{L-k}P_i]\Tr[\rho_{L-k}P_j]\nonumber\\
    =&\frac{W_k}{m_s}
      \sum_i\abs{c_i^{(k)}}\Tr[\rho_{L-k}P_i]^2-\frac{1}{m_s}
      \left(\sum_i c_i^{(k)}\Tr[\rho_{L-k}P_i]\right)^2\,,
    \label{eq:app_exact_event_variance}
\end{align}
where in the first equality we just expanded the sum in $i$. In the second equality we used Eq.~\eqref{eq:app_local_coeff_variance} and Eq.~\eqref{eq:app_local_coeff_covariance} as mentioned above and in the last equality we wrote both negative sums into one. 

The second term in Eq.~\eqref{eq:app_exact_event_variance} is non-positive, thus we can use this and the fact that $|\Tr[\rho_{L-k}P_i]|\leq 1$ to upper bound  
\begin{align}
    \Ebb\left[\left(\sum_i\delta c_i^{(k)}\Tr[\rho_{L-k}P_i]\right)^2\right] \leq \frac{W_k}{m_s}
      \sum_i\abs{c_i^{(k)}}\Tr[\rho_{L-k}P_i]^2 \leq \frac{W_k^2}{m_s}\,.
\end{align}
This bound holds conditionally on $\FC_k$. Taking the outer expectation value and summing over the resampling events gives
\begin{align}
    \MSE[\widehat f]\leq \frac{1}{m_s}\sum_{k\in\RC}\Ebb\left[W_k^2\right]\,,
\end{align}
which is Eq.~\eqref{eq:app_unbiased_general_bound}. If resampling occurs at every layer and we further assume that $W_k\leq \gamma^k W_0$, then
\begin{align}
    \MSE[\widehat f]\leq  \frac{W_0^2}{m_s}\sum_{k=1}^{L} \gamma^{2k} = \frac{W_0^2 \gamma^2}{m_s} \frac{\gamma^{2L}-1}{\gamma^2-1}\,,
\end{align}
by the standard geometric sum $\sum_{k=1}^K x^k = x\frac{x^K-1}{x-1}$, which concludes the proof. This last step requires $\RC=\{1,\dots,L\}$: for a sparse $\RC$ the sum $\sum_{k\in\RC}\gamma^{2k}$ is not equal to $\sum_{k=1}^{\abs{\RC}}\gamma^{2k}$, and is in general much larger.

\end{proof}

\subsubsection{Deterministic truncation followed by resampling}
\label{app:one_norm_sampling_error_bias}

We will now show how the upper bound of Theorem \ref{thm:one_norm_sampling_error_unbiased} is modified when one introduces deterministic truncation before sampling. The one-norm resampling with bias can be described as one norm resampling with modified coefficients. Therefore, the MSE can be upper bounded by the same bound as in the case without truncation, but with the modified one norm, to which is added a bias term depending on the truncation. Because of the bias term, the error does not tend to zero when the number of samples goes to infinity.

\begin{theorem}[MSE bound for the one-norm sampling with truncation]\label{thm:one_norm_sampling_error_bias}
    Consider $f = \Tr[\rho U^\dagger O U]$, such that $U = U_L \dots U_2 U_1$. Suppose that we propagate the observable forward using the one-norm sampling strategy preceded by truncation as in Appendix~\ref{app:one_norm_resampling_bias}. At every layer $k$ of a subset $\RC\subseteq\{1,\dots,L\}$ we use a bias threshold $b_k$
    and we perform $m_s$ independent one-norm draws. Then the mean squared error is upper bounded as
\begin{align}
    \MSE[\widehat f]
    \leq
      \Ebb\left[\left(\sum_{k\in\RC}\left(W_k-\widetilde W_k\right)\right)^{2}\right]
    +\frac{2}{\sqrt{m_s}}\sum_{k > \ell \in\RC}
      \sqrt{\Var[\beta_k]\,\Ebb\left[\widetilde W_\ell^2\right]}
    +\frac{1}{m_s}
      \sum_{k\in\RC}\Ebb\left[\widetilde W_k^2\right]\,,
    \label{eq:app_truncated_general_bound}
\end{align}
where $W_k = \displaystyle \sum_i |c_i^{(k)}|$ is the one norm without truncation at layer $k$, $\widetilde W_k = \displaystyle \sum_{i, |c_i^{(k)}| > b_k} |c_i^{(k)}|$ is the one norm of the coefficients above the threshold $b_k$, and $\beta_k$ is the truncation error committed at layer $k$, $\beta_k=-\sum_{i,\,\abs{c_i^{(k)}}\leq b_k}
      c_i^{(k)}\Tr[\rho_{L-k}P_i]$. As in Theorem~\ref{thm:one_norm_sampling_error_unbiased}, all three quantities are random for $k\geq1$, which is why the bound carries expectation values; if $W_k-\widetilde W_k$ and $\widetilde W_k$ admit deterministic almost-sure upper bounds, those may be substituted and the expectation values dropped, which is how the model of Appendix~\ref{app:one_norm_sampling_mse_model} uses this result.
\end{theorem}

\begin{proof}

The one-norm resampling with bias can be described as one norm resampling with modified coefficients. The modified coefficients, modified one-norm, and modified probabilities at layer $k$ with threshold $b_k$ can be respectively expressed as
\begin{align}
    \widetilde c_i^{(k)}
    &=c_i^{(k)}
      \mathbf 1\!\left(\abs{c_i^{(k)}}>b_k\right),
    \label{eq:app_truncated_coefficient_layer}\\
    \widetilde W_k
    &=\sum_i\abs{c_i^{(k)}}
      \mathbf 1\!\left(\abs{c_i^{(k)}}>b_k\right),
    \label{eq:app_truncated_norm_layer}\\
    \widetilde p_i^{(k)}
    &=\frac{\abs{c_i^{(k)}}}{\widetilde W_k}
      \mathbf 1\!\left(\abs{c_i^{(k)}}>b_k\right)\,.
    \label{eq:app_truncated_probability_layer}
\end{align}
This directly follows from Appendix~\ref{app:one_norm_resampling_bias}. The estimator of the coefficient $i$ at layer $k$ is equal to
\begin{align}
    \widehat c_i^{(k)}
    =\sign(c_i^{(k)})
      \frac{\widetilde W_k}{m_s}N_i^{(k)}\,,
    \label{eq:app_truncated_estimator}
\end{align}
where the $N_i^{(k)}$ follow a multinomial distribution of $m_s$ draws with probabilities $\widetilde p_i^{(k)}$.

For a retained index, $\abs{c_i^{(k)}}>b_k$, the sampling step is exactly the unbiased one-norm resampling of Appendix~\ref{app:one_norm_resampling} applied to the modified Pauli sum $\sum_i\widetilde c_i^{(k)}P_i$. The computations of Appendix~\ref{app:one_norm_sampling_error_unbiased} therefore hold with $W_k$ replaced by $\widetilde W_k$, so that Eqs.~(\ref{eq:app_zero_local_bias},\ref{eq:app_local_coeff_variance},\ref{eq:app_local_coeff_covariance}) become
\begin{align}
    \Ebb\left[\delta c_i^{(k)}\right] &= 0\,,
    \label{eq:app_truncated_retained_mean}\\
    \Ebb\left[(\delta c_i^{(k)})^2\right] &= \frac{\widetilde W_k}{m_s}\abs{c_i^{(k)}}
      \left(1-\frac{\abs{c_i^{(k)}}}{\widetilde W_k}\right)\,,
    \label{eq:app_truncated_retained_variance}\\
    \Ebb\left[\delta c_i^{(k)}\delta c_j^{(k)}\right] &= -\frac{c_i^{(k)}c_j^{(k)}}{m_s}\,,
    \qquad i\neq j\,,
    \label{eq:app_truncated_retained_crossed}
\end{align}
where $\delta c_i^{(k)}=\widehat c_i^{(k)}-c_i^{(k)}$ as in Lemma~\ref{lemma:general_mse}, and both $i$ and $j$ are retained indices. For a discarded index, $\abs{c_i^{(k)}}\leq b_k$, Eq.~\eqref{eq:app_truncated_probability_layer} gives $\widetilde p_i^{(k)}=0$, so that index can never be drawn. Equation~\eqref{eq:app_truncated_estimator} then yields $\widehat c_i^{(k)}=0$, and therefore $\delta c_i^{(k)} = -c_i^{(k)}$ for every discarded index $i$.

The two groups of indices therefore contribute in two very different ways to the error that layer $k$ hands to Lemma~\ref{lemma:general_mse}. Splitting that error accordingly we find
\begin{align}
    \sum_i\delta c_i^{(k)}\Tr[\rho_{L-k}P_i] = S_k+\beta_k\,,
    \label{eq:app_truncated_layer_split}
\end{align}
where we introduced the sampling noise of the retained coefficients
\begin{align}
    S_k=\sum_{i,\,\abs{c_i^{(k)}}>b_k}
      \delta c_i^{(k)}\Tr[\rho_{L-k}P_i]\,,
    \label{eq:app_truncated_layer_noise_def}
\end{align}
and the truncation error committed at layer $k$
\begin{align}
    \beta_k=-\sum_{i,\,\abs{c_i^{(k)}}\leq b_k}
      c_i^{(k)}\Tr[\rho_{L-k}P_i]\,.
    \label{eq:app_truncated_layer_bias_def}
\end{align}

Similarly to the calculations in Eqs.~(\ref{eq:app_truncated_retained_mean}-\ref{eq:app_truncated_retained_crossed}), we can compute the first two conditional moments of $S_k$. Conditioning on $\FC_k$ fixes the coefficients arriving at layer $k$, and therefore also $\widetilde W_k$ and the set of retained indices, so that
\begin{align}
    \Ebb\left[S_k\right]
    =\sum_{i,\,\abs{c_i^{(k)}}>b_k}
      \Ebb\left[\delta c_i^{(k)}\right]\Tr[\rho_{L-k}P_i]
    =0\,,
    \label{eq:app_truncated_noise_centred}
\end{align}
and, by the tower property, $\Ebb[S_k]=0$ as well. Similarly, the conditional second moment is
\begin{align}
    \Ebb\left[S_k^2\right]
    =&\:\frac{\widetilde W_k}{m_s}
      \sum_{i,\,\abs{c_i^{(k)}}>b_k}\abs{c_i^{(k)}}
      \Tr[\rho_{L-k}P_i]^2 -\frac{1}{m_s}\left(
      \sum_{i,\,\abs{c_i^{(k)}}>b_k}c_i^{(k)}
      \Tr[\rho_{L-k}P_i]\right)^2
    \;\leq\;\frac{\widetilde W_k^2}{m_s}\,,
    \label{eq:app_truncated_noise_bound}
\end{align}
where the derivation is the exact same as Eqs.~(\ref{eq:app_exact_event_variance_init}-\ref{eq:app_exact_event_variance}).

Equipped with this, we are ready to invoke Lemma~\ref{lemma:general_mse} and expand the square over the resampling steps, as we did in Eq.~\eqref{eq:mse_nobias_expansion}, allows us to find the following expression
\begin{align}
    \MSE[\widehat f]
    =&\:\Ebb\left[\left(\sum_{k\in\RC}\sum_i
      \delta c_i^{(k)}\Tr[\rho_{L-k}P_i]\right)^2\right]\\
    =&\:\Ebb\left[\left(\sum_{k\in\RC}
      \left(S_k+\beta_k\right)\right)^2\right]\\
    =&\:\sum_{k\in\RC}\Ebb\left[\left(S_k+\beta_k\right)^2\right]
    +2\sum_{k>\ell\in\RC}\Ebb\left[
      \left(S_k+\beta_k\right)\left(S_\ell+\beta_\ell\right)\right]\,,
    \label{eq:app_truncated_mse_expansion}
\end{align}
where in the second equality we used Eq.~\eqref{eq:app_truncated_layer_split} and in the third one $(\sum_kx_k)^2=\sum_kx_k^2+2\sum_{k>\ell}x_kx_\ell$.

Now we will analyze the different terms. For the diagonal ones we condition on $\FC_k$, which fixes $\beta_k$ and leaves $S_k$ random, 
\begin{align}
    \Ebb\left[\left(S_k+\beta_k\right)^2\right]
    =&\:\Ebb\left[S_k^2\right]
    +2\Ebb\left[\beta_k\,\Ebb\left[S_k\,\middle|\,\FC_k\right]\right]
    +\Ebb\left[\beta_k^2\right]\\
    \leq&\:\Ebb\left[\beta_k^2\right]+\frac{\Ebb\left[\widetilde W_k^2\right]}{m_s}\,,
    \label{eq:app_truncated_diagonal_bound}
\end{align}
where in the first equality we used the tower property $\Ebb[\cdot] = \Ebb[\Ebb[\cdot\mid\FC_k]]$, and that $\Ebb[\beta_k|\FC_k] = \beta_k$. In the inequality we used that the cross term vanishes by Eq.~\eqref{eq:app_truncated_noise_centred} together with the bound of Eq.~\eqref{eq:app_truncated_noise_bound}. The diagonal terms are therefore exactly the sum of a squared truncation error and a sampling noise, with no interference between the two within a layer.

For the off-diagonal terms with $k>\ell$ we condition on $\FC_k$ as well. The layer-$\ell$ error is then fixed, because the resampling at layer $\ell$ has already taken place, so it can be taken out of the inner expectation value as we did in Eq.~\eqref{eq:onenorm_conditioning}, and Eq.~\eqref{eq:app_truncated_noise_centred} leaves
\begin{align}
    \Ebb\left[\left(S_k+\beta_k\right)\left(S_\ell+\beta_\ell\right)\right]
    =&\:\Ebb\left[\Ebb\left[\left(S_k+\beta_k\right)
      \left(S_\ell+\beta_\ell\right)\,\middle|\,\FC_k\right]\right] \\
    =&\:\Ebb\left[\left(S_\ell+\beta_\ell\right)
      \left(\Ebb\left[S_k\,\middle|\,\FC_k\right]+\beta_k\right)\right] \\
    =&\:\Ebb\left[\beta_k\left(S_\ell+\beta_\ell\right)\right] = \Ebb[\beta_k S_\ell] + \Ebb[\beta_k\beta_\ell]\,,
    \label{eq:app_truncated_cross_layer}
\end{align}

where in the first equality we used again the tower property, in the second one that $S_\ell$, $\beta_\ell$ and $\beta_k$ are all $\FC_k$-measurable for $\ell<k$, and in the third one Eq.~\eqref{eq:app_truncated_noise_centred}. This is the only step where the truncated calculation departs from the unbiased one. 
The truncation errors of different layers are correlated with one another, as $\beta_k$ is built from the coefficients from an earlier $k$ with the earlier sampling noise, and this is what makes the truncation error accumulate coherently along the circuit.

To bound Eq.~\eqref{eq:app_truncated_cross_layer} we treat the two terms separately. Since $\Ebb[S_\ell]=0$ by Eq.~\eqref{eq:app_truncated_noise_centred}, the first term in Eq.~\eqref{eq:app_truncated_cross_layer} can be expressed as a covariance, and bounded by Cauchy--Schwarz inequality
\begin{align}
    \left|\Ebb\left[\beta_kS_\ell\right]\right|
    =\left|\Covar\left[\beta_k,S_\ell\right]+\Ebb\left[\beta_k\right]\Ebb\left[S_\ell\right]\right|
    \leq\sqrt{\Var[\beta_k]\Var[S_\ell]}
    \leq\sqrt{\Var[\beta_k]}\;\frac{\sqrt{\Ebb\left[\widetilde W_\ell^2\right]}}{\sqrt{m_s}}\,,
    \label{eq:app_truncated_cross_layer_bound}
\end{align}
where in the last inequality we used $\Var[S_\ell]\leq \Ebb[S_\ell^2]=\Ebb\left[\Ebb\left[S_\ell^2\,\middle|\,\FC_\ell\right]\right]\leq\Ebb[\widetilde W_\ell^2]/m_s$, which follows from Eq.~\eqref{eq:app_truncated_noise_bound} and the tower property.

Collecting Eqs.~(\ref{eq:app_truncated_diagonal_bound},\ref{eq:app_truncated_cross_layer},\ref{eq:app_truncated_cross_layer_bound}) in the expansion of Eq.~\eqref{eq:app_truncated_mse_expansion},
\begin{align}
    \MSE[\widehat f]
    \leq&\:\sum_{k\in\RC}\Ebb\left[\beta_k^2\right]
    +2\sum_{k>\ell\in\RC}\Ebb\left[\beta_k\beta_\ell\right] \nonumber\\
    &+\frac{2}{\sqrt{m_s}}\sum_{k>\ell\in\RC}
      \sqrt{\Var[\beta_k]\,\Ebb\left[\widetilde W_\ell^2\right]}
    +\frac{1}{m_s}\sum_{k\in\RC}\Ebb\left[\widetilde W_k^2\right]\\
    =&\:\Ebb\left[\left(\sum_{k\in\RC}\beta_k\right)^2\right]
    +\frac{2}{\sqrt{m_s}}\sum_{k>\ell\in\RC}
      \sqrt{\Var[\beta_k]\,\Ebb\left[\widetilde W_\ell^2\right]}
    +\frac{1}{m_s}\sum_{k\in\RC}\Ebb\left[\widetilde W_k^2\right]\,,
    \label{eq:app_truncated_mse_collected}
\end{align}
where in the last equality we recognised the first two sums as the expansion of the square of the total truncation error. 

Finally, applying $|\Tr[\rho_{L-k}P_i]|\leq1$ to Eq.~\eqref{eq:app_truncated_layer_bias_def} bounds the truncation error of a single layer by the mass that the threshold discards,
\begin{align}
    \abs{\beta_k}
    \leq\sum_{i,\,\abs{c_i^{(k)}}\leq b_k}\abs{c_i^{(k)}}
    =W_k-\widetilde W_k\,,
    \label{eq:app_discarded_mass_bound}
\end{align}
so that the triangle inequality, followed by the monotonicity of $x\mapsto x^2$ on $[0,\infty)$, gives
\begin{align}
    \Ebb\left[\left(\sum_{k\in\RC}\beta_k\right)^2\right]
    \leq\Ebb\left[\left(\sum_{k\in\RC}\left(W_k-\widetilde W_k\right)\right)^2\right]\,,
    \label{eq:app_total_bias_bound}
\end{align}
which together with Eq.~\eqref{eq:app_truncated_mse_collected} is the claim of Eq.~\eqref{eq:app_truncated_general_bound} and concludes the proof.

We keep the variance of $\beta_k$ explicit in Eq.~\eqref{eq:app_truncated_general_bound} because it is the quantity that controls the coupling between layers, and it can always be bounded in the worst case using Eq.~\eqref{eq:app_discarded_mass_bound},
\begin{align}
    \Var[\beta_k]\leq\Ebb\left[\beta_k^2\right]
    \leq\Ebb\left[\left(W_k-\widetilde W_k\right)^2\right]\,,
    \label{eq:app_truncated_bias_variance_bound}
\end{align}
which turns the whole bound into an expression involving only the discarded masses and the retained one-norms.
\end{proof}

\subsubsection{Mean squared error model}
\label{app:one_norm_sampling_mse_model}

In order to better parse the consequences of Theorem~\ref{thm:one_norm_sampling_error_bias}, we propose a phenomenology model to explain the results observed in Fig.~\ref{fig:resampling_increase_bias}, by making assumptions on the behaviour of the one-norm and the discarded one-norm after truncation. We can retrieve the fact that truncating the smallest terms will add a floor to the error, but will decrease the asymptotically dominant term by a constant factor, compared to the unbiased case.  The approximate inequality $\lesssim$ captures that our bound holds approximately in the limit that $L \gg 1$.

\begin{result}[Model for the MSE of the one-norm sampling with truncation]
\label{result:app_truncated_mse_model}
    Assume that relative truncation is applied at each of the $L$ layers with $b_k=\epsilon_{\mathrm{rel}}\max_j\abs{c_j^{(k)}}$, that the fraction of the one-norm discarded at layer $k$ decays geometrically as $\varphi_k=\nu\epsilon_{\mathrm{rel}}\lambda^{k-1}$ with $0<\lambda<1$ while the one-norm grows by a factor $\gamma>1$ per layer, that $\mu\epsilon_{\mathrm{rel}}\ll1$ with $\mu=\frac{\nu}{1-\lambda}$, and that the truncation error of a layer varies little from run to run. Then, for $q=\gamma\lambda\neq1$,
\begin{align}
    \MSE[\widehat f]
    \lesssim
    \left(\nu\epsilon_{\mathrm{rel}}W_0\gamma\,
      \frac{q^{L}-1}{q-1}\right)^{2}
    +\frac{\left(1-\mu\epsilon_{\mathrm{rel}}\right)^2W_0^2\gamma^2}{\gamma^2-1}
      \frac{\gamma^{2L}}{m_s}\,,
    \qquad q=\gamma\lambda\,,
    \label{eq:app_truncated_mse_model_result}
\end{align}
where $\nu>0$ is a problem-dependent constant. For $q=1$ the geometric sum degenerates and $\left(q^{L}-1\right)/\left(q-1\right)$ must be replaced by $L$, which is the marginal case discussed below and the one quoted in the main text in Eq.~\eqref{eq:phenomenological_truncated_mse}.
\end{result}

We will consider a relative threshold,
\begin{align}
    b_k=\epsilon_{\mathrm{rel}}
    \max_j\abs{c_j^{(k)}}\,.
    \label{eq:app_relative_cutoff}
\end{align}
Both quantities that Theorem~\ref{thm:one_norm_sampling_error_bias} needs, the discarded mass $W_k-\widetilde W_k$ and the retained norm $\widetilde W_k$, follow from a single modelling assumption on the fraction of the one-norm that the threshold removes at each layer. Let
\begin{align}
    \varphi_k=\frac{W_k-\widetilde W_k}{W_k}
    \label{eq:app_discarded_fraction}
\end{align}
be that fraction, and assume that it decays geometrically along the circuit at a rate $\lambda$ 
\begin{align}
    \varphi_k=\nu\epsilon_{\mathrm{rel}}\lambda^{k-1}\,,
    \qquad 0<\lambda<1\,,
    \label{eq:app_discarded_fraction_decay}
\end{align}
while the one-norm grows on average by a factor $\gamma>1$ per layer, so that the norm arriving at layer $k$ is $W_k=\gamma\widetilde W_{k-1}$. The linear dependence on $\epsilon_{\mathrm{rel}}$ is what one expects in the first layers, where the coefficient distribution is still peaked and $\max_j\abs{c_j^{(k)}}$ sits well above the typical coefficient. The threshold then cuts into the bulk of the distribution, so that the mass below it grows linearly in $\epsilon_{\mathrm{rel}}$. The decay in $k$ reflects the flattening of the distribution as the circuit scrambles: once $\max_j\abs{c_j^{(k)}}$ approaches the typical coefficient, the threshold falls below it and much less is discarded.

Applying Eq.~\eqref{eq:app_discarded_fraction_decay} to the retained norm, we get the recursion
\begin{align}
    \widetilde W_k=\left(1-\varphi_k\right)W_k
    =\gamma\left(1-\varphi_k\right)\widetilde W_{k-1}\,,
    \label{eq:app_retained_norm_recursion}
\end{align}
and iterating from $\widetilde W_0=W_0$ gives
\begin{align}
    \widetilde W_k=W_0\gamma^k\prod_{j=1}^{k}\left(1-\varphi_j\right)\,.
    \label{eq:app_retained_norm_product}
\end{align}
Because Eq.~\eqref{eq:app_discarded_fraction_decay} is summable, the product converges and for $\nu\epsilon_{\mathrm{rel}}\ll1$
\begin{align}
    \prod_{j=1}^{\infty}\left(1-\varphi_j\right)
    =\exp\left(\sum_{j=1}^{\infty}\ln\left(1-\varphi_j\right)\right)
    \simeq\exp\left(-\sum_{j=1}^{\infty}\varphi_j\right)
    \simeq1-\mu\epsilon_{\mathrm{rel}}\,,
    \qquad \mu=\frac{\nu}{1-\lambda}\,,
    \label{eq:app_cepsilon}
\end{align}
where we expanded the logarithm and the exponential to first order and summed the geometric series. For sufficiently large $k$ the retained norm is therefore
\begin{align}
    \widetilde W_k\simeq\left(1-\mu\epsilon_{\mathrm{rel}}\right)W_0\gamma^k\,.
    \label{eq:app_retained_norm_model}
\end{align}

Applying Eq.~\eqref{eq:app_discarded_fraction_decay} to the discarded mass instead, and using $W_k=\gamma\widetilde W_{k-1}$ together with Eq.~\eqref{eq:app_retained_norm_product}, we find
\begin{align}
    W_k-\widetilde W_k=\varphi_kW_k
    =\nu\epsilon_{\mathrm{rel}}W_0\gamma\,q^{k-1}
      \prod_{j=1}^{k-1}\left(1-\varphi_j\right)
    \simeq\nu\epsilon_{\mathrm{rel}}W_0\gamma\,q^{k-1}\,,
    \qquad q=\gamma\lambda\,,
    \label{eq:app_discarded_mass_model}
\end{align}
where in the last step we used that the product is $1+\order{\epsilon_{\mathrm{rel}}}$ and that the prefactor $\varphi_k$ of Eq.~\eqref{eq:app_discarded_fraction_decay} is already linear in $\epsilon_{\mathrm{rel}}$. 

Finally, assume that the truncation error of a layer varies little from run to run,
\begin{align}
    \Var[\beta_k]\ll\left(W_k-\widetilde W_k\right)^2\,,
    \label{eq:app_phenom_bias_deterministic}
\end{align}
so that the cross term of Eq.~\eqref{eq:app_truncated_general_bound} can be dropped. What makes $\beta_k$ fluctuate strongly is a coefficient sitting so close to $b_k$ that the resampling noise of the previous layers decides whether it is kept or discarded. Away from that situation the set of surviving indices is a stable function of the incoming coefficients, so $\beta_k$ varies only through the smooth run-to-run drift of those coefficients rather than through discrete changes of the discard set. It is in this sense, and not as a deterministic number, that $\beta_k$ is treated as fixed here.

Inserting Eqs.~(\ref{eq:app_retained_norm_model},\ref{eq:app_discarded_mass_model},\ref{eq:app_phenom_bias_deterministic}) into Eq.~\eqref{eq:app_truncated_general_bound} with $\RC=\{1,\ldots,L\}$ we find
\begin{align}
    \MSE[\widehat f]
    \lesssim&\:\left(\nu\epsilon_{\mathrm{rel}}W_0\gamma
      \sum_{k=1}^{L}q^{k-1}\right)^{2}
    +\frac{\left(1-\mu\epsilon_{\mathrm{rel}}\right)^2W_0^2}{m_s}
      \sum_{k=1}^{L}\gamma^{2k}\\
    =&\:\left(\nu\epsilon_{\mathrm{rel}}W_0\gamma\,
      \frac{q^{L}-1}{q-1}\right)^{2}
    +\frac{\left(1-\mu\epsilon_{\mathrm{rel}}\right)^2W_0^2\gamma^2}{m_s\left(\gamma^2-1\right)}
      \left(\gamma^{2L}-1\right)\,,
    \label{eq:app_phenomenological_mse}
\end{align}
which for large $L$ is Eq.~\eqref{eq:app_truncated_mse_model_result}. In the first line the cross term was dropped by Eq.~\eqref{eq:app_phenom_bias_deterministic}, and in the second line we summed the two geometric series.

The truncation term of Eq.~\eqref{eq:app_phenomenological_mse} behaves in three different ways depending on $q$. For $q<1$ the discarded fraction shrinks faster than the one-norm grows, the series converges and the truncation error saturates at a depth-independent floor
\begin{align}
    \left(\sum_{k=1}^{\infty}\left(W_k-\widetilde W_k\right)\right)^2
    =\left(\frac{\nu\epsilon_{\mathrm{rel}}W_0\gamma}{1-q}\right)^{2}\,,
    \label{eq:app_bias_plateau}
\end{align}
so the bias does not accumulate with depth at all. For $q=1$ the sum is exactly $L$ and we recover the linear accumulation $\left(\nu\epsilon_{\mathrm{rel}}W_0\gamma L\right)^2$ quoted in the main text in Eq.~\eqref{eq:phenomenological_truncated_mse}. For $q>1$ the truncation error grows exponentially as $\epsilon_{\mathrm{rel}}^2q^{2L}$, and since it carries no factor of $1/m_s$ it cannot be reduced by sampling more.

\subsection{Comb sampling}
\label{app:comb_sampling_error}
Throughout this subsection, $m_s$ denotes the number of comb teeth at each resampling event, taken to be the same at every event. This holds exactly for uncalibrated comb sampling, where $m_s=m$, whereas calibration often yields $m_s>m$ depending on the specific coefficient structure.

\subsubsection{Bound on the variance of the comb sampling strategy}
\label{app:comb_sampling_error_unbiased}
As we did for the one-norm sampling, let us start by analysing the variance of the estimator with the comb sampling strategy assuming that no truncation is done. We call this a \textit{result} rather than theorem because we introduce some technical assumptions that we cannot guarantee will universally hold, and we denote it by the approximate $\lesssim$ inequality. We retrieve the observed fact of Fig.~\ref{fig:memory_utilization} that when the coefficient distribution is peaked (small $L$), the error made with comb sampling has a better scaling with the number of samples than the one norm resampling (inverse quadratic vs. inverse linear), and when the coefficient distribution flattens (when $L$ increases) the two methods have similar errors.

\begin{result}[Unbiased comb sampling]\label{result:comb_sampling_error_unbiased}
Consider $f = \Tr[\rho U^\dagger O U]$, such that $U = U_L \dots U_2 U_1$. Suppose that we propagate the observable forward using the comb sampling strategy of Appendix~\ref{app:comb_sampling} to reduce the memory overhead at every layer $k$ of a subset $\RC\subseteq\{1,\dots,L\}$, with comb spacing $\Delta_k=W_k/m_s$. Assume that the observable is generic with respect to the comb, meaning that the overlaps $\Tr[\rho_{L-k}P_i]$ carry no information about the order in which the coefficients are laid out along the cumulative weights, so that the error probes the typical size of the correlations between comb counts rather than their worst combination. Then the mean squared error is bounded as
\begin{align}
    \MSE[\widehat f]
    \lesssim\frac{1}{m_s^2}
    \sum_{k\in \RC}\Ebb\left[W_k^2\sum_i \left\{\frac{\abs{c_i^{(k)}}}{\Delta_k}\right\}\right]\,,
    \label{eq:app_comb_general_bound}
\end{align}
where $\{x\}$ denotes the non-integer part of $x$, and where the expectation value accounts for the fact that the coefficients arriving at layer $k$ are random, exactly as in Theorem~\ref{thm:one_norm_sampling_error_unbiased}. If $\abs{c_i^{(k)}}<\Delta_k$ for all $i$, so that no coefficient fills a whole tooth, then $\sum_i\left\{\abs{c_i^{(k)}}/\Delta_k\right\}=m_s$ and we recover the one-norm sampling bound of Theorem~\ref{thm:one_norm_sampling_error_unbiased},
\begin{align}
    \MSE[\widehat f]
    \lesssim\frac{1}{m_s}
    \sum_{k\in \RC}\Ebb\left[W_k^2\right]\,,
    \label{eq:app_comb_bound_small_coefficients}
\end{align}
whereas if $\abs{c_i^{(k)}}\gg\Delta_k$ for all $i$, so that every coefficient spans many teeth, then only the last tooth of each coefficient is uncertain, the number $\abs{\PC_k}$ of stored Pauli strings obeys $\abs{\PC_k}\ll m_s$, and $\sum_i\left\{\abs{c_i^{(k)}}/\Delta_k\right\}\leq\abs{\PC_k}\ll m_s$, so that
\begin{align}
    \MSE[\widehat f]
    \lesssim\frac{1}{m_s^2}
    \sum_{k\in \RC}\Ebb\left[W_k^2\:\abs{\PC_k}\right]\,,
    \label{eq:app_comb_bound_large_coefficients}
\end{align}
which is smaller than the one-norm bound of Eq.~\eqref{eq:app_comb_bound_small_coefficients} by the factor $\abs{\PC_k}/m_s\ll1$. Here $\PC_k$ denotes the set of Pauli strings stored at layer $k$. In the extreme case where the one norm is carried by $\order{1}$ coefficients, the bound becomes $\MSE[\widehat f]\lesssim m_s^{-2}\sum_{k\in\RC}\Ebb[W_k^2]$.
\end{result}

\begin{proof}
Let us start by recalling the comb estimator of Eq.~\eqref{eq:app_comb_coefficient}, which for a spacing $\Delta_k=W_k/m_s$ reads
\begin{align}
    \widehat c_i^{(k)}
    =\sign(c_i^{(k)})
      \frac{W_k}{m_s}N_i^{(k)}\,,
    \label{eq:app_comb_layer_estimator}
\end{align}
where $N_i^{(k)}$ is the number of comb teeth that fall in the interval corresponding to $c_i^{(k)}$. Contrary to the one-norm sampling, $N_i^{(k)}$ does not follow a multinomial law. Instead, Eq.~\eqref{eq:N_i_asbernoully} gives
\begin{align}
    N_i^{(k)}
    =\left\lfloor\frac{\abs{c_i^{(k)}}}{\Delta_k}\right\rfloor + I_i^{(k)}\,,
    \label{eq:app_comb_layer_counts}
\end{align}
where $I_i^{(k)}$ takes the values $0$ or $1$ with $\Ebb[I_i^{(k)}]=\left\{\abs{c_i^{(k)}}/\Delta_k\right\}$ by Eq.~\eqref{eq:expval_random_I_i_N_i}, and where the different $I_i^{(k)}$ are not independent of each other because they are all driven by the single random offset of the comb.

The first two moments of the counts follow. The integer part of Eq.~\eqref{eq:app_comb_layer_counts} is deterministic, so
\begin{align}
    \Ebb[N_i^{(k)}]
    &=\left\lfloor\frac{\abs{c_i^{(k)}}}{\Delta_k}\right\rfloor+\left\{\frac{\abs{c_i^{(k)}}}{\Delta_k}\right\}
    =\frac{\abs{c_i^{(k)}}}{\Delta_k}\,,
    \label{eq:app_comb_count_mean}\\
    \Var[N_i^{(k)}]
    &=\Var[I_i^{(k)}]
    =\left\{\frac{\abs{c_i^{(k)}}}{\Delta_k}\right\}\left(1-\left\{\frac{\abs{c_i^{(k)}}}{\Delta_k}\right\}\right)\,,
    \label{eq:app_comb_count_variance}
\end{align}
where in the first line we used $x=\lfloor x\rfloor+\{x\}$ and in the second one that $I_i^{(k)}$ is a Bernoulli variable of mean $\left\{\abs{c_i^{(k)}}/\Delta_k\right\}$. Recalling $\delta c_i^{(k)}=\widehat c_i^{(k)}-c_i^{(k)}$ from Lemma~\ref{lemma:general_mse}, Eqs.~(\ref{eq:app_comb_layer_estimator}-\ref{eq:app_comb_count_variance}) give
\begin{align}
    \Ebb\left[\delta c_i^{(k)}\right] &= 0\,,
    \label{eq:app_comb_local_bias}\\
    \Ebb\left[(\delta c_i^{(k)})^2\right]
    &=\frac{W_k^2}{m_s^2}\left\{\frac{\abs{c_i^{(k)}}}{\Delta_k}\right\}\left(1-\left\{\frac{\abs{c_i^{(k)}}}{\Delta_k}\right\}\right)\,,
    \label{eq:app_comb_local_variance}\\
    \Ebb\left[\delta c_i^{(k)}\delta c_j^{(k)}\right]
    &=\frac{W_k^2}{m_s^2}\sign(c_i^{(k)})\sign(c_j^{(k)})
      \Covar[N_i^{(k)},N_j^{(k)}]\,,\qquad i\neq j\,.
    \label{eq:app_comb_local_covariance}
\end{align}
This is where comb sampling departs from the one-norm sampling. There the covariance of the counts is known in closed form, Eq.~\eqref{eq:app_multinomial_covar}, whereas here the $I_i^{(k)}$ are correlated through the single offset of the comb and Eq.~\eqref{eq:app_comb_local_covariance} cannot be evaluated term by term. We therefore keep the covariances as a matrix.

Because Eq.~\eqref{eq:app_comb_local_bias} holds also when conditioned on the history of the propagation, the cross-layer cancellation of Eq.~\eqref{eq:onenorm_conditioning} applies directly and the expansion of Eq.~\eqref{eq:mse_nobias_expansion} carries over unchanged,
\begin{align}
    \MSE[\widehat f]
    =\sum_{k\in\RC}\Ebb\left[\left(\sum_i
      \delta c_i^{(k)}\Tr[\rho_{L-k}P_i]\right)^2\right]\,.
    \label{eq:app_comb_mse_expansion}
\end{align}

Recall that we could not compute the value of the covariances, thus we expand the square of a single layer with Eqs.~(\ref{eq:app_comb_local_variance},\ref{eq:app_comb_local_covariance}) and write it as a matrix vector product
\begin{align}
    \Ebb\left[\left(\sum_i
      \delta c_i^{(k)}\Tr[\rho_{L-k}P_i]\right)^2\right]
    =\frac{W_k^2}{m_s^2}\,x^\top\Lambda\, x\,,
    \label{eq:app_comb_quadratic_form}
\end{align}
where
\begin{align}
    x_i=\sign(c_i^{(k)})\Tr[\rho_{L-k}P_i]\,,
    \qquad
    \Lambda_{ij}=\Covar[N_i^{(k)},N_j^{(k)}]\,,
    \label{eq:app_comb_quadratic_form_def}
\end{align}
so that the diagonal of $\Lambda$ is $\Lambda_{ii}=\Var[N_i^{(k)}]$.

The matrix $\Lambda$ is a covariance matrix, hence positive semidefinite, and admits an eigendecomposition
\begin{align}
    \Lambda=\sum_{a=1}^{\abs{\PC_k}}\lambda_a u_au_a^\top\,,
    \label{eq:app_comb_eigendecomposition}
\end{align}
where $\PC_k$ denotes the set of Pauli strings in the sum at layer $k$, so that $\abs{\PC_k}$ is their number, and the $u_a$ are orthonormal. By the assumption of Result~\ref{result:comb_sampling_error_unbiased}, $x$ has no particular alignment with the eigenvectors of $\Lambda$, so that it can be approximated as
\begin{align}
    x\simeq\frac{\norm{x}}{\sqrt{\abs{\PC_k}}}\sum_{a=1}^{\abs{\PC_k}}u_a\,,
    \label{eq:app_comb_no_alignment}
\end{align}
where we could have attached a different sign to each term without changing what follows. Using Eq.~\eqref{eq:app_comb_no_alignment} twice in Eq.~\eqref{eq:app_comb_quadratic_form} we find
\begin{align}
    x^\top\Lambda\, x
    &\simeq\frac{\norm{x}}{\sqrt{\abs{\PC_k}}}\sum_{b=1}^{\abs{\PC_k}}u_b^\top
      \left(\frac{\norm{x}}{\sqrt{\abs{\PC_k}}}\sum_{a=1}^{\abs{\PC_k}}\lambda_au_a\right)
    =\frac{\norm{x}^2}{\abs{\PC_k}}\sum_{a=1}^{\abs{\PC_k}}\lambda_a\,,
    \label{eq:app_comb_quadratic_form_trace}\\
    \norm{x}^2
    &=\sum_i\sign(c_i^{(k)})^2\Tr[\rho_{L-k}P_i]^2
    =\sum_i\Tr[\rho_{L-k}P_i]^2\leq\abs{\PC_k}\,,
    \label{eq:app_comb_norm_x}
\end{align}
where in the first line we used that the $u_a$ are orthonormal eigenvectors of $\Lambda$, and in the second one that $|\Tr[\rho_{L-k}P_i]|\leq1$. Combining Eqs.~(\ref{eq:app_comb_quadratic_form_trace},\ref{eq:app_comb_norm_x}) with $\sum_a\lambda_a=\Tr[\Lambda]$ and Eq.~\eqref{eq:app_comb_count_variance},
\begin{align}
    \Ebb\left[\left(\sum_i
      \delta c_i^{(k)}\Tr[\rho_{L-k}P_i]\right)^2\right]
    \lesssim&\:\frac{W_k^2}{m_s^2}\Tr[\Lambda]
    =\frac{W_k^2}{m_s^2}\sum_i\left\{\frac{\abs{c_i^{(k)}}}{\Delta_k}\right\}\left(1-\left\{\frac{\abs{c_i^{(k)}}}{\Delta_k}\right\}\right) \nonumber\\
    \leq&\:\frac{W_k^2}{m_s^2}\sum_i\left\{\frac{\abs{c_i^{(k)}}}{\Delta_k}\right\}\,,
    \label{eq:app_comb_layer_bound}
\end{align}
where in the last inequality we used $\left\{\abs{c_i^{(k)}}/\Delta_k\right\}\leq1$. Inserting Eq.~\eqref{eq:app_comb_layer_bound} into Eq.~\eqref{eq:app_comb_mse_expansion}, and taking the outer expectation value over the coefficients arriving at layer $k$, gives Eq.~\eqref{eq:app_comb_general_bound}.

The two limits follow from the definition of the non-integer part. If $\abs{c_i^{(k)}}<\Delta_k$ for every $i$, then the integer part vanishes, so $\left\{\abs{c_i^{(k)}}/\Delta_k\right\}=\abs{c_i^{(k)}}/\Delta_k$ and
\begin{align}
    \sum_i\left\{\frac{\abs{c_i^{(k)}}}{\Delta_k}\right\}
    =\frac{1}{\Delta_k}\sum_i\abs{c_i^{(k)}}
    =\frac{W_k}{\Delta_k}=m_s\,,
    \label{eq:app_comb_remainder_small}
\end{align}
which turns Eq.~\eqref{eq:app_comb_general_bound} into Eq.~\eqref{eq:app_comb_bound_small_coefficients}. Every coefficient is then resolved by at most one tooth and the comb has no advantage over independent draws. If instead $\abs{c_i^{(k)}}\gg\Delta_k$ for every $i$, each coefficient is resolved by many teeth and only the last of them is uncertain, so each non-integer part is a subleading remainder and there are only $\abs{\PC_k}$ of them. Since $\sum_i\abs{c_i^{(k)}}=W_k=m_s\Delta_k$ while every $\abs{c_i^{(k)}}\gg\Delta_k$, we must have $\abs{\PC_k}\ll m_s$, and $\sum_i\left\{\abs{c_i^{(k)}}/\Delta_k\right\}\leq\abs{\PC_k}\ll m_s$, which gives Eq.~\eqref{eq:app_comb_bound_large_coefficients}. This is the origin of the better scaling of comb resampling at early depth.

\end{proof}

\subsubsection{Deterministic truncation followed by comb resampling}

We now analyze the error made with comb sampling when deterministic truncation is added. We obtain a similar bound as in Theorem \ref{thm:one_norm_sampling_error_bias}, but the terms corresponding to the unbiased sampling on modified coefficients are replaced by the error made with comb sampling as described in Result \ref{result:comb_sampling_error_unbiased}.

\label{app:comb_sampling_error_bias}
\begin{result}[MSE bound for the comb sampling with truncation]
\label{result:comb_sampling_error_bias}
    Consider $f = \Tr[\rho U^\dagger O U]$, such that $U = U_L \dots U_2 U_1$. Suppose that we propagate the observable forward using the comb sampling strategy preceded by truncation as in Appendix~\ref{app:one_norm_resampling_bias}. At every layer $k$ of a subset $\RC\subseteq\{1,\dots,L\}$ we use a bias threshold $b_k$ and a comb of spacing $\widetilde\Delta_k=\widetilde W_k/m_s$ over the retained coefficients. Then, under the same assumption as Result~\ref{result:comb_sampling_error_unbiased}, the mean squared error is bounded as
\begin{align}
    \MSE[\widehat f]
    \lesssim\Ebb\left[\left(\sum_{k\in\RC}\left(W_k-\widetilde W_k\right)\right)^{2}\right]
    +\frac{2}{m_s}\sum_{k > \ell \in\RC}\sqrt{\Var[\beta_k]\;\Ebb\left[\widetilde W_\ell^2\sum_i\left\{\frac{\abs{\widetilde c_i^{(\ell)}}}{\widetilde\Delta_\ell}\right\}\right]}
    +\frac{1}{m_s^2}\sum_{k\in\RC}\Ebb\left[\widetilde W_k^2\sum_i\left\{\frac{\abs{\widetilde c_i^{(k)}}}{\widetilde\Delta_k}\right\}\right]\,,
\label{eq:app_comb_truncated_general_bound}
\end{align}
where $W_k$ and $\widetilde W_k$ are the one norms without and with truncation as in Theorem~\ref{thm:one_norm_sampling_error_bias}, $\beta_k$ is the truncation error committed at layer $k$, $\beta_k=-\sum_{i,\,\abs{c_i^{(k)}}\leq b_k}
      c_i^{(k)}\Tr[\rho_{L-k}P_i]$, and $\left\{\abs{\widetilde c_i^{(k)}}/\widetilde\Delta_k\right\}$ denotes the same non-integer part built with the truncated spacing $\widetilde\Delta_k$, which follows the same convention as $\widetilde W_k$ and $\widetilde c_i^{(k)}$.
\end{result}

\begin{proof}

    We recall the definitions of $S_k$ and $\beta_k$ introduced in Eqs.~(\ref{eq:app_truncated_layer_noise_def},\ref{eq:app_truncated_layer_bias_def})

    \begin{align*}
    S_k &=\sum_{i,\,\abs{c_i^{(k)}}>b_k}
      \delta c_i^{(k)}\Tr[\rho_{L-k}P_i]\,,\\
    \beta_k &=-\sum_{i,\,\abs{c_i^{(k)}}\leq b_k}
      c_i^{(k)}\Tr[\rho_{L-k}P_i]\,.
    \end{align*}

    and the expression of the mean square error of $\hat f$ as in Eq.~\eqref{eq:app_truncated_mse_expansion}
    \begin{align}
    \MSE[\widehat f] =&\:\sum_{k\in\RC}\Ebb\left[\left(S_k+\beta_k\right)^2\right]
    +2\sum_{k>\ell\in\RC}\Ebb\left[      \left(S_k+\beta_k\right)\left(S_\ell+\beta_\ell\right)\right]\,,
    \label{eq:app_truncated_mse_expansion2}
\end{align}

We have, with the same justification as in Eq.~\eqref{eq:app_truncated_diagonal_bound} and ~\eqref{eq:app_truncated_cross_layer} that 
\begin{align}
    \Ebb\left[\left(S_k+\beta_k\right)^2\right] &= \Ebb[S_k^2] + \Ebb[\beta_k^2]\label{eq:app_truncated_diagonal_bound2}\\
    \Ebb\left[      \left(S_k+\beta_k\right)\left(S_\ell+\beta_\ell\right)\right]&= \Ebb[\beta_k(S_\ell + \beta_\ell)]\label{eq:app_truncated_cross_layer2}
\end{align}

We can transform the last equation in the same way as in Eq. ~\eqref{eq:app_truncated_cross_layer_bound} that gives us
\begin{align}
    \Ebb\left[      \left(S_k+\beta_k\right)\left(S_\ell+\beta_\ell\right)\right]&= \Ebb[\beta_k\beta_\ell] + \Ebb[\beta_k S_\ell]\\
    &\leq \Ebb[\beta_k\beta_\ell] + \sqrt{\Var(\beta_k)\Var(S_\ell)}
    \label{eq:app_truncated_cross_layer_bound2}
\end{align}

Then we can recover Eq.~\eqref{eq:app_truncated_mse_expansion2}, and the derived Eqs.~(\ref{eq:app_truncated_diagonal_bound2},\ref{eq:app_truncated_cross_layer_bound2}) to find the following bound
\begin{align}
    \MSE[\widehat f] &\leq \sum_{k\in\RC} \Ebb[\beta_k^2] +  \sum_{k\in\RC} \Ebb[S_k^2] + 2\sum_{k>\ell\in\RC} \Ebb[\beta_k\beta_\ell] + 2\sum_{k>\ell\in\RC} \sqrt{\Var(\beta_k)\Var(S_\ell)}\\
    &=\Ebb\left[\left(\sum_{k\in\RC}\beta_k\right)^2\right] + 2\sum_{k>\ell\in\RC} \sqrt{\Var(\beta_k)\Var(S_\ell)} + \sum_{k\in\RC} \Ebb[S_k^2]
    \label{eq:mse_reduced_comb_bias}
\end{align}
where in the equality we grouped the terms summing in $\beta_i$.

The comb resampling with bias can be described as comb resampling with modified coefficients. The modified coefficients, modified one-norm and modified probabilities at layer $k$ with threshold $b_k$ are those of Eqs.~(\ref{eq:app_truncated_coefficient_layer}-\ref{eq:app_truncated_probability_layer}), and the estimator of the coefficient $i$ at layer $k$ is now
\begin{align}
    \widehat c_i^{(k)}
    =\sign(c_i^{(k)})
      \frac{\widetilde W_k}{m_s}N_i^{(k)}\,,
    \label{eq:app_comb_truncated_estimator}
\end{align}
where $N_i^{(k)}$ counts the teeth of a comb of spacing $\widetilde\Delta_k$ laid over the cumulative weights of the retained coefficients only.
For a retained index, $\abs{c_i^{(k)}}>b_k$, the resampling step is exactly the comb sampling of Appendix~\ref{app:comb_sampling} applied to the modified Pauli sum $\sum_i\widetilde c_i^{(k)}P_i$. The computations of Result~\ref{result:comb_sampling_error_unbiased} therefore hold with $W_k$ replaced by $\widetilde W_k$ and $\Delta_k$ by $\widetilde\Delta_k$, so that Eqs.~(\ref{eq:app_comb_local_bias}-\ref{eq:app_comb_local_covariance}) become
\begin{align}
    \Ebb\left[\delta c_i^{(k)}\right] &= 0\,,
    \label{eq:app_comb_truncated_retained_mean}\\
    \Ebb\left[(\delta c_i^{(k)})^2\right]
    &=\frac{\widetilde W_k^2}{m_s^2}\left\{\frac{\abs{\widetilde c_i^{(k)}}}{\widetilde\Delta_k}\right\}
      \left(1-\left\{\frac{\abs{\widetilde c_i^{(k)}}}{\widetilde\Delta_k}\right\}\right)\,,
    \label{eq:app_comb_truncated_retained_variance}\\
    \Ebb\left[\delta c_i^{(k)}\delta c_j^{(k)}\right]
    &=\frac{\widetilde W_k^2}{m_s^2}\sign(c_i^{(k)})\sign(c_j^{(k)})
      \Covar[N_i^{(k)},N_j^{(k)}]\,,\qquad i\neq j\,.
    \label{eq:app_comb_truncated_retained_crossed}
\end{align}
For a discarded index, $\abs{c_i^{(k)}}\leq b_k$, Eq.~\eqref{eq:app_truncated_coefficient_layer} gives $\widetilde c_i^{(k)}=0$, so the interval that this index occupies in the cumulative weights has zero width and no tooth can fall in it. Equation~\eqref{eq:app_comb_truncated_estimator} then yields $\widehat c_i^{(k)}=0$, and therefore $\delta c_i^{(k)}=-c_i^{(k)}$ for every discarded index $i$, exactly as in the one-norm case.

The second moment of $S_k$ follows from Eqs.~(\ref{eq:app_comb_quadratic_form}-\ref{eq:app_comb_layer_bound}) with the sums restricted to the retained indices,
\begin{align}
    \Ebb\left[S_k^2\right]
    \lesssim\frac{\widetilde W_k^2}{m_s^2}
      \sum_i\left\{\frac{\abs{\widetilde c_i^{(k)}}}{\widetilde\Delta_k}\right\}\,.
    \label{eq:app_comb_truncated_noise_bound}
\end{align}
This is Eq.~\eqref{eq:app_comb_general_bound} with the full one-norm replaced by the retained one and the remainders recomputed after truncation, and it is the sense in which truncation reduces the sampling noise of the comb.

By plugging Eq.~\eqref{eq:app_comb_truncated_noise_bound}, and using that $\Ebb[S_k] = 0$ in Eq. ~\eqref{eq:mse_reduced_comb_bias}, we have
\begin{align}
    \MSE[\widehat f]
    \lesssim\Ebb\left[\left(\sum_{k\in\RC}\beta_k\right)^2\right]
    +\frac{2}{m_s}\sum_{k>\ell\in\RC}
      \sqrt{\Var[\beta_k]\;\Ebb\left[\widetilde W_\ell^2
      \sum_i\left\{\frac{\abs{\widetilde c_i^{(\ell)}}}{\widetilde\Delta_\ell}\right\}\right]}
    +\frac{1}{m_s^2}\sum_{k\in\RC}\Ebb\left[\widetilde W_k^2
      \sum_i\left\{\frac{\abs{\widetilde c_i^{(k)}}}{\widetilde\Delta_k}\right\}\right]\,,
    \label{eq:app_comb_truncated_mse_collected}
\end{align}
which concludes the proof.

\end{proof}

\subsubsection{Mean squared error model}
\label{app:comb_sampling_mse_model}

We apply the same phenomenological model as in Result \ref{res:phenomenological_truncated_mse} to the comb sampling estimator. We obtain the same result, and we cannot show a distinctive effect of comb sampling because it would require to model how the concentration of the coefficient distribution evolves, and there is not a clear way to do it.

\begin{result}[Model for the MSE of the comb sampling with truncation]
\label{result:app_comb_truncated_mse_model}
    Assume the same three ingredients as in Result~\ref{result:app_truncated_mse_model}, namely a relative threshold $b_k=\epsilon_{\mathrm{rel}}\max_j\abs{c_j^{(k)}}$, a discarded fraction $\varphi_k=\nu\epsilon_{\mathrm{rel}}\lambda^{k-1}$ against a one-norm growing by $\gamma>1$ per layer, a truncation error that varies little from run to run, and $\nu\epsilon_{\mathrm{rel}}\ll1$. Then, for $q=\gamma\lambda\neq1$,
\begin{align}
    \MSE[\widehat f]
    \lesssim
    \left(\nu\epsilon_{\mathrm{rel}}W_0\gamma\,
      \frac{q^{L}-1}{q-1}\right)^{2}
    +\frac{\left(1-\mu\epsilon_{\mathrm{rel}}\right)^2W_0^2\gamma^2}{\gamma^2-1}
      \frac{\gamma^{2L}}{m_s}\,,
    \label{eq:app_comb_truncated_mse_model_result}
\end{align}
where $q=\gamma\lambda$ and $\mu=\nu/(1-\lambda)$, with $\left(q^{L}-1\right)/\left(q-1\right)$ replaced by $L$ when $q=1$. This is exactly the model of Result~\ref{result:app_truncated_mse_model}: once the remainder factor is bounded by $m_s$, the comb bound collapses onto the one-norm bound, so this model cannot resolve the difference between the two schemes. The advantage of the comb is the one identified in the second regime of Result~\ref{result:comb_sampling_error_unbiased}, where the remainder factor is $\abs{\PC_k}\ll m_s$ rather than $m_s$.
\end{result}

\begin{proof}
    
The truncation part of Eq.~\eqref{eq:app_comb_truncated_general_bound} is identical to that of Theorem~\ref{thm:one_norm_sampling_error_bias}, so the assumptions of Appendix~\ref{app:one_norm_sampling_mse_model} apply to it word for word. The relative threshold of Eq.~\eqref{eq:app_relative_cutoff} and the geometric decay of Eq.~\eqref{eq:app_discarded_fraction_decay} give the retained norm of Eq.~\eqref{eq:app_retained_norm_model} and the discarded mass of Eq.~\eqref{eq:app_discarded_mass_model}, and Eq.~\eqref{eq:app_phenom_bias_deterministic} lets us drop the cross term. What is new is the extra factor $\sum_i\left\{\abs{\widetilde c_i^{(k)}}/\widetilde\Delta_k\right\}$ carried by the sampling noise, and for it no further assumption is needed. Since the non-integer part obeys $\{x\}\leq x$ for $x\geq0$,
\begin{align}
    \sum_i\left\{\frac{\abs{\widetilde c_i^{(k)}}}{\widetilde\Delta_k}\right\}
    \leq\frac{1}{\widetilde\Delta_k}\sum_i\abs{\widetilde c_i^{(k)}}
    =\frac{\widetilde W_k}{\widetilde\Delta_k}
    =m_s\,,
    \label{eq:app_comb_remainder_bound}
\end{align}
where in the last two steps we used the definition of the retained one-norm and $\widetilde\Delta_k=\widetilde W_k/m_s$.

Inserting Eqs.~(\ref{eq:app_retained_norm_model},\ref{eq:app_discarded_mass_model},\ref{eq:app_phenom_bias_deterministic},\ref{eq:app_comb_remainder_bound}) into Eq.~\eqref{eq:app_comb_truncated_general_bound} with $\RC=\{1,\ldots,L\}$, the sampling noise becomes
\begin{align}
    \frac{1}{m_s^2}\sum_{k=1}^{L}\widetilde W_k^2
      \sum_i\left\{\frac{\abs{\widetilde c_i^{(k)}}}{\widetilde\Delta_k}\right\}
    \leq\frac{1}{m_s}\sum_{k=1}^{L}\widetilde W_k^2
    \simeq\frac{\left(1-\mu\epsilon_{\mathrm{rel}}\right)^2W_0^2\gamma^2}{m_s\left(\gamma^2-1\right)}
      \left(\gamma^{2L}-1\right)\,,
    \label{eq:app_comb_noise_model}
\end{align}
which for large $L$ gives Eq.~\eqref{eq:app_comb_truncated_mse_model_result}. Equation~\eqref{eq:app_comb_remainder_bound} is saturated when every retained coefficient is smaller than one comb tooth, which is the first limit of Result~\ref{result:comb_sampling_error_unbiased}, and which concludes the proof.
\end{proof}

\subsection{Merging and the number of effective draws}
\label{app:error_merging_one_norm}
All the bounds derived so far are controlled by the number of raw draws $m_s$, whereas the memory only holds the number $m_p$ of distinct Pauli strings. The two differ because repeated draws are merged, and how much they differ depends on how peaked the coefficient distribution is. The following occupancy model relates them.

\begin{result}[Effective number of draws]
\label{result:app_effective_draws}
   Suppose that $m_s$ independent one-norm draws are taken from the coefficient distribution induced by the one-norm resampling, with $p_i = \abs{c_i}/W$ as in Eq.~\eqref{eq:app_one_norm_probability} and that the $m_p$ distinct outcomes are merged and stored. Then, up to a fixed relative deviation $\delta>0$ on the number of stored strings, the number of raw draws that $m_p$ represents obeys
\begin{align}
    m_s\propto m_p^{2}
    \qquad\text{with probability}\quad
    1-\order{\frac{1}{\delta^2m_p}}
    \label{eq:app_effective_draws_powerlaw}
\end{align}
for the power law $p_i=Ci^{-2}$, and
\begin{align}
    m_s\propto m_p
    \qquad\text{with probability}\quad
    1-\order{\frac{1}{\delta^2m_p}}
    \label{eq:app_effective_draws_uniform}
\end{align}
for the uniform distribution over $N$ Pauli strings with $m_p=cN$ and $0<c<1$ held fixed. A variance proportional to $1/m_s$ decays as $m_p^{-2}$ for the power law but only as $m_p^{-1}$ for the uniform distribution.
\end{result}
\begin{proof}
To relate $m_s$ and $m_p$, let $I_i(m_s)$ indicate whether outcome $i$ has appeared at least once in $m_s$ independent draws. Since the draws are independent, outcome $i$ is missed with probability $(1-p_i)^{m_s}$ and
\begin{align}
    \Ebb[I_i(m_s)]=1-(1-p_i)^{m_s}\,,
    \label{eq:app_occupancy_indicator}
\end{align}
so that the number of distinct Pauli strings, $U(m_s)=\sum_iI_i(m_s)$, has mean
\begin{align}
    \Ebb[U(m_s)]
    =\sum_i\left[1-(1-p_i)^{m_s}\right]\,.
    \label{eq:app_expected_unique}
\end{align}
Identifying $m_p$ with $\Ebb[U(m_s)]$ and inverting Eq.~\eqref{eq:app_expected_unique} estimates the number of raw draws represented by $m_p$ stored Pauli strings. This is an occupancy model for the scaling seen in the data, and it does not assert that every detail of merged propagation is equivalent to independent trajectories.

That identification is only useful if $U(m_s)$ does not fluctuate much from run to run, so let us bound its variance. Because $I_i(m_s)\in\{0,1\}$ we have $I_i(m_s)^2=I_i(m_s)$, and therefore
\begin{align}
    \Var[I_i]
    &=\Ebb[I_i(m_s)^2]-\Ebb[I_i(m_s)]^2
    =\Ebb[I_i(m_s)]\left(1-\Ebb[I_i(m_s)]\right)\nonumber\\
    &=(1-p_i)^{m_s}\left[1-(1-p_i)^{m_s}\right]\,.
    \label{eq:app_occupancy_indicator_variance}
\end{align}
The indicators are not independent, because a fixed number of draws makes the outcomes compete for them. For $i\neq j$,
\begin{align}
    \Covar[I_i,I_j]
    &=\Pr(I_i=1,I_j=1)-\Pr(I_i=1)\Pr(I_j=1)\nonumber\\
    &=\left[1-(1-p_i)^{m_s}-(1-p_j)^{m_s}+(1-p_i-p_j)^{m_s}\right]\nonumber\\
    &\quad-\left[1-(1-p_i)^{m_s}\right]\left[1-(1-p_j)^{m_s}\right]\nonumber\\
    &=(1-p_i-p_j)^{m_s}-(1-p_i)^{m_s}(1-p_j)^{m_s} \leq 0\,,
    \label{eq:app_occupancy_covariance}
\end{align}
where the second equality is inclusion--exclusion for $\Pr(I_i=1,I_j=1)$ and $\Pr(I_i=0,I_j=0)=(1-p_i-p_j)^{m_s}$. The result is non-positive because $(1-x-y)^k\leq [(1-x)(1-y)]^k$ for $x,y\geq0$, so that summing Eqs.~(\ref{eq:app_occupancy_indicator_variance},\ref{eq:app_occupancy_covariance}) gives
\begin{align}
    \Var[U(m_s)] = \sum_i \Var[I_i] + \sum_{i\neq j}\Covar[I_i,I_j]
    \leq\sum_i(1-p_i)^{m_s}\left[1-(1-p_i)^{m_s}\right]
    \leq\Ebb[U(m_s)]\,,
    \label{eq:app_occupancy_variance_bound}
\end{align}
where in the last inequality we used $(1-p_i)^{m_s}\leq1$ together with Eqs.~(\ref{eq:app_occupancy_indicator},\ref{eq:app_expected_unique}). Equation~\eqref{eq:app_occupancy_variance_bound} holds for any coefficient distribution and is all we need, since Chebyshev's inequality then bounds the probability of departing from Eq.~\eqref{eq:app_expected_unique} by a fraction $\delta$ as
\begin{align}
    \Pr\left(\abs{U(m_s)-\Ebb[U(m_s)]}>\delta\,\Ebb[U(m_s)]\right)
    \leq\frac{\Var[U(m_s)]}{\delta^2\,\Ebb[U(m_s)]^2}
    \leq\frac{1}{\delta^2\,m_p}\,.
    \label{eq:app_occupancy_concentration}
\end{align}
The number of stored Pauli strings therefore concentrates around Eq.~\eqref{eq:app_expected_unique}. Equation~\eqref{eq:app_occupancy_concentration} is expressed in terms of $m_p$, and substituting the relation between $m_p$ and $m_s$ turns it into a bound in terms of the number of raw draws, which we do for each of the two distributions of interest below.

\paragraph*{Peaked distributions.}
Take $p_i=Ci^{-\alpha}$ with $\alpha>1$. For small $p_i$ and large $m_s$ we have $(1-p_i)^{m_s}\simeq e^{-m_sp_i}$, so that Eq.~\eqref{eq:app_expected_unique} becomes
\begin{align}
    m_p\simeq\sum_{i=1}^{\infty}
    \left(1-e^{-Cm_s/i^\alpha}\right)\,.
    \label{eq:app_powerlaw_sum}
\end{align}
Approximating the sum by an integral and substituting $u=Cm_s/x^\alpha$ gives
\begin{align}
    m_p
    &\simeq\int_0^\infty
      \left(1-e^{-Cm_s/x^\alpha}\right)dx\nonumber\\
    &=\frac{(Cm_s)^{1/\alpha}}{\alpha}
      \int_0^\infty\left(1-e^{-u}\right)u^{-(1+1/\alpha)}du\nonumber\\
    &=(Cm_s)^{1/\alpha}
      \Gamma\!\left(1-\frac{1}{\alpha}\right)\,,
    \label{eq:app_powerlaw_integral}
\end{align}
so that inverting Eq.~\eqref{eq:app_powerlaw_integral} gives
\begin{align}
    m_s\simeq\frac{1}{C}
    \left[\frac{m_p}{\Gamma\!\left(1-1/\alpha\right)}\right]^{\alpha}
    =\Theta\!\left(m_p^{\alpha}\right)\,.
    \label{eq:app_powerlaw_inversion}
\end{align}
For $\alpha=2$ we have $\Gamma(1/2)=\sqrt{\pi}$, so that $m_p\simeq\sqrt{\pi Cm_s}$ and Eq.~\eqref{eq:app_powerlaw_inversion} reads $m_s\simeq m_p^2/(\pi C)$, which together with Eq.~\eqref{eq:app_occupancy_concentration} is Eq.~\eqref{eq:app_effective_draws_powerlaw}. A variance proportional to $1/m_s$ becomes proportional to $1/m_p^2$ and the estimator standard deviation falls as $1/m_p$. This is the mechanism consistent with the intermediate-depth curve in Fig.~\ref{fig:memory_utilization}. Nothing in this derivation is special to $\alpha=2$, and Eq.~\eqref{eq:app_powerlaw_inversion} together with Eq.~\eqref{eq:app_occupancy_concentration} gives a deviation probability $\order{1/(\delta^2m_s^{1/\alpha})}$ for any $\alpha>1$.

For this distribution the general bound of Eq.~\eqref{eq:app_occupancy_variance_bound} can be sharpened into an estimate. Using $(1-p_i)^{m_s}\simeq e^{-m_sp_i}$ in Eq.~\eqref{eq:app_occupancy_indicator_variance} and approximating the sum by an integral as above,
\begin{align}
    \Var[U(m_s)]
    &\simeq\sum_i\left(e^{-m_sp_i}-e^{-2m_sp_i}\right)
    \simeq\int_0^\infty\left(e^{-Cm_s/x^\alpha}-e^{-2Cm_s/x^\alpha}\right)dx\nonumber\\
    &=\Gamma\!\left(1-\frac{1}{\alpha}\right)
      \left(2^{1/\alpha}-1\right)(Cm_s)^{1/\alpha}
    =\left(2^{1/\alpha}-1\right)\Ebb[U(m_s)]\,,
    \label{eq:app_occupancy_variance_powerlaw}
\end{align}
where in the last equality we used Eq.~\eqref{eq:app_powerlaw_integral}. This is exact to leading order in the Poissonized occupancy model. For a fixed number of draws the covariance contribution of Eq.~\eqref{eq:app_occupancy_covariance} is subleading, scaling as $\order{m_s^{2/\alpha-1}}=o(m_s^{1/\alpha})$, and therefore does not alter the leading result.

\paragraph*{Uniform distributions.}
If the distribution is uniform over $N$ Pauli strings, then Eq.~\eqref{eq:app_expected_unique} gives
\begin{align}
    m_p=N\left[1-\left(1-\frac{1}{N}\right)^{m_s}\right]\,.
    \label{eq:app_uniform_occupancy}
\end{align}
Writing $m_p=cN$ with $0<c<1$ held fixed and using $\log(1-1/N)\simeq-1/N$ gives
\begin{align}
    m_s\simeq-N\log\left(1-c\right)
    =\frac{-\log\left(1-c\right)}{c}\,m_p\,,
    \label{eq:app_uniform_inversion}
\end{align}
which together with Eq.~\eqref{eq:app_occupancy_concentration} is Eq.~\eqref{eq:app_effective_draws_uniform}. The ordinary $1/\sqrt{m_p}$ standard-error scaling therefore returns once the coefficient distribution becomes flat. Note that Eq.~\eqref{eq:app_occupancy_concentration} applies here unchanged, since Eq.~\eqref{eq:app_occupancy_variance_bound} was derived without any assumption on the distribution, so the two regimes concentrate equally well and differ only in the exponent relating $m_s$ to $m_p$.

The same memory therefore encodes many more effective draws in the peaked regime than in the uniform one, which is why merging helps strongly at intermediate depth and much less at late depth.
\end{proof}

\section{Additional Monte Carlo tools}
\label{app:additional_tools}

\subsection{Probabilistic truncation}

In an ordinary breadth-first calculation, a deterministic coefficient threshold
discards every term below a chosen magnitude. A local unbiased alternative is
often called probabilistic truncation or Russian roulette. Let
\begin{equation}
    O=\sum_Pc_PP,
    \qquad c_P\in\mathbb C,
\end{equation}
and choose a stochastic threshold $b>0$. Coefficients above $b$ are kept
exactly. A smaller coefficient is either removed or promoted to magnitude $b$,
\begin{equation}
    \widehat c_P=
    \begin{cases}
        c_P,
        &\abs{c_P}\geq b,\\
        \sign(c_P)b,
        &\text{with probability }\abs{c_P}/b,\\
        0,
        &\text{with probability }1-\abs{c_P}/b,
    \end{cases}
    \label{eq:app_probabilistic_truncation}
\end{equation}
where the last two lines apply when $0<\abs{c_P}<b$ and
$\sign(c_P)=c_P/\abs{c_P}$. Zero coefficients remain zero.

The reweighting is essential. For a sub-threshold term,
\begin{equation}
    \Ebb[\widehat c_P]
    =\frac{\abs{c_P}}{b}
      \sign(c_P)b
    =c_P.
    \label{eq:app_probabilistic_truncation_mean}
\end{equation}
The coefficient-level variance is
\begin{equation}
    \Ebb\!\left[\abs{\widehat c_P-c_P}^2\right]
    =
    \begin{cases}
        0,&\abs{c_P}\geq b,\\
        \abs{c_P}(b-\abs{c_P}),&0<\abs{c_P}<b.
    \end{cases}
    \label{eq:app_probabilistic_truncation_variance}
\end{equation}
The one-norm is conserved in expectation, but the surviving small coefficients
are larger than before. The one-norm is consequently spread over fewer terms,
and the two-norm of a typical realization increases. Repeating the rule through
many layers therefore accumulates variance even though it does not accumulate
bias.

One can reduce this variance by adding a lower deterministic threshold or by
retaining the original coefficient rather than promoting it to $b$. Either
choice is biased. It can still be sensible if the resulting bias is below the
target precision, but the error will no longer disappear with more
repetitions.

The connection to semi-deterministic comb is close. Both rules leave a
head untouched and randomize a tail. Probabilistic truncation acts on each tail
coefficient independently and does not guarantee the number of stored Pauli
strings. Semi-deterministic comb sampling correlates the decisions globally so
that the tail stays close to a prescribed number of stored Pauli strings. When
every tail coefficient can
receive only zero or one comb tooth, the coefficient update is essentially the
same as Eq.~\eqref{eq:app_probabilistic_truncation}. The difference is the
control over the number of stored Pauli strings supplied by the comb.

The local rule remains attractive when a global cumulative sum is expensive or
when an occasional fluctuation in the number of stored Pauli strings is
harmless. When exceeding memory
would terminate the job, global resampling is the safer outer mechanism.

\subsection{Importance sampling}

One-norm sampling is agnostic about the final observable. It samples large transition
amplitudes because they are locally important, not because their descendants necessarily
have a large final overlap. If prior information identifies a subspace that is likely to
dominate the answer, an importance function can steer samples toward it without
introducing bias~\cite{robert2004monte}. A prime example is that of Pauli weight, where Ref.~\cite{angrisani2024classically}
has shown that expectation values in random quantum circuits are dominated by low Pauli
weight. Importance sampling remains unbiased while leading to strong variance reduction
if the prior is indeed correct. It can, however, also strongly increase variance if the
prior is misaligned.

A note on terminology: we will use wording such as ``biased sampling''. This does not mean
that the estimator becomes biased, but that we bias the transition probabilities and
appropriately re-weight the running coefficient to keep the estimator unbiased. This is an
unfortunate collision of terminology throughout the Monte Carlo literature.

Consider again an incoming Pauli $P_i$ with coefficient $c_i$ and a gate with transfer rule
\begin{equation}
    P_i \longmapsto \sum_j a_{ij} P_j .
    \label{eq:transfer-rule}
\end{equation}
Let $g(i,j) > 0$ be a biasing function. Define modified transition amplitudes and their
one-norm by
\begin{equation}
    \tilde{a}_{ij} = a_{ij} g(i,j), \qquad
    \tilde{A}_i = \sum_j |\tilde{a}_{ij}| .
    \label{eq:biased-amplitudes}
\end{equation}
The outgoing branch is drawn from
\begin{equation}
    q(j \mid i) = \frac{|\tilde{a}_{ij}|}{\tilde{A}_i} .
    \label{eq:biased-distribution}
\end{equation}
To compensate for the distorted probability, an incoming coefficient $c_i$ is updated to
\begin{equation}
    c'_j = c_i \sign(a_{ij}) \frac{\tilde{A}_i}{g(i,j)} .
    \label{eq:biased-update}
\end{equation}
The cancellation between $g$ and $1/g$ verifies unbiasedness directly,
\begin{align}
    \Ebb_q\!\left[c'_j P_j\right]
    &= \sum_j \frac{|a_{ij}| g(i,j)}{\tilde{A}_i} \,
       c_i \sign(a_{ij}) \frac{\tilde{A}_i}{g(i,j)} \, P_j \nonumber \\
    &= c_i \sum_j a_{ij} P_j ,
    \label{eq:biased-unbiased}
\end{align}
which holds for \textit{all} strictly positive choices of $g(i,j)$. Because the
re-weighting is local to each gate, unbiasedness persists throughout the circuit and
yields an unbiased estimator of the final Pauli sum.

\paragraph*{One-norm and uniform sampling.}
The trivial choice $g(i,j) = 1$ gives $\tilde{a}_{ij} = a_{ij}$ and
$\tilde{A}_i = A_i = \sum_j |a_{ij}|$, so that Eq.~\eqref{eq:biased-distribution} reduces
to plain one-norm sampling and Eq.~\eqref{eq:biased-update} to
$c'_j = c_i \sign(a_{ij}) A_i$. Conversely, the choice
$g(i,j) = 1/|a_{ij}|$ gives $\tilde{a}_{ij} = \sign(a_{ij})$ and
$\tilde{A}_i = n_i$, where $n_i$ is the number of branches, i.e., the number of $P_j$ with
$a_{ij} \neq 0$. The sampler then draws uniformly, $q(j \mid i) = 1/n_i$, and the update
becomes $c'_j = n_i \, a_{ij} \, c_i$. One-norm sampling is the variance-optimal choice of
biasing function~\cite{robert2004monte}. Importantly, this is a \textit{local} optimality
statement for one gate. It neither implies that this is the variance-optimal estimator for
the coefficient of a Pauli string after two gates or the entire circuit, nor does it imply
that it is the variance-optimal estimator for expectation values of the form
$\Tr[\rho P_j]$ even after a single gate.

\paragraph*{Biasing by Pauli weight.}
A relevant example is importance sampling by the \textit{Pauli weight} $|P|$, i.e., the
number of non-identity Paulis in a Pauli string. Low-weight Pauli strings typically exhibit
higher overlaps with generic states, so it may reduce variance to guide the sampler away
from commonly populating high-weight Pauli strings. Let
$\delta_{ij} = |P_j| - |P_i|$ be the change in Pauli weight incurred by the gate; we avoid $\Delta$ here because it denotes the comb spacing throughout Appendix~\ref{app:resampling_comparison}. We can
then penalize transitions that raise the weight,
\begin{equation}
    g(i,j) = \exp\!\left(-\gamma \max\!\left(\delta_{ij},\, 0\right)\right) ,
    \label{eq:weight-biasing}
\end{equation}
where $\gamma \geq 0$ is a tunable hyperparameter controlling the strength of the sampling
bias, and $\gamma = 0$ recovers one-norm sampling. We decorate the symbol because $\gamma$ denotes the growth rate of the one norm throughout Appendix~\ref{app:error_accumulation}. Because the biasing function enters
Eq.~\eqref{eq:biased-amplitudes} multiplicatively, priors can be combined at will. For
instance, biasing a uniform sampler by Pauli weight corresponds to
\begin{equation}
    g(i,j) = \frac{1}{|a_{ij}|}
        \exp\!\left(-\gamma\: \max\!\left(\delta_{ij},\, 0\right)\right) .
\end{equation}

\begin{figure}
    \centering
    \includegraphics[width=0.75\linewidth]{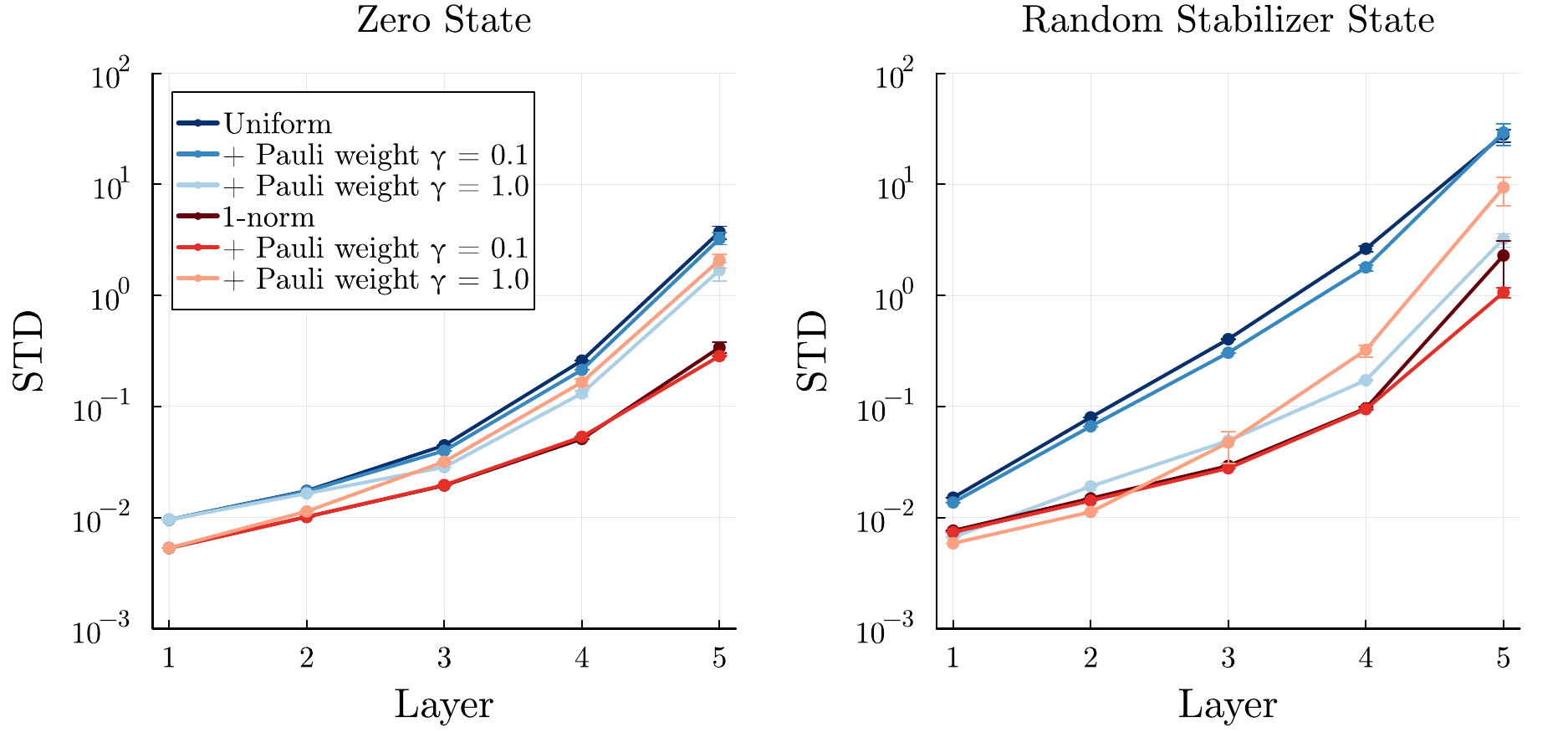}
    \caption{\textbf{Importance sampling the path sampling approach by Pauli weight.} \textbf{Left:} Importance sampling by Pauli weight for a $6\times5$ TFIM Trotter circuit with the all-zero initial state. \textbf{Right:} The same circuit and observable but with a random stabilizer initial state. We see that importance sampling can be beneficial, but it can quickly become detrimental.
    }
    \label{fig:importance_sampling}
\end{figure}

We numerically demonstrate the effect of importance sampling for the $6\times5$ TFIM Trotter circuit and the single Pauli-Z observable used throughout this work. Fig.~\ref{fig:importance_sampling} compares uniform sampling and 1-norm sampling, and each of those additionally biased by Pauli weight for various $\gamma$, where every estimator averages over $R = 10^4$ paths. We use two different initial states, the all-zero state as before, and a random stabilizer state. The latter is chosen to be closer to the assumptions that yield theoretical guarantees for Pauli weight truncation~\cite{angrisani2024classically}. 

It becomes evident that biasing the sampling against high Pauli weight can significantly reduce the STD of the estimator compared to uniform sampling. This is a particularly stark effect for the random stabilizer initial state. At the same time, sampling by 1-norm of the transition amplitudes is again significantly better, and Pauli weight biasing on top of 1-norm sampling yields no significant improvements, even for the random initial state.

\paragraph*{Weight degeneracy.}
Large $\gamma$ makes paths that visit high Pauli weight very rare. If one of those
paths later returns to low weight or otherwise contributes to the observable,
its $1/g$ compensation can be enormous. A finite run may see none of these
events and report a deceptively small empirical variance. The estimator can
look biased until enough rare, high-weight samples appear, even though its exact
expectation is unbiased. This is the familiar weight-degeneracy problem in
sequential Monte Carlo~\cite{robert2004monte}. Importance strength should
therefore be increased gradually and checked against $\gamma=0$ and against
multiple independent batches. The same idea can be applied while resampling
the Pauli sum, but the danger of rare compensating weights remains.

\subsection{Nonlinear observables and independent replica}
\label{app:nonlinear_observables}

The one-norm estimators are constructed for quantities that depend linearly on
the propagated coefficients. Many physical quantities are nonlinear. Examples
include purities, fluctuations, squared expectation values, and
out-of-time-ordered correlators.
This section details the strategies mentioned in the main text.

\subsubsection{Squaring an estimator}

Let $X$ be a single-pass estimator of a scalar $\mu$ with
\begin{equation}
    \Ebb[X]=\mu,
    \qquad
    \Var[X]=\sigma^2.
    \label{eq:app_single_estimator_moments}
\end{equation}
Squaring the estimator does not give an unbiased estimate of $\mu^2$,
\begin{equation}
    \Ebb[X^2]=\mu^2+\sigma^2.
    \label{eq:app_squaring_bias}
\end{equation}
The extra term is positive and is exactly the estimator variance. Averaging
first reduces, but does not remove, the bias. If
\begin{equation}
    \overline X=\frac{1}{R}\sum_{r=1}^{R}X_r,
\end{equation}
then
\begin{align}
    \Ebb\left[\overline X^2\right]
    &=\mu^2+\frac{\sigma^2}{R}
    \label{eq:app_average_then_square}\\
    \Var\left[\overline X\right] &= \Ebb\left[\overline X^2\right] - \Ebb\left[\overline X\right]^2= \frac{\sigma^2}{R}
\end{align}
This squaring approach can have relatively low variance, but its systematic
upward shift is especially problematic when many squared observables are added.
Every term contributes a non-negative bias.

\subsubsection{Independent replica}

The dual-replica construction, familiar from projector Monte Carlo estimates of
reduced density matrices~\cite{overy2014unbiased}, replaces the square by a
product of independent estimates. Run two statistically independent groups of
propagations and form
\begin{equation}
    Q=\overline X^{(1)}\overline X^{(2)}.
    \label{eq:app_dual_replica_product}
\end{equation}
Independence gives
\begin{equation}
    \Ebb[Q]
    =\Ebb[\overline X^{(1)}]
     \Ebb[\overline X^{(2)}]
    =\mu^2.
    \label{eq:app_dual_replica_unbiased}
\end{equation}
If each replica mean contains $R$ passes and each single pass has variance
$\sigma^2$, then
\begin{align}
    \Var[Q]&= \Ebb\left[\left(\overline X^{(1)}\right)^2\left(\overline X^{(2)}\right)^2\right] - \Ebb[Q]^2 = \Ebb\left[\left(\overline X^{(1)}\right)^2\right]\Ebb\left[\left(\overline X^{(2)}\right)^2\right]- \Ebb[Q]^2\\
    &= \left(\Var\left[\overline X^{(1)}\right] +\Ebb\left[\overline X^{(1)}\right]^2\right)\left(\Var\left[\overline X^{(2)}\right] +\Ebb\left[\overline X^{(2)}\right]^2\right)-\Ebb[Q]^2\\
    &=\left(\frac{\sigma^2}{R} + \mu^2\right)^2 - \mu^4\\
    &=\frac{\sigma^4}{R^2}
     +\frac{2\mu^2\sigma^2}{R}.
    \label{eq:app_dual_replica_variance}
\end{align}
In the noise-dominated regime, where
$\mu^2\ll\sigma^2/R$, the standard deviation of $Q$ is approximately
$\sigma^2/R$. This has the same $1/R$ scaling as the bias in
Eq.~\eqref{eq:app_average_then_square}. In the signal-dominated regime, the
second term controls the error and the usual $1/\sqrt R$ behaviour returns,
\begin{equation}
    \STD[Q]
    \simeq\frac{\sqrt{2}\abs{\mu}\sigma}{\sqrt R}.
    \label{eq:app_replica_signal_regime}
\end{equation}

The replica estimator is unbiased but generally noisier than squaring one sample
mean. Its advantage is clearest for a sum of many squares, of which the $511$ Pauli expectation values entering the collision entropy are an example. The individual squaring
biases add with the same sign, whereas centred replica errors can partially
cancel. Higher powers are obtained in the same way by multiplying the
corresponding number of mutually independent replica. Independence must include
the random seed, the resampling offsets, and every randomized Pauli sum used in
the propagation.

Replica independence removes only the nonlinear sampling bias. If deterministic
truncation gives $\Ebb[X]=\mu+\beta$, then the product estimates
$(\mu+\beta)^2$, not $\mu^2$, where $\beta$ is the truncation bias. Convergence in the truncation threshold is therefore
still required.

\end{document}